\documentclass[12pt]{article}

\usepackage[utf8]{inputenc}
 \usepackage{textgreek} 
\usepackage[T1]{fontenc}
\usepackage[margin=1in]{geometry}
\usepackage{amsmath,amssymb,amsfonts,amsthm}
\usepackage{algorithm}
\usepackage{algpseudocode}
\usepackage{graphicx}
\usepackage{xcolor}
\usepackage{url}
\usepackage{bbm}
\usepackage{natbib}
\usepackage{subfigure}
\usepackage{booktabs}
\usepackage{authblk}
\usepackage{setspace}
\usepackage{appendix}
\usepackage{hyperref}
\usepackage{comment}

\usepackage{threeparttable} 
\usepackage{enumitem}

\graphicspath{{Fig_pdf/}}

\newtheorem{theorem}{Theorem}
\newtheorem{proposition}[theorem]{Proposition}

\newtheorem{lemma}{Lemma}

\newtheorem{definition}{Definition}
\newtheorem{example}{Example}
\newtheorem{remark}{Remark}

\DeclareMathOperator*{\argmin}{arg\,min}

\newcommand{\Sp}[1]{Sp}

\usepackage{multirow}

\newcommand{\D}{\mathcal{D}}
\newcommand{\R}{\mathbb{R}}

\usepackage{titlesec} 

\usepackage{algorithm}
\usepackage{algpseudocode}
\algrenewcommand\algorithmicrequire{\textbf{Input:}}
\algrenewcommand\algorithmicensure{\textbf{Output:}}

\newcommand\Tau{\mathcal{T}}

\newcommand{\ttheta}{\pmb{\theta}}

\title{DUST: A Duality-Based Pruning Method\\ For Exact Multiple Change-Point Detection}

\author[1]{Vincent Runge\thanks{Corresponding author: vincent.runge@univ-evry.fr}}
\author[2]{Charles Truong}
\author[3,4]{Simon Querné}

\affil[1]{Université Paris-Saclay, CNRS, Univ Evry, Laboratoire de Mathématiques et Modélisation d'Evry, 91037, Evry-Courcouronnes, France.}
\affil[2]{Centre Borelli, Universit\'e Paris-Saclay, CNRS, ENS Paris-Saclay, 4 avenue des Sciences,  91190, Gif-sur-Yvette, France}
\affil[3]{Laboratoire de mathématiques de Versailles, Université Paris-Saclay, UVSQ, CNRS, 45 avenue des États-Unis,78000, Versailles, France}
\affil[4]{IFPEN, 1-4 Av. du Bois Pr\'eau,  92852, Rueil-Malmaison, France}
\date{}

\begin{document}

\maketitle

\abstract{We tackle the challenge of detecting multiple change points in large time series by optimising a penalised likelihood derived from exponential family models. Dynamic programming algorithms can solve this task exactly with at most quadratic time complexity. In recent years, the development of pruning strategies has drastically improved their computational efficiency. However, the two existing approaches have notable limitations: PELT struggles with pruning efficiency in sparse-change scenarios, while FPOP's structure is not adapted to multi-parametric settings. To address these issues, we introduce the DUal Simple Test (DUST) framework, which prunes candidate changes by evaluating a dual function against a threshold. This approach is highly flexible and broadly applicable to parametric models of any dimension. Under mild assumptions, we establish strong duality for the underlying non-convex pruning problem.
We demonstrate DUST’s effectiveness across various change-point regimes and models. In particular, for one-parametric models, DUST matches the simplicity of PELT with the efficiency of FPOP. Its use is especially advantageous for non-Gaussian models.
Finally, we apply DUST to mouse monitoring time series under a change-in-variance model, illustrating its ability to recover the optimal change-point structure efficiently.}


\section{Introduction}
\label{sec:intro}

Single and multiple change-point detection are well-established unsupervised machine learning tasks within the field of time-series analysis, with foundational work dating back to the 1950s \cite{page1954continuous, page1955test}. Over the past decades, the topic has been the subject of extensive research, resulting in numerous monographs \cite{series1994change, brodsky2013nonparametric, csorgo1997limit, chen2000parametric} and comprehensive review articles \cite{jandhyala2013inference, aminikhanghahi2017survey, truong2020selective}. Until recently, the long-standing focus of the scientific community has mainly been the statistical modelling and the calibration challenge for detecting and localising change points. With the rise of big data, the demand for computationally efficient algorithms has become increasingly pressing. Time efficiency is particularly crucial in many application domains, including genomics \cite{jia2022efficient, liehrmann2023diffsegr}, econometrics \cite{bai2003computation, gu2013fast}, climatology \cite{reeves2007review, shi2022changepoint}, speech processing \cite{harchaoui2009regularized, chang2010fast}, and network analysis \cite{zhou2022asymptotic, banerjee2020change}, to name just a few.

This work addresses the computational challenge of recovering multiple change points in a time series of fixed length. We consider change-point problems based on optimizing a penalised likelihood whose penalty is proportional to the number of changes. Although algorithms with quadratic time complexity have been available for some time \cite{auger1989algorithms,jackson2005algorithm}, the central objective is to approach quasi-linear execution time as closely as possible, while preserving the exact resolution of the underlying optimisation problem.

Detecting multiple change points requires carefully designed algorithmic strategies often balancing exactness with computational efficiency. Although the quasi-linear Binary Segmentation (BS) algorithm \cite{scott1974cluster, sen1975tests, zbMATH03766903} has historically been the dominant (approximate) method, recent advances in (exact) sub-quadratic dynamic programming (DP) techniques are beginning to challenge its prominence. We illustrate this with a simple example. With a time budget of 10 seconds, we estimate the maximum data length that an algorithm can segment when the signal consists of $10$ segments of equal length. The observations follow $y_t \sim \mathcal{N}(\mu_t, 1)$, with piecewise constant means $\mu_t \in \{0, 1\}$. BS stopping at $10$ segments can process up to approximately $n = 75 \times 10^6$ data points. In contrast, the classical exact quadratic-time Optimal Partitioning (OP) algorithm \cite{jackson2005algorithm} with BIC penalty \cite{yao1988estimating} is limited to around $n = 120 \times 10^3$. However, improved exact algorithms such as FPOP can handle up to $n = 30 \times 10^6$ data points within the same time budget, while the DUST algorithm can analyse around $n = 42 \times 10^6$ data points\footnote{Simulations were conducted on a MacBook Pro equipped with an Apple M1 chip (8-core CPU: 4 performance cores and four efficiency cores), 16GB of unified memory, running macOS Sequoia $15.5$. The code was implemented in R and executed using our dust Rcpp package (for DUST and OP) and the fpop packages provided in~\cite{maidstone2017optimal} (for FPOP and BS).}.

DP methods have experienced a period of renaissance with the development of accelerating pruning strategies, making their execution time competitive with BS on simulations (as just illustrated) and on real data sets (see Figure 6 in \cite{maidstone2017optimal}). Among them, inequality-based pruning (PELT) \cite{killick2012optimal} and functional-based pruning (FPOP) \cite{maidstone2017optimal} are the two extreme pruning strategies available on the ``pruning scale''. In the latest developments, DP with functional pruning has made possible the inference of complex structured models that constrain the successive segment parameter values (throughout a graph of constraints) \cite{Hocking2020, JSSv106i06}. One-parametric model with exponential decays \cite{jewell2020fast} or data modelled with auto-correlation and random drift \cite{romano2022detecting} can also be considered, among others.

PELT pruning is efficient for time series with many changes and is simple to code, but does not prune well when there are only a few changes. FPOP pruning is well-suited for problems with a one-parametric functional description. However, extensions to multivariate problems are uneasy as the parameter space description made of intervals in a one-parametric cost function is no longer possible. In this case, the partition of the parameter space is made of potential non-convex and unconnected elements \cite{runge2020finite}. Approximations of the parameter space with sphere-like or rectangle-like sets have been tested and significantly improve the time complexity, but only for small dimensions in the independent Gaussian multivariate setting \cite{pishchagina2024geometricbased}.\\

This work proposes a new pruning method called DUST that goes beyond these limitations. To do so, we recast the pruning task as a constrained optimisation problem. Within this framework, we explicitly derive the dual formulation for multivariate data drawn from the exponential family. We analyse its properties and, in particular, establish a strong duality result that underpins the effectiveness of our pruning rule.
The proposed DUST rule consists of evaluating a dual function and comparing it to a threshold -- yielding a test as simple as that used in PELT. For one-parameter cost functions (corresponding essentially to univariate data), we further derive a simple inequality-based rule by maximising the dual function explicitly. The PELT criterion corresponds to evaluating the dual at zero and is provably upper-bounded by the DUST value.
Normalising the dual into a decision function simplifies its interpretation and opens the door to more sophisticated DUST rules applicable to multivariate time series. A comprehensive simulation study confirms the efficiency and versatility of DUST, highlighting in particular its time robustness across all change-point regimes and under misspecified calibration parameters. We apply DUST for change in variance to large-scale data collected from mouse monitoring via a force platform aimed at quantifying muscle fatigue.\\

This paper has the following structure. In Section \ref{sec:functional} we present the functional problem for pruning in the context of multiple change-point detection. In Section \ref{sec:simple} we describe DUST in the one-constraint case and illustrate DUST on a first simple example. Astonishingly, the one-parametric case leads to a simple closed formula for the positivity test associated with the dual. The general presentation of the duality method is given in Section \ref{sec:duality}. The simulation study of Section~\ref{sec:simus} explores the efficiency of DUST. We eventually describe an example on a real-world data set in Section~\ref{sec:application}. If not in the main body of the article, proofs are given in the appendix.

\section{Change-point problem and its functional description}
\label{sec:functional}

\subsection{Model and optimisation problem}

We consider a time series consisting of $n$ data points, denoted by $y_1, y_2, \ldots, y_n$, where each point is drawn independently from a distribution belonging to the exponential family. A segment, denoted by $y_{ab}$, refers to a contiguous sub-sequence of the data, $y_{a+1}, \ldots, y_b$, where $0 \leq a < b \leq n$. In its canonical form with a minimal representation, the likelihood of a segment $y_{ab}$ can be expressed as:
\begin{align*}
f(y_{ab}; \ttheta)&=  \prod_{i=a+1}^b \biggr[ h(y_i)\exp\Big(\sum_{j=1}^d \theta_j \cdot T_j(y_i) - A(\ttheta) \Big) \biggr]\nonumber\,,\\
 &= \prod_{i=a+1}^b \biggr[h(y_i)\biggr]\exp\Big( \ttheta\cdot  \sum_{i=a+1}^b \mathbf{T}(y_i) - (b-a)A(\ttheta) \Big)\,,
\end{align*}
where $\ttheta = (\theta_1, \ldots, \theta_d)^T$ is the natural parameter, belonging to a convex domain $\Theta \subset \mathbb{R}^d$. The function $A$ is the log-partition function, which is strictly convex due to the minimal representation. The functions $T_1, \ldots, T_d$ are the sufficient statistics, aggregated in the vector $\mathbf{T} = (T_1, \ldots, T_d)^T$, and $h$ is the base measure (a normalising term). The dot sign denotes the scalar product. 

To define a parametric cost function, we take the negative log-likelihood of a segment and omit the data-dependent term $\prod h(y_i)$, which is constant with respect to the segmentation structure (as seen in Equation~\eqref{eq:global_cost}). For a segment $y_{ab}$, we define the sufficient statistic sum
 $\mathbf{S}_{ab} = \sum_{i=a+1}^b \mathbf{T}(y_i) \in\R^d$, and get the cost function:
\begin{equation}\label{eq:cost}
c(y_{ab}; \ttheta) =  (b-a)A(\ttheta) - \ttheta\cdot \mathbf{S}_{ab}\,.
\end{equation}
An important property of this cost function is its additivity over disjoint segments. Specifically, for any $a < b < t$, we have:
\begin{equation}\label{eq:additive}
c(y_{at}; \ttheta) - c(y_{bt}; \ttheta) = c(y_{ab}; \ttheta)\,.
\end{equation}

We present some well-known examples of exponential family models and their corresponding segment cost functions.

\begin{example}
Let $S_{ab} = \sum_{i=a+1}^b y_i$ denote the sum of univariate data points over the segment $y_{ab}$. With a Poisson model, its one-parametric cost function is given by $c(y_{ab}; \theta) =  (b-a)\exp(\theta) - S_{ab}\theta$; for exponential model, $c(y_{ab}; \theta) =  (b-a)(-\log(-\theta)) - S_{ab}\theta$; for a binomial distribution, $c(y_{ab}; \theta) =  (b-a)\log(1 +\exp(\theta)) - S_{ab}\theta$; while for a Gaussian distribution with unitary variance,  $c(y_{ab}; \theta) =  (b-a)\frac{\theta^2}{2} - S_{ab}\theta$.
\end{example}

\begin{example}
\label{ex:MeanAndVar}
For the Gaussian distribution with unknown mean and variance, the canonical form leads to a bi-parametric cost function:
\begin{equation}
\label{eq:MeanAndVar}
c(y_{ab}; [\theta_1; \theta_2]) =  (b-a)\Big(-\frac{\theta_1^2}{4\theta_2} + \frac{1}{2}\log\Big(-\frac{1}{2\theta_2} \Big)\Big) -  \theta_1 (\sum_{i=a+1}^{b}y_i) - \theta_2 (\sum_{i=a+1}^{b}y_i^2)\,, 
\end{equation}
which is quadratic in parameter $\theta_1$ only and $(\theta_1, \theta_2) \in \R \times \R^-$.
\end{example}

In this work, we address the problem of multiple change-point detection via penalised maximum likelihood, where a penalty proportional to the number of segments is introduced to control model complexity. Specifically, each segment is assigned a positive penalty value $\beta$, often referred to as the unitary penalty. A larger $\beta$ encourages sparser segmentations, i.e., fewer change points. The optimal penalised cost over the entire time series is defined as:
\begin{equation}\label{eq:global_cost}
Q_n = \min_{\tau \in \Tau} \left( \sum_{k=0}^{K} \Big[ \min_{\ttheta\in \Theta}c(y_{\tau_k\tau_{k+1}}; \ttheta) + \beta \Big] \right)\,,
\end{equation}
where $\tau = (\tau_0 = 0, \tau_1, \ldots, \tau_K, \tau_{K+1} = n)$ is a change-point vector, and $\Tau$ denotes the set of all admissible segmentations $\Tau = \{\tau \in \mathbb{N}^{K+2}, K \in \mathbb{N}, 0 = \tau_0 <\tau_1 < \ldots < \tau_{K} < \tau_{K+1} = n \}$. 
Here, $K$ -- the number of change points -- is not fixed in advance but is instead inferred in the optimisation, depending on the choice of penalty value $\beta$. The quantity $Q_n$ is often called the global cost.

This framework, based on the exponential family with independent segments and a linear penalty, is widely used in the literature (e.g., \cite{jackson2005algorithm, killick2012optimal}). Non-linear penalty terms have also been investigated in various settings (see \cite{Zhang2007, cleynen2017model, Verzelen2020}). Although our current work focuses on a specific setting, we believe that the proposed pruning strategy is highly versatile and has the potential to be extended to a broader class of models. These include, for example, online change-point detection \cite{pishchagina2023online}, models with dependencies across segments \cite{fearnhead2019detecting, romano2022detecting, JSSv106i06}, and non-linear penalisation schemes \cite{rigaill2015pruned}. We leave the development of these extensions for future work.

\subsection{Functional cost for pruning}
\label{subsec:pruning}

The exact solution to the optimisation problem \eqref{eq:global_cost} can be obtained using the optimal partitioning algorithm \cite{jackson2005algorithm}. In our case of the exponential family, the term $\min_{\ttheta\in \Theta}c(y_{st}; \ttheta)$ can be computed in constant time. As a result, the global cost $Q_t$ can be evaluated in quadratic time using the following recursion over the last segment position using the initial value $Q_0 = 0$:
\begin{equation}
\label{eq:OP}
Q_t = \min_{0 \le s < t} \left\{Q_s + \min_{\ttheta \in \Theta} c(y_{st}; \ttheta) + \beta \right\}\,.
\end{equation}

This dynamic programming approach has been significantly accelerated through pruning strategies, most notably PELT \cite{killick2012optimal} and, more recently, FPOP \cite{maidstone2017optimal}. Pruning aims to reduce the number of candidate indices $s$ considered in the minimisation of~\eqref{eq:OP}, by applying conditions that ensure no optimal solution is discarded -- thus preserving exactness. In the FPOP framework, the optimisation over the natural parameter $\ttheta$ in recursion~\eqref{eq:OP} is postponed. That is, we search for the best last change-point location in truncated data $y_{0t}$ for $t$ from $1$ to $n$ for each value of the parameter~$\ttheta$:
\begin{equation}\label{eq:fpop}
Q_t(\ttheta) = \min_{0 \le s < t} \Big( Q_s + c(y_{st}; \ttheta) + \beta\Big)\,.
\end{equation}
Here, the global cost $Q_t$ has been transformed into a functional cost by treating the parameter 
$\ttheta$ as a free variable. We recover its value by $Q_t = \min_{\ttheta \in \Theta}Q_t(\ttheta)$. Defining inner cost function $q_t^s$ as:
\begin{equation}\label{eq:qts}
q_{t}^{s}(\ttheta) =  Q_s + c(y_{st}; \ttheta) + \beta\,,
\end{equation}
we can write $Q_t(\ttheta) = \min_{s\in \Tau_t}\{q_t^s(\ttheta)\}$ with $\Tau_t \subset \{0,\ldots,t-1\}$ being the set of non-pruned indices. An effective pruning rule is both fast to test and capable of eliminating a substantial number of indices from 
$\{0,\ldots,t-1\}$. For example, the simple pruning criterion used in PELT yields a set 
$\Tau_t$ of bounded size, even for large time series, under the assumption that the number of change points grows proportionally with the data length (see Theorem 3.2 in \cite{killick2012optimal}). A more efficient pruning strategy, FPOP, relies on the following principle.

\begin{definition}\label{pr:pruning_principle} (Functional pruning principle)
Considering an index $s$ in $\{0,\ldots,t-1\}$, if for all $\ttheta$ in $\Theta$ there exists $r$ (depending on $\ttheta$) such that $q_t^{r}(\ttheta) <  q_t^{s}(\ttheta)$ then $s$ can be removed from $\Tau_{t'}$ for all $t'>t$.
\end{definition}
In other words, it means that inner functions which are unseen in the minimisation at some time step $t$ for all values in the parametric space $\Theta$ will never be seen again later and can therefore be safely discarded. This insight follows directly from the additive property \eqref{eq:additive}, which implies the equivalence $q_t^{r}(\ttheta) <  q_t^{s}(\ttheta) \iff   q_{t'}^{r}(\ttheta) < q_{t'}^{s}(\ttheta)$ for $t'>t$. Functional pruning is the most efficient pruning strategy.

\begin{proposition}\label{prop:max_pruning} (FPOP maximal pruning)
The FPOP-like pruning presented in Definition~\ref{pr:pruning_principle} is the maximal possible pruning for the optimisation problem of type \eqref{eq:global_cost} with additive cost property \eqref{eq:additive}. Say differently, at each time step $t$, the index set $\Tau_{t}$ obtained by FPOP is minimal: a smaller $\Tau_{t}$ would potentially lead to an under-optimal solution.
\end{proposition}

The proof relies on the following argument. If we remove an index $s_0$ in $\Tau_t$ such that $Q_t(\ttheta) = q_t^{s_0}(\ttheta)$ in a neighbourhood of $\ttheta_0$, we can prove that, by adding data points "centred on $\theta_0$" at further iterations (with unitary cost
$A(\ttheta) - \ttheta\cdot\nabla A(\ttheta_0)$), the last segment in \eqref{eq:fpop} would start at index $s_0+1$ at some later time $t_0$. However, since $s_0$ was pruned, the solution of the DP algorithm is no longer exact. 

It is important to note that, although pruning rules can drastically reduce execution time in practice, there exist examples of data for which no pruning rule could discard any index, leading back to the worst-case quadratic complexity of the optimal partitioning algorithm. We construct such examples using increasing time series where all inner cost functions attain the same minimal value at time $n$. We first illustrate this phenomenon in the simple univariate Gaussian case. 

\begin{proposition}
\label{prop:worstcaseGauss}
We consider the Gaussian univariate model with fixed variance. Its inner cost functions are given by relations $q_{t}^{s}(\theta) =  Q_s + \frac{t-s}{2}\theta^2 - S_{st}\theta + \beta$ with $S_{st}  = \sum_{i=s+1}^{t}y_i$ and $y_i \in \R$. If we observe the following data points: 
\begin{equation}\label{eq:baddata}
y_t = \sqrt{\frac{\beta}{n}}\Big(\sqrt{n-1}-\sqrt{t(n-t)}+\sqrt{(t-1)(n-t+1)} \Big)\,,\quad t = 1,\dots,n
\end{equation}
then no pruning happens, that is $\Tau_n = \{0,\ldots,n-1\}$.
\end{proposition}

We extend the result to a broader class of continuous distributions within the natural exponential family in Appendix~\ref{app:worstCase2}.

\subsection{Functional pruning rule as an optimisation problem}

In the evolution of the dynamic programming algorithm, if an inner function $q_t^s$ is still accessible, it can be discarded only by the last introduced function at time $t+1$, that is $Q_t + c(y_{t(t+1)}; \ttheta) + \beta$ (due to the additive property \eqref{eq:additive}). This naturally leads to the following pruning strategy: compare each inner function at time $t$ with the constant threshold $Q_t + \beta$. We denote by $\Theta_t^s \subset \Theta$ the set of parameter values for which the function $q_t^s$ attains the minimum among all candidate functions at time $t$, defined explicitly as:
\begin{equation}    
\label{eq:thetaSet}
    \Theta_t^s := \{ \ttheta\in\Theta \,|\, \forall r \in \Tau_t \setminus \{s\},\ q_t^s(\ttheta) \leq  q_t^{r}(\ttheta) \}\,.
\end{equation}

The criterion "$\Theta_t^s = \emptyset$ ?" forms the basis of the functional pruning rule introduced in FPOP \cite{maidstone2017optimal} for univariate data, and has since been adopted in several one-dimensional extensions \cite{jewell2020fast, romano2022detecting}. For multivariate problems, a recent method proposes to approximate the set  $\Theta_t^s$ using simple geometric shapes \cite{pishchagina2024geometricbased}. However, this approach becomes challenging to implement efficiently for non-Gaussian cost functions, resulting in relatively slow updates (see also \cite{runge2020finite}). Instead, with DUST, we propose to evaluate the minimal value of $q_t^s$ over its visibility region 
$\Theta_t^s$, and compare it directly to the threshold $Q_t + \beta$. While this may appear to introduce additional computation, it avoids the need to characterise the potentially complex shape of $\Theta_t^s$ explicitly. Thus, the task reduces to finding a single value -- or a lower bound on it -- which is often more tractable in practice.

\begin{proposition}[Functional pruning with minimum value rule]\label{prop:func-pruning}
We define $R_t^s:=\min_{\ttheta\in\Theta_t^s}q_t^s(\ttheta)$ with $\Theta_t^s$ given in \eqref{eq:thetaSet}. By convention, if $\Theta_t^s = \emptyset$, then $R_t^s=+\infty$.
    The functional pruning condition is as follows. If there exists $s$ in $\Tau_{t}$ satisfying:
    \begin{equation}\label{eq:func-pruning-condition}
        R_t^s > Q_t + \beta,
    \end{equation}
    then $s$ can never be the last change point of a segmentation of $y_{0t'}$ for all $t'> t$: that is $s \not\in \Tau_{t'}$.
\end{proposition}
\begin{proof}
    Let $\ttheta\not\in\Theta_t^s$. By definition of $\Theta_t^s$ there exists $r$ such that $\ q_t^s(\ttheta) >  q_t^{r}(\ttheta)$ and therefore $\ q_{t'}^s(\ttheta) >  q_{t'}^{r}(\ttheta)$. Let $\ttheta\in\Theta_t^s$. If the functional pruning condition~\eqref{eq:func-pruning-condition} holds, we have $q_{t'}^s(\ttheta) = q_t^s(\ttheta) + c(y_{tt'}, \ttheta) 
 > Q_t + \beta + c(y_{tt'}, \ttheta) = q_{t'}^t(\ttheta)$ which means that index $s$ is never seen in the minimisation problem (it is hidden by index $t$) and can be discarded. 
\end{proof}

Unfolding the definition of $\Theta_t^s$, we observe that the central challenge in pruning for multiple change-point detection is a problem of constrained optimisation. Its dual formulation forms the core of the DUST pruning method.

\begin{definition}[Pruning problem]
The pruning problem for multiple change-point detection at time $t$ for testing change point $s$ is the following problem of optimisation under constraints:
\begin{equation}\label{pb:pruningProblem}
 \left\{
      \begin{aligned}
        &\min_{\ttheta\in \Theta}   q_t^s(\ttheta)\,,\\
        &\hbox{   s.t.}\,\,q_t^s(\ttheta)- q_t^{r}(\ttheta) \le 0  \quad \hbox{for all} \quad r \in \Tau_t \setminus \{s\} \subset \{0,\ldots,t-1\}\,.\end{aligned}
    \right.
\end{equation}
A result greater than threshold $Q_t + \beta$ implies that the function $ q_t^s(\ttheta)$ can be removed and the index $s$ pruned.
\end{definition}

We aim to compute either the exact solution or a tight lower bound using a highly efficient algorithm. To this end, we begin by considering a simplified version of the pruning problem with a single constraint, which allows us to derive a lower bound. This setting offers a clear framework to introduce and illustrate our dual approach. Furthermore, in the case of one-parametric cost functions, this formulation enables the design of a particularly efficient pruning test.

\subsection{Notations for dual and decision functions}

For clarity, we compile below all the notations related to dual and decision functions used in the sequel. The function $\mathcal{D^*}$ is given by $\mathcal{D^*}(x) = x \cdot (\nabla A)^{-1}(x) - A((\nabla A)^{-1}(x))$ and is strictly convex. We introduce the following mean values:
$$\overline{\mathbf{S}}_{rs} = \left\{\begin{aligned}
&\frac{{\mathbf{S}}_{rs}}{s-r}\quad\hbox{if}\quad r < s\,,\\
&\frac{{\mathbf{S}}_{sr}}{r-s}\quad\hbox{if}\quad s < r\,,
\end{aligned}
\right. \quad \hbox{and index indicator}\quad
\psi_{rs} = \left\{
      \begin{aligned}
1\quad &\hbox{if}\quad r < s\,,\\
-1\quad &\hbox{if}\quad s < r\,. \end{aligned}
\right.
$$
By analogy, we also use a mean value between global costs, $\overline{Q}_{rs} =\frac{Q_s-Q_r}{s-r}$, as well as a difference operator between mean values, which incorporates the order between the first two indices:
$$\Delta \overline{\mathbf{S}}_{rst} = \psi_{rs}(\overline{\mathbf{S}}_{st}- \overline{\mathbf{S}}_{rs})\quad, \quad \Delta \overline{Q}_{rst} = \psi_{rs}( \overline{Q}_{st} - \overline{Q}_{rs})\,.$$
The absence of bold notation (e.g., 
$S_{rs},\overline{S}_{rs}, \Delta \overline{S}_{rst},\ldots$) indicates that the data is univariate. In this case, some standard notation will also be used for means and variance when $r<s$:
$$\overline{S}_{rs} = \overline{y}_{rs} = \frac{1}{s-r}\sum_{i=r+1}^s y_i\, , \quad \overline{S^2_{rs}} = \overline{y_{rs}^2} = \frac{1}{s-r}\sum_{i=r+1}^s y_i^2\,,$$
and $V(y_{rs}) = \overline{y_{rs}^2} - (\overline{y}_{rs})^2$. Eventually, we define the following linear function:
$$\sigma_{\mathcal{R}}(\mathbf{x}) = \overline{\mathbf{S}}_{st} + \sum_{r\in \mathcal{R}}x_{r}\Delta\overline{\mathbf{S}}_{rst} \quad,\quad \phi_{\mathcal{R}}(\mathbf{x}) = \overline{Q}_{st} + \sum_{r\in \mathcal{R}} x_r \Delta \overline{Q}_{rst}\,.$$
The vector $\mathbf{x} = (x_r)_{r \in \mathcal{R}}$ is identified with $(x_r)_{r=1,\ldots,|\mathcal{R}|}$ depending on the context (same with vector $\mu$). For a vector $\mathbf{v} \in \R^d$, its Euclidean norm is denoted by $\|\mathbf{v}\|$.

\section{DUST method with one constraint}
\label{sec:simple}

We present the dual formulation for the single-constraint case without delving into the theoretical analysis of the dual, which is deferred to Section~\ref{sec:duality} in a generic framework. We provide a detailed description of the DUST change-point algorithm. As a concrete illustration, we explore the change-in-mean-and-variance problem in depth. This simplified presentation is intended to make our methodology more accessible and to promote its application in other settings.

\subsection{A one-dimensional dual function}
\label{subsec:dualtrick}

Optimisation problem \eqref{pb:pruningProblem} with one constraint is as follows:  
\begin{equation}\label{pb:oneConstraint}
\left\{
      \begin{aligned}
        &\min_{\ttheta\in \Theta} q_t^s(\ttheta)\,,\\
        &\hbox{   s.t.}\,\, q_t^s(\ttheta)- q_t^{r}(\ttheta) \le 0  \,. \end{aligned}
    \right.
\end{equation}
where we consider that $r < s$. With $r > s$, the obtained (convex) problem \eqref{pb:oneConstraint} leads to a very inefficient pruning. The optimal value of \eqref{pb:oneConstraint} is defined as $R_t^{rs}$. We can easily write down the Lagrangian function as:
\begin{align*}
\mathcal{L}(\ttheta, \mu) &= q_t^s(\ttheta)+ \mu( q_t^s(\ttheta)- q_t^{r}(\ttheta) )\,, \\
&= \Big((t-s)-\mu(s-r)\Big) \Big(A(\ttheta) - \ttheta\cdot \frac{ \mathbf{S}_{st} - \mu \mathbf{S}_{rs}}{(t-s) - \mu (s-r)}\Big) + \mu(Q_s - Q_r) + Q_s  + \beta \,,
\end{align*}
or with a change of variable $\mu \to \mu \frac{t-s}{s-r}$ and the parametric mean $\mathbf{m}(\mu) = \frac{ \overline{\mathbf{S}}_{st} - \mu \overline{\mathbf{S}}_{rs}}{1 - \mu}$:
$$
\mathcal{L}(\ttheta, \mu)= (t-s)\Big((1-\mu) (A(\ttheta) - \ttheta\cdot \mathbf{m}(\mu)) + \mu  \frac{Q_{s}-Q_{r}}{s-r} \Big)+ Q_s + \beta\,.$$
The dual function presented in the following proposition serves as a lower bound for the target quantity $R_t^s$ (see Proposition~\ref{prop:func-pruning}). 

\begin{proposition}
\label{prop:dual1D}
The dual function $\D: [0, \mu_{max}) \to \R$ to Problem \eqref{pb:oneConstraint} with inner cost functions \eqref{eq:qts} and cost \eqref{eq:cost} for pruning index $s$ using index $r$ with $r<s$ is
given by:
\begin{equation}
\label{dualFunction1C}
\D(\mu) = (t-s)\Big(-(1-\mu) \mathcal{D^*}(\mathbf{m}(\mu)) + \mu  \overline{Q}_{rs} \Big)+ Q_s + \beta\,,
\end{equation}
where $\mathbf{m}(\mu) = \frac{ \overline{\mathbf{S}}_{st} - \mu \overline{\mathbf{S}}_{rs}}{1 - \mu}$. The domain of the dual is a segment bounded by a value $\mu_{max}$ which is the largest value for which $\mathcal{D^*}(\mathbf{m}(\mu))$ is finite or can be computed (e.g. when we have log terms). The value $\mu_{max}$ is always smaller than or equal to $1$.
\end{proposition}

At time $t$, we prune index $s$ using a candidate index $r < s$ if $\mathcal{D}(\mu_0) > Q_t + \beta$, for some $\mu_0 \in [0, \mu_{\text{max}})$. The solution $R_t^{rs}$ to Problem~\eqref{pb:oneConstraint} is smaller than the solution $R_t^s$ to Problem~\eqref{pb:pruningProblem}, and is lower bounded by the dual function: for all $\mu \in [0, \mu_{\text{max}})$, we have: $\mathcal{D}(\mu) \le R_t^{rs} \le R_t^s$. Therefore, using any value $\mu_0 \in [0, \mu_{\text{max}})$, pruning based on the condition $\mathcal{D}(\mu_0) > Q_t + \beta$ is safe: it guarantees that $R_t^s > Q_t + \beta$, so no index is removed by mistake. Ideally, we would evaluate the dual function at its maximum or some value close to it.

\subsection{DUST decision rule}
\label{subsec:dustDecision}

Pruning relies on the existence of a value $\mu_0$ such that $\mathcal{D}(\mu_0) - (Q_t + \beta) > 0$. Since dividing the left-hand side by any positive function does not change the decision, we normalise the dual function $\mathcal{D}$ from Proposition~\ref{prop:dual1D} to allow its direct maximisation (with a closed formula) in some particular cases. The new function $\mathbb{D}$, called the decision function, has this simple property: pruning of index $s$ occurs as soon as $\mathbb{D}$ returns a positive value for some point in its domain.

\begin{proposition}
\label{prop:dual1D_max}
The decision function $\mathbb{D}: [0, x_{max}) \to \R$ to Problem \eqref{pb:oneConstraint} with inner cost functions \eqref{eq:qts} and cost \eqref{eq:cost} for pruning index $s$ using index $r$ with $r<s$ is
given by:
\begin{equation}
\label{eq:decisionDual1D}\mathbb{D}(x) = -\D^*\Big(\overline{\mathbf{S}}_{st} + x\Delta \overline{\mathbf{S}}_{rst}\Big) - \Big(\overline{Q}_{st} + x\Delta \overline{Q}_{rst}\Big) \,.
\end{equation}
The value $x_{max}$ is derived from $\mu_{max}$.
\end{proposition}

\begin{proof}
We shift the dual $\D$ of  Equation~\eqref{dualFunction1C} by the quantity $Q_t + \beta$ to transform the dual pruning threshold test into a sign test. We then divide the obtained function by $(t-s)(1-\mu)$, which is always a positive value. We eventually introduce the variable $x = \frac{\mu}{1-\mu}\ge 0$. This is summarised by the transformation:
$$\mathbb{D}(x) = \Big((t-s)(1-\frac{x}{1+x}\Big)^{-1}\Big(\D\Big(\frac{x}{1+x}\Big) - (Q_t + \beta)\Big)\,.$$
\end{proof}

A typical setting in change-point detection is that of univariate data drawn from the exponential family \cite{chen2000parametric, JSSv106i06, romano2023fast}. For this setting, we can explicitly compute the maximum of the decision function.

\begin{theorem}
\label{prop:dual1D_max_exact}
The decision function $\mathbb{D}: [0, x_{max}) \to \R$ to Problem \eqref{pb:oneConstraint} given in Proposition~\ref{prop:dual1D_max} admits a closed-form maximum for the single-constraint case and one-parametric cost functions. If the argument of the maximum $x^{\star}$ of this function is not located on the frontier of the domain, then we prune index $s$ using index $r$ with $r<s$ when:
\begin{equation}
\label{eq:DUSTexact1D1C}
q_t^s\Big( -\frac{\Delta \overline{Q}_{rst}}{\Delta \overline{S}_{rst} } \Big) > Q_t + \beta\,.  
\end{equation}
\end{theorem}

Note that certain special cases must be handled before testing inequality \eqref{eq:DUSTexact1D1C}, such as $x_{\max} = 0$, $\overline{S}_{st} = \overline{S}_{rs}$ (i.e., the decision function is linear) and $x^\star \notin (0, x_{\max})$. This theorem advantageously replaces the PELT pruning rule, which was given by $q_t^s((\nabla A)^{-1}(\overline{S}_{st})) > Q_t + \beta$. We will highlight the effectiveness of the DUST rule in simulation Section~\ref{sec:simus}.



\subsection{DUST algorithm}
\label{subsec:dust1constraint}

We have designed our new pruning method to be as straightforward and practical as the PELT rule. The DUST multiple change-point detection algorithm is described in Algorithm~\ref{algo:dust}. Here, the value $c(y_{st})$ is defined as the minimum of the segment cost function $\ttheta \mapsto c(y_{st}; \ttheta)$. For the smallest index in $\mathcal{T}_t$, we use the standard PELT pruning rule by default. The index $r$ is selected from a given probability distribution $p_{r}$ over the indices in $\mathcal{T}_t$ that are smaller than $s$, while the value $\mu_0$ is sampled from a distribution $p_{\mu}$ supported on the segment $[0, \mu_{\text{max}} \le 1)$. For one-parametric cost functions, we can consider a Dirac data-dependent $p_{\mu}$ which matches relation~\eqref{eq:DUSTexact1D1C} in addition to its limit cases.

Simulations (Section~\ref{sec:simus}) explore several choices for the distributions $p_r$ and $p_{\mu}$. In the single-constraint setting, an effective strategy is to select $r$ as the closest index below $s$, and to choose $\mu_0$ as the argument that maximises the dual function, using a simple iterative maximisation procedure or relation~\eqref{eq:DUSTexact1D1C} for one-parametric cost functions. For simplicity, the algorithm is presented with the dual $\D$ with bounded domain, but a similar algorithm can be written with decision function  $\mathbb{D}$. From the output of Algorithm~\ref{algo:dust}, we can easily recover the optimal segmentation for the initial problem~\eqref{eq:global_cost} via a standard backtracking step described in Algorithm~\ref{algo:backtrack}.

\begin{algorithm}[!ht]
\caption{DUST algorithm}
\label{algo:dust}
\begin{algorithmic}[1]
\Require Time series $y_{0n}$, penalty value $\beta>0$, rules $p_r$, $p_{\mu}$
\Ensure Sequences of global costs $Q_{0n}$ and last changes $\hat s_{0n}$ 
\State $Q_0 \gets 0,\quad \mathcal{T}_1 \gets \{0\}$
\For{$t = 1,\dots, n$} \Comment{OP step}
    \State $Q_t \gets \min_{s \in \mathcal{T}_t} \left\{ Q_s + c(y_{st}) + \beta \right\}$
    \State $\hat s_t \gets \argmin_{s \in \mathcal{T}_t} \left\{Q_s + c(y_{st}) + \beta \right\}$
    \For{$s \in \mathcal{T}_t$}  \Comment{DUST pruning step}
       \State Draw $r$ in $\mathcal{T}_t$ such that $r < s$ from distribution $p_{r}$
        \State Compute $\mu_{max} = \mu_{max}^{r,s,t}$ (data dependent) 
        \State Draw $\mu_0$ in $[0, \mu_{max})$ from distribution $p_{\mu}$
         \If {$\D(\mu_0) > Q_t + \beta$}  \Comment{DUST test}
         \State {$\mathcal{T}_t \gets \mathcal{T}_t \backslash\{s\}$}
       \EndIf
    \EndFor
 \State $\mathcal{T}_{t+1} \gets \mathcal{T}_t \cup  \{t\}$
\EndFor
\end{algorithmic}
\end{algorithm}

\begin{algorithm}[!ht]
\caption{Backtracking the change-point locations}
\label{algo:backtrack}
\begin{algorithmic}[1]
\Require $\hat s_{0n}$
\Ensure Set of optimal change point indices $\widehat{\Tau} = \{\tau_1, \tau_2,\dots\}$
\State $\tau \gets n$, $\widehat{\Tau} \gets \emptyset$
\While{$\tau>0$}
\State $\widehat{\Tau} \gets (\tau, \widehat{\Tau})$, $t \gets \hat{s}_{\tau}$
\EndWhile
\end{algorithmic}
\end{algorithm}

\begin{remark}
DUST extends PELT by evaluating the dual function beyond zero. While PELT applies the test $\mathcal{D}(0) > Q_t + \beta$, derived from the unconstrained form of Problem~\eqref{pb:oneConstraint}, DUST samples $\mu_0$ in $[0, \mu_{\text{max}})$, potentially yielding stronger pruning. This makes DUST especially effective in time series with few change points, where PELT tends to prune poorly. Despite its improved pruning, DUST maintains a complexity close to PELT, with minimal overhead from computing $\mu_{\text{max}}$ and selecting $r$ and $\mu_0$.
\end{remark}

\subsection{Example of the change-in-mean-and-variance problem}
\label{subsec:example}

Detecting changes in both the mean and variance of a large Gaussian time series is challenging. To the best of our knowledge, no efficient implementation currently exists in this setting due to two main difficulties. First, the problem involves two parameters, which limit the applicability of methods like FPOP \cite{maidstone2017optimal}. Second, a logarithmic term in the cost function complicates root-finding, making methods such as geomFPOP \cite{pishchagina2024geometricbased} difficult to apply. We address this challenge using the DUST method. We write down the likelihood of segment $y_{st}$ with mean $m$ and variance $\sigma^2$,
$$
\mathcal{L}(y_{st}; [m; \sigma^2])=  \prod_{i=s+1}^t \biggr[ \frac{1}{\sqrt{2\pi\sigma^2}}\exp\Big(-\frac{(m-y_i)^2}{2\sigma^2}\Big) \biggr]\,,
$$
and its associated cost:
$$c(y_{st}; [m; \sigma^2]) =  (t-s)\Big(\frac{(m-\overline{y}_{st})^2 + V (y_{st})}{2\sigma^2} + \frac{1}{2}\log(\sigma^2)\Big)\,.$$

With the change of variable $(\theta_1,\theta_2) = (\frac{m}{\sigma^2}; -\frac{1}{2\sigma^2}) \in \R \times \R^-$, we obtain the canonical form of Example~\ref{ex:MeanAndVar} with $A(\theta_1,\theta_2) =-\frac{\theta_1^2}{4\theta_2} + \frac{1}{2}\log(-\frac{1}{2\theta_2})$. We can now derive the exact maximum value of the associated decision function.

\begin{lemma}
The one-dimensional decision function for the change-in-mean-and-variance problem with a single constraint from the cost function \eqref{eq:MeanAndVar} is given by the function:
$$
\mathbb{D}(x) =  \frac{1}{2}\Bigg[1+ \log\Big( \overline{S^2_{st}}+x\Delta\overline{S_{rst}^2}-(\overline{S}_{st}+ x\Delta\overline{S}_{rst})^2\Big)\Bigg]- \Big(\overline{Q}_{st} + x\Delta \overline{Q}_{rst}\Big)\,,$$
with $x \in (x_0 - \sqrt{x_1}, x_0 + \sqrt{x_1})$ and 
$$x_0 = \frac{1}{2}\Big(\frac{V(y_{st}) - V(y_{rs})}{(\Delta\overline{S}_{rst})^2} - 1\Big)\,\quad x_1 = x_0^2 +  \frac{V(y_{st})}{(\Delta\overline{S}_{rst})^2}$$
and its maximum, with notation $x_2 = \Delta \overline{Q}_{rst}$, is evaluated in:
$$x^{\star} = \max \Big\{0, x_0 + (2 x_2)^{-1} - \operatorname{sign}(x_2) \sqrt{x_1 + (2 x_2)^{-2}}\Big\}\,.$$
\end{lemma}

In Figure~\ref{fig:pruningMV}, we evaluate the efficiency of DUST for this example. For each index~$s$ considered for pruning, we choose the biggest non-pruned index smaller than $s$ and evaluate the dual at its maximum position $x^{\star}$. With $n=10^4$, $\beta = 4 \log(n)$ and data with no change ($\mathcal{N}(0,1)$), only $2.95\%$ of the indices remain non-pruned and the overall number of indices to be considered is reduced by a factor $28$ compared to PELT (left panel). 

\begin{figure}[!t]%
\centering
\includegraphics[width=0.47\textwidth]{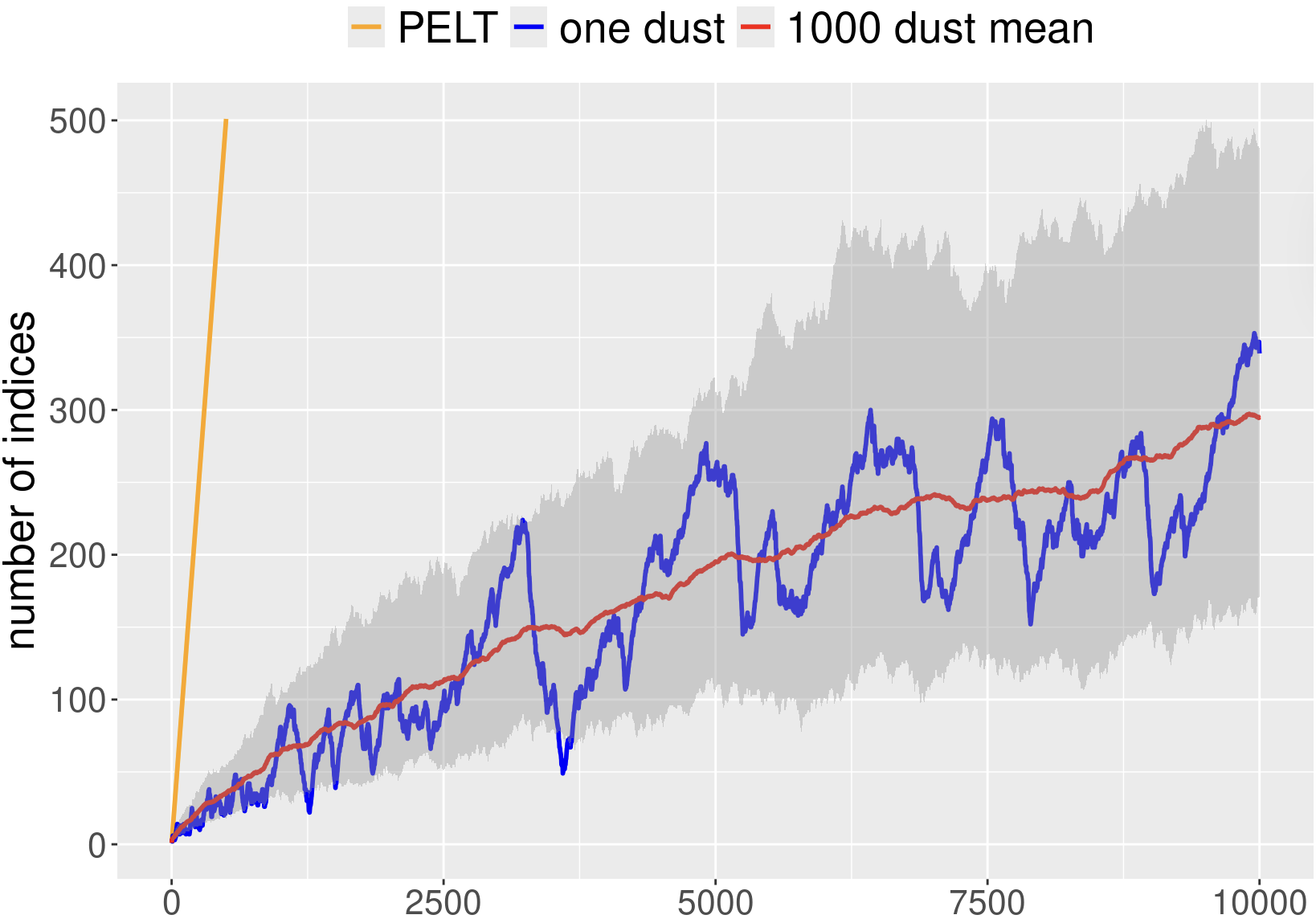}
\includegraphics[width=0.47\textwidth]{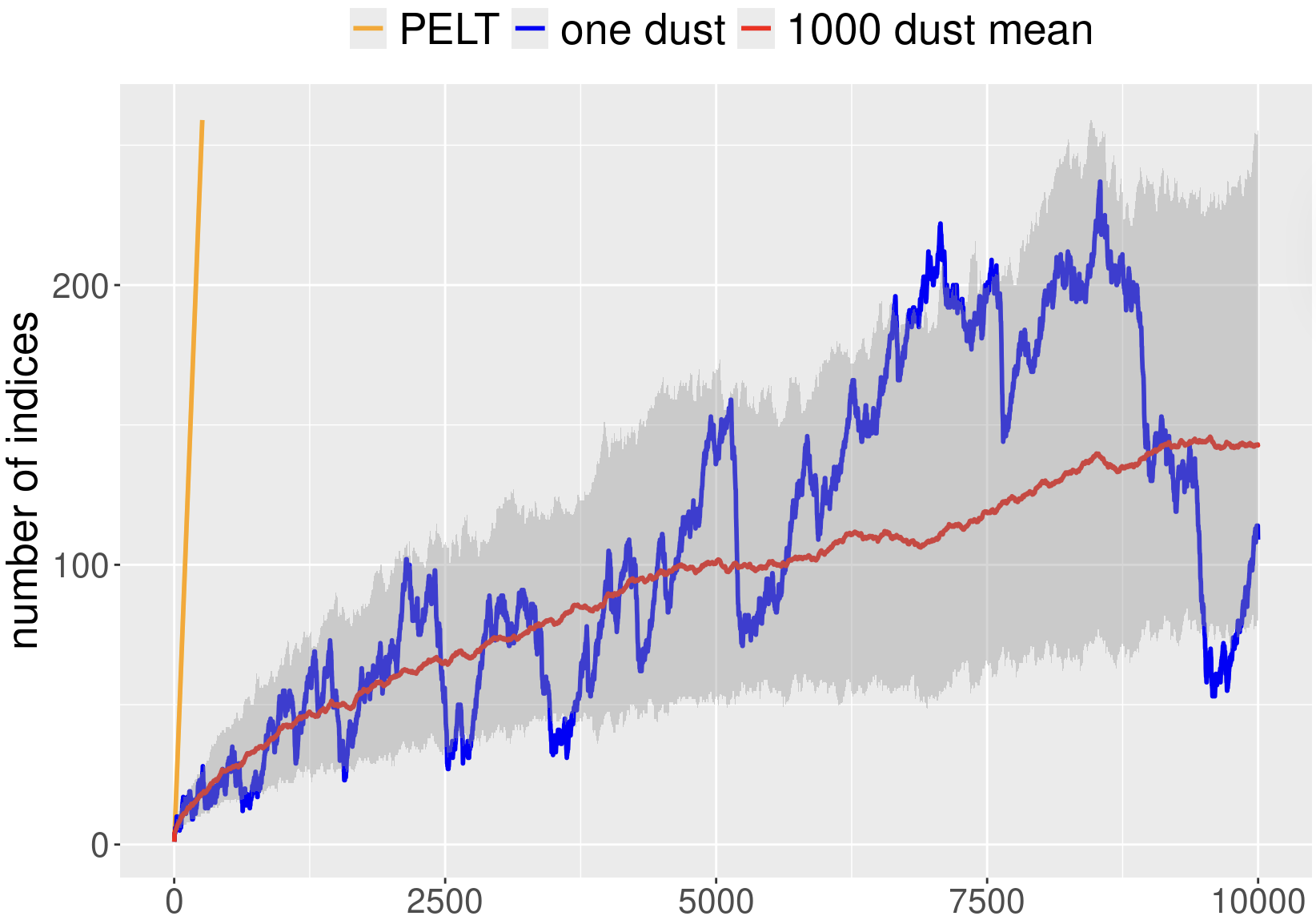}
\caption{For data with no change, the number of indices saved by DUST over time is consistently much less than PELT (no pruning), for the one-constraint DUST (left) and the two-constraint case (right).}\label{fig:pruningMV}
\end{figure}

A higher lower bound for $R_t^s$ is obtained when we consider two constraints: $q_t^s - q_t^{r_1} \le 0$ and $q_t^s - q_t^{r_2} \le 0$. This leads to a bi-parametric decision function:
$$
\mathbb{D}(x_1,x_2) =  \frac{1}{2}\Bigg[1+ \log\Big( \overline{S^2_{st}}+x_1\Delta\overline{S_{r_1st}^2}+x_2\Delta\overline{S_{r_2st}^2})-(\overline{S}_{st}+ x_1\overline{S}_{r_1st} + x_2\overline{S}_{r_2st})^2\Big)\Bigg]$$
$$-\Big(\overline{Q}_{st} + x_1 \Delta \overline{Q}_{r_1st} + x_2 \Delta \overline{Q}_{r_2st}\Big)\,.$$
Such a general dual and its properties are studied in the next Section. In this mean-and-variance problem, explicit maximal values are available for the three values to be tested: $\max_{x_1}\mathbb{D}(x_1,0)$,  $\max_{x_2}\mathbb{D}(0,x_2)$ and  $\max_{x_1,x_2}\mathbb{D}(x_1,x_2)$ with positive $x_1, x_2$. In Figure~\ref{fig:pruningMV} (right panel) we run the same simulation as in the left panel but with a bi-parametric dual, built from the two highest indices smaller than $s$ ($r_1 < r_2 < s$). Only $1.42\%$ of the indices remain non-pruned and the overall number of indices to be considered is reduced by a factor $54$ compared to PELT. In this example, $0$ indice is pruned by PELT, $36\%$ of indices are pruned by the $r_2$-dual,  $23\%$ of indices are pruned by the $r_1$-dual and $40\%$ by the maximum of the  $(r_1,r_2)$-dual (function $\mathbb{D}(x_1,x_2)$ with $x_1 > 0$ and $x_2 > 0$). We recall that this DUST pruning procedure does not involve any iterative routine, but relies solely on three inequality tests (as PELT) evaluated at three points of the decision function, each computed from a closed-form expression. With $10^6$ data point, the percentage of remaining indices drops to $0.5\%$. 

\section{Duality function in change-point problems}
\label{sec:duality}

\subsection{The shape of the dual}

We consider the pruning problem \eqref{pb:pruningProblem}: $\min_{\ttheta\in \Theta} q_t^s(\ttheta)$ under constraints $q_t^s(\ttheta)- q_t^{r}(\ttheta) \le 0$ for all indices $r \ne s$ in $\{0,\ldots,t-1\}$. As for the single-constraint case, the dual and the decision functions can be explicitly written in closed form.

\begin{theorem}
\label{th:dualMultiple}
The dual function $\D_{st}: (\R^+)^{t-1} \to \R$ to problem~\eqref{pb:pruningProblem} with inner cost functions \eqref{eq:qts} and cost \eqref{eq:cost} for pruning index $s$ using all indices $r$ with $r\ne s$ is
given by:
$$\D_{st}(\mu) =  (t-s) \Bigg[- l(\mu)  \mathcal{D^*}(\mathbf{m}(\mu)) + \sum_{r \ne s}\mu_{r}\psi_{rs}\overline{Q}_{rs}\Bigg]  + Q_s + \beta\,,$$
with $\mu = (\mu_{r})_{0 \le r < t, r \ne s}$ and
\begin{equation}\label{eq:dual_ratio}
\mathbf{m}(\mu) = \frac{ \overline{\mathbf{S}}_{st} - \sum_{r \ne s}\mu_{r}\psi_{rs} \overline{\mathbf{S}}_{rs}}{1 -\sum_{r \ne s}\mu_{r}\psi_{rs}} \in \R^d\,,\quad
l(\mu) = 1 -\sum_{r \ne s}\mu_{r}\psi_{rs}\,.    
\end{equation}
We also define the decision function $\mathbb{D}_{st}$ (as done in previous section) with the change of variable $x_r = \mu_r/l(\mu)$ and notation $\mathbf{x} =  (x_{r})_{0 \le r < t, r \ne s}$. It is given by a simpler relation:
\begin{align}
\label{eq:decisionDualMD}
\mathbb{D}_{st}(\mathbf{x}) 
&= -\D^*\Big(\overline{\mathbf{S}}_{st} + \sum_{r \ne s} x_{r} \Delta \overline{\mathbf{S}}_{rst}\Big) - \Big(\overline{Q}_{st} + \sum_{r \ne s} x_r \Delta \overline{Q}_{rst}\Big)\,, \nonumber \\
&= -\D^*(\sigma(\mathbf{x})) - \phi(\mathbf{x})\,,
\end{align}
leading to the pruning rule $\mathbb{D}_{st}(\mathbf{x}_0) >0$ for some chosen $\mathbf{x}_0$ in the domain $\Omega_{\mathbf{x}}$.
\end{theorem}

Notice that $\mathbf{m}(\mu)$ is a vector if the statistics $\mathbf{S}_{st}$ are vector-valued. The case $\overline{\mathbf{S}}_{st}=\overline{\mathbf{S}}_{rs}$ for all $r \ne s$ can be left apart. It corresponds to $\mathbf{m}(\mu) = \overline{\mathbf{S}}_{st}$ and the dual $\D_{st}$ becomes a linear function. The values of the decision function can be interpreted as evaluations of $(t-s)^{-1}(q_t^s - (Q_t + \beta))$ where values $\overline{\mathbf{S}}_{st}$ and $\overline{Q}_{st}$ are replaced by linear combinations of the available $\Delta \overline{\mathbf{S}}_{rst}$ and $\Delta \overline{Q}_{rst}$.

\begin{remark}
With constraints of type "$r < s$" only, the domain $\Omega_\mu$ of $\D_{st}$ is bounded and included in the simplex of dimension $d$.
\end{remark}

Appendix \ref{app:exampleDual} presents the form of the function $\D^*$ for many distributions. We give here the explicit form of the decision function for two examples in dimension $d$ with $q$ constraints (in index set $\mathcal{R} \subset \{0,\ldots,t-1\}$) with $|\mathcal{R}| = q \le d$.

(i) Multivariate Gaussian:
$$\mathbb{D}_{st}^{\mathcal{N}}(\mathbf{x})= -\frac{1}{2} \sum_{i=1}^d \sigma_i(\mathbf{x})^2 - \phi(\mathbf{x}) = -\frac{1}{2} \sum_{i=1}^d \left((\overline{\mathbf{S}}_{st})_i + \sum_{r \in \mathcal{R}} x_r (\Delta \overline{\mathbf{S}}_{rst})_i \right)^2 - \phi(\mathbf{x})\,.$$

(ii) Multivariate Bernoulli:
$$\mathbb{D}_{st}^{\mathcal{B}}(\mathbf{x}) =  -\sum_{i=1}^d \Bigg[\sigma_i(\mathbf{x})\log(\sigma_i(\mathbf{x})) -  (1-\sigma_i(\mathbf{x}))\log(1-\sigma_i(\mathbf{x}))\Bigg] - \phi(\mathbf{x})\,.$$

In this latter case, its domain $\Omega_{\mathbf{x}}$ is at the intersection of $3d$ half-spaces in dimension $q$ given by relations $\sigma_i(\mathbf{x}) \ge 0$, $1 - \sigma_i(\mathbf{x}) \ge 0$ and the positive orthant $\mathbf{x} > 0$.  The shape of the domains ($\Omega_{\mathbf{x}}$ and $\Omega_{\mu}$) depends on the underlying distribution used.

\begin{definition}[DUST pruning rule with multiple constraints]
\label{def:dualmultiD}
We consider the dual function $\mathcal{D}_{st}: \Omega_\mu \to \mathbb{R}$ and the decision function $\mathbb{D}_{st}: \Omega_{\mathbf{x}} \to \mathbb{R}$. At time $t$, we prune index $s$ using candidate indices from a set $\mathcal{R} \subset \{0, \ldots, t-1\} \setminus \{s\}$ to construct the dual, if there exists $\mu_0 \in \Omega_\mu$ such that $\mathcal{D}_{st}(\mu_0) > Q_t + \beta$, or equivalently, if there exists $\mathbf{x}_0 \in \Omega_{\mathbf{x}}$ such that $\mathbb{D}_{st}(\mathbf{x}_0) > 0$.
\end{definition}

This definition leaves implicit three important choices that must be addressed: (i) the selection of candidate indices in $\mathcal{R}$, (ii) the method for identifying a suitable evaluation point ($\mu_0$ or $\mathbf{x}_0$), and (iii) a precise characterisation of the domain of definition ($\Omega_\mu$ or $\Omega_{\mathbf{x}}$). Since these domains always lie within the positive orthant, solving the equation $\nabla \mathbb{D}_{st}(\mathbf{x}) = 0$ is never a valid option. For now, we address the problem using a quasi-Newton iterative algorithm or random sampling, based on a carefully chosen distribution over $\Omega_\mu$. This domain is defined through the mean parameter space $\mathcal{M}$ \cite{wainwright2008graphical}. We can show that the correct number $q$ of constraints to consider is upper bounded by the dimension $d$ of the parameter space $\Theta \subset \mathbb{R}^d$.

\subsection{d-Strong duality}

Geometrically, the solution to our optimisation problem~\eqref{pb:pruningProblem} lies at the intersection of at most $d+1$ inner functions in a parameter space of dimension $d$. This implies that no more than $d$ constraints can be active at the solution point. Identifying the correct subset of $q \le d$ active constraints seems out of reach. Nevertheless, focusing on at most $d$ constraints among the available $t-1$ is particularly appealing: in this case, the maximum of the dual function coincides with the solution of the original optimisation problem.
\begin{theorem}
\label{th:noDualityGap}
When the dual function $\D_{st}$ is built with at most $d$ constraints for d-parametric cost functions, there is no duality gap: $$R_t^s:=\min_{\ttheta\in\Theta_t^s(\mathcal{R})\subset \R^d}\Big\{q_t^s(\ttheta)\Big\} = \max_{\mu \in \Omega_{\mu} \subset \R^q} \Big\{\D_{st}(\mu)\Big\}\,,$$
where $\Theta_t^s(\mathcal{R})$ is the feasible set defined by the $q$ constraints with $q \le d$ such that $q_t^s(\ttheta) - q_t^{r_i}(\ttheta) \le 0$ with $r_i \in \mathcal{R}$, $i=1,\ldots,q = |\mathcal{R}|$ (see also~\eqref{eq:thetaSet}).
\end{theorem}

The proof provided in Appendix~\ref{app:noDualityGap} is based on the geometric interpretation of duality. We consider $d$ arbitrary constraints. In case, $q < d$, we can still add $d-q$ constraints of type $q_t^s(\ttheta)- q_t^{r}(\ttheta) \le 0$ that we know would be unused (no equality at the optimum point). With $d$ constraints, we consider the geometric object $\mathcal{O}$ in the Euclidean space $\R \times \R^{d}$ whose elements $(x_0,x_1,...,x_d)$ are given by a parametric representation:
$$\mathcal{O}=\Big\{(x_0,...,x_d)\,,\exists\, \ttheta \in \R^d\,,\,(x_0,...,x_d) = (q_t^s(\ttheta), (q_t^s- q_t^{r_1})(\ttheta),\ldots, (q_t^s- q_t^{r_d})(\ttheta))\Big\}\,.$$
We show that the epigraph of the objective function $\mathcal{O}$ is convex. This result relies primarily on the fact that the functions $q_t^j$ are similar. Notably, the duality gap is zero -- that is, strong duality holds -- even in the presence of concave constraints (i.e., when $r < s$).\\

We illustrate the failure of strong duality when $q > d$ by showing that it does not hold with $d+1$ constraints. This is demonstrated using a simple example of a one-dimensional Gaussian cost with three data points, $(y_1, y_2, y_3) = (2, -1, 0)$, penalty $\beta = 2$, and initial cost $Q_0 = -\beta$. These data points yield the following functions: $q_3^0(\ttheta) = \frac{3}{2}\theta^2 - \theta$, $q_3^1(\ttheta) = \theta^2 + \theta$ and $q_3^2(\ttheta) = \frac{1}{2}\theta^2 + \frac{3}{2}$. In this example, the dual function (without normalising the Lagrangian parameters) for testing index $2$ with indices $0$ and $1$ is given by:
$$
\mathcal{D}(\mu_1, \mu_2) = \frac{1}{2}\frac{(\mu_1 - \mu_2)^2}{1 - 2\mu_1 - \mu_2} + \frac{3}{2}(1 + \mu_1 + \mu_2)\,,
$$
which attains a maximum value of $2.5$. The solution to the corresponding optimisation problem is $\frac{10 + \sqrt{7}}{4}$, as illustrated in Figure~\ref{fig:dualCE}.

\begin{figure}[h!]
\centering
\includegraphics[width=0.6\textwidth]{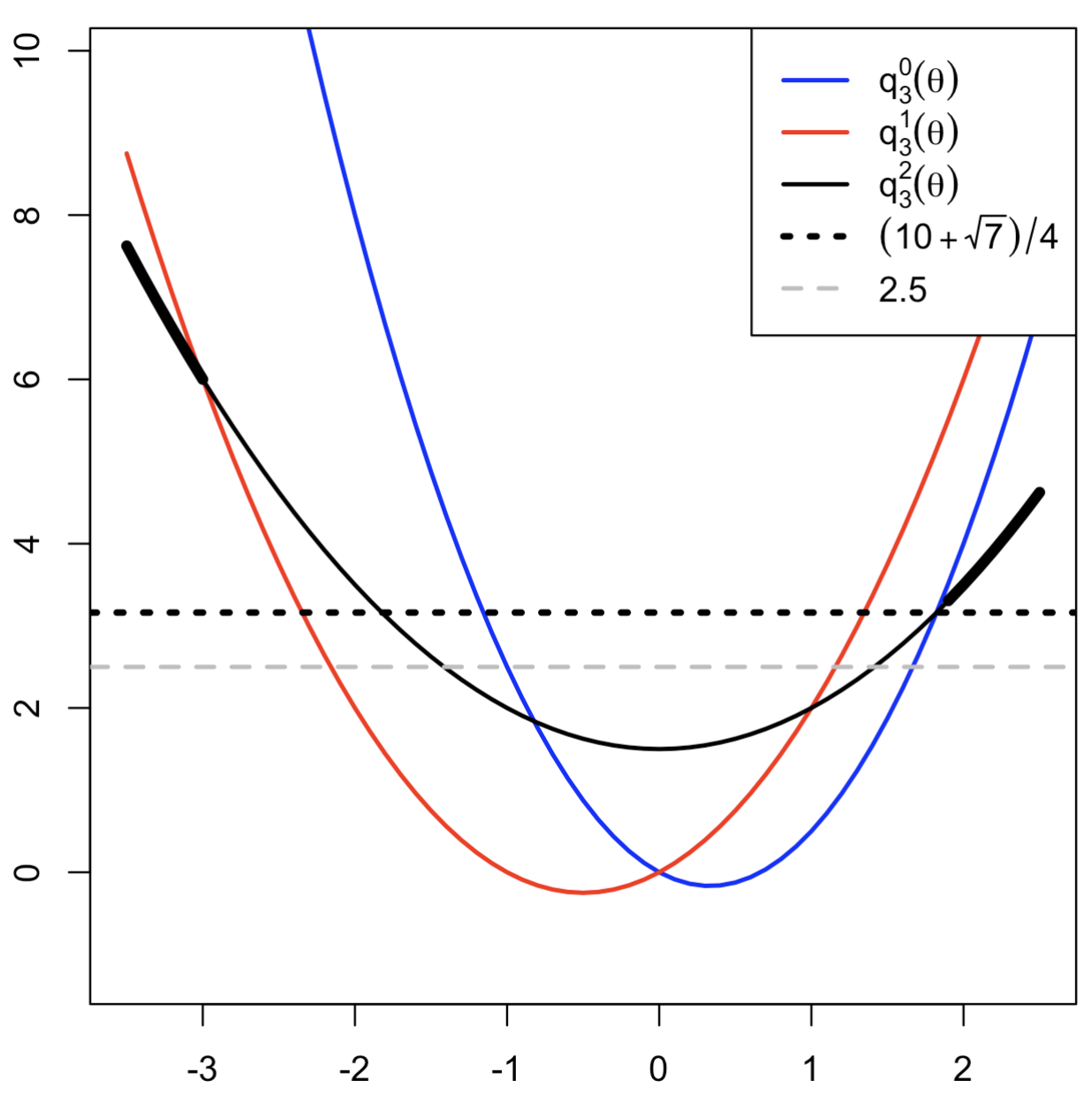}
\caption{Illustration of the dual maximum value ($2.5$) and the true threshold $(\frac{10 + \sqrt{7}}{4})$, solution of the optimisation problem~\eqref{pb:pruningProblem}. We use a Gaussian cost with three points, $(y_1, y_2, y_3) = (2, -1, 0)$, penalty $\beta = 2$ and initial cost $Q_0 = -\beta$.}
\label{fig:dualCE}
\end{figure}

\subsection{Gaussian case}
\label{subsec:gaussian}

The maximum of the dual function can be explicitly expressed in closed form in the single-constraint case. When the time series contains no change points and the penalty parameter is appropriately chosen, the expression simplifies further using the relation  $Q_u = -\frac{u}{2}\|\overline{\mathbf{S}}_{0u}\|^2$.

\begin{proposition}
In $d$-variate Gaussian model the maximum value of the dual with one constraint ($r < s <t$), $\max_{\mu \in [0,1]}\Big\{\D^{gauss}(\mu)\Big\}$, is given by expression:
$$\left\{\begin{aligned}
&- \frac{t-s}{2}\|\overline{\mathbf{S}}_{st}\|^2 +Q_s + \beta\quad\hbox{if}\quad \|\Delta\overline{\mathbf{S}}_{rst} \| \ge R_{rs} \quad (\hbox{case } \mu = 0)\,,\\
& - \frac{t-s}{2}\|\overline{\mathbf{S}}_{st}\|^2 +Q_s + \beta + \frac{t-s}{2}(  \|\Delta\overline{\mathbf{S}}_{rst} \|- R_{rs})^2\quad\hbox{if}\quad \|\Delta\overline{\mathbf{S}}_{rst} \| < R_{rs}\,,
\end{aligned}\right.  $$
where:
$R_{rs} =  \sqrt{\|\overline{\mathbf{S}}_{rs}\|^2 + 2 \overline{Q}_{rs}}$. Index $s$ is pruned at time $t$ by index $r$ when this maximum is greater than $Q_t + \beta$. If no change is in the data, we get the value:
$$ \frac{1}{2}\Bigg( \sqrt{\|\Delta\overline{\mathbf{S}}_{rst}\|^2} - \sqrt{\frac{r}{s}\|\Delta\overline{\mathbf{S}}_{0rs} \|^2}\Bigg)^2- \frac{1}{2}\Bigg(\sqrt{\frac{s}{t}\|\Delta\overline{\mathbf{S}}_{0st} \|^2}\Bigg)^2\,.$$
\end{proposition}

The proof is straightforward by derivation of the dual. The proof for the no-change case uses variance expressions in Appendix~\ref{app:means_relations}. For such case, we prune when:
 $$\sqrt{\frac{r}{s}}\|\Delta\overline{\mathbf{S}}_{0rs}\| - \|\Delta\overline{\mathbf{S}}_{rst}\| > \sqrt{\frac{s}{t}}\|\Delta\overline{\mathbf{S}}_{0st}\|\,.$$
With $d$ (independent) constraints, the maximum of the dual can still be obtained using basic linear algebra, provided the positivity of $\mu$ is not enforced. The true dual solution is recovered only when the maximum lies strictly within the positive orthant.

\subsection{Maximum of the non-constrained decision function}
\label{subsec:decision}

The decision function admits a closed-form expression for its critical point thanks to its simple structure. However, this point may fall outside the positive orthant. Enforcing this constraint analytically remains intractable at present, necessitating an iterative optimisation algorithm.

\begin{proposition}
\label{prop:decisionMD}
In the $d$-variate model, the maximum of the decision function with $q$ constraints -- when solved over the entire space $\mathbb{R}^q$ rather than the constrained set $\Omega_{\mathbf{x}}$ -- i.e., $\max_{\mathbf{x} \in \mathbb{R}^q} \mathbb{D}(\mathbf{x})$, corresponds to the solution of the following system of equations, involving $\mathbf{x} \in \mathbb{R}^q$ and an instrumental variable $\mathbf{y} \in \mathbb{R}^d$:
$$
\left\{
\begin{aligned}
&(\Delta \overline{\mathbf{S}}_{\bullet st}) \, \mathbf{x}  = \nabla A(\mathbf{y}) - \overline{\mathbf{S}}_{st}\,, \\
&(\Delta \overline{\mathbf{S}}_{\bullet st})^T \, \mathbf{y} = - \Delta \overline{Q}_{\bullet st}\,,
\end{aligned}
\right.
$$
where $(\Delta \overline{\mathbf{S}}_{\bullet st})$ is the matrix of size $d \times |\mathcal{R}|$ gathering all vectors $\Delta \overline{\mathbf{S}}_{rst}$ in columns and $\Delta \overline{Q}_{\bullet st} \in \R^{|\mathcal{R}|}$. If $q=d$ and the matrix $(\Delta \overline{\mathbf{S}}_{\bullet st})$ is invertible, we get the solution:
\begin{equation}
\label{eq:xDecision}
\mathbf{x}^{\star} = (\Delta \overline{\mathbf{S}}_{\bullet st})^{-1} \Bigg(\nabla A \Big(-((\Delta \overline{\mathbf{S}}_{\bullet st})^T)^{-1}\Delta \overline{Q}_{\bullet st}\Big) - \overline{\mathbf{S}}_{st} \Bigg)\,.
\end{equation}
\end{proposition}

The explicit solution~\eqref{eq:xDecision} is promising, but it fails to incorporate constraint inequalities, particularly the simple positivity constraint on $\mathbf{x}$. Relation~\eqref{eq:xDecision} generalises the simpler result of Equation~\eqref{eq:DUSTexact1D1C}.

\subsection{General decision function}

The form of the decision function naturally leads to the introduction of a more general function that encapsulates all duality tests across all indices $s$. This function can take the following general form:
\begin{equation}
\label{eq:generalDual}
\mathbb{GD}(\mathbf{z}) 
= -\D^*\Big(\sum_{i = 1}^{|\Tau_t|-1}z_i\overline{\mathbf{S}}_{s_is_{i+1}} \Big) - \Big(\sum_{i = 1}^{|\Tau_t|-1}z_i\overline{Q}_{s_is_{i+1}}\Big)\,,
\end{equation}
with $s_i \in \Tau_t$ sorted in increasing order  with $s_{|\Tau_t|} = t$ and constraint $\sum_{i = 1}^{|\Tau_t|-1}z_{i} = 1$. We recall that $\Tau_t$ is the set of non-pruned indices at time~$t$.

The case where the vector $\mathbf{z}$ is equal to $\overline{\mathbf{z}}(s_{uv})$, defined as
$$\overline{\mathbf{z}}(s_{uv}) = \Big(0,\ldots, 0, \frac{s_{u+1} - s_{u}}{s_v-s_{u}}, \frac{s_{u+2} - s_{u+1}}{s_v-s_{u}},\ldots, \frac{s_v-s_{v-1}}{s_v-s_u},0,\ldots, 0\Big)\,,$$
corresponds to the standard PELT rule for pruning index $s_u$ at time $s_v$ as we get $-\D^*(\overline{\mathbf{S}}_{s_u s_v}) - \overline{Q}_{s_u s_v}$. The single-constraint case of Section~\ref{sec:simple} with $r<s$ for pruning at time $t$ is obtained with $\mathbf{z} = (1+x)\overline{\mathbf{z}}(s_{st}) - x\overline{\mathbf{z}}(s_{rs})$. We believe that the general formulation~\eqref{eq:generalDual}, which unifies all duality problems into a single function, has the potential to further improve pruning efficiency. This is left for future work.

\section{Simulation study}
\label{sec:simus}

This section presents DUST's performance on simulated data and compares its efficiency against state-of-the-art algorithms. {\color{violet} (This first draft version does not yet include the study of multivariate signals.)}

\subsection{Framework}

\subsubsection{Data}

We simulate time series of length $\exp(n)$, where $n$ is a sequence of regularly spaced values between $\log 10^2$ and $\log 10^8$. We have $8$ available models from the exponential family: Gauss (G), Poisson (P), Exponential (E), Geometric (G), Bernoulli (Be), Binomial (Bi), Negative Binomial (NB), and Variance (V). 
The simulated data consists of alternating two segments of length $k$. We often set $k=n$, ensuring that only the first segment is observed. This no-change regime presents the most challenging scenario for pruning and computational efficiency. The typical parameters used are outlined in Table~\ref{tab:params}.\\

\begin{table}[!ht]
\caption{Parameter values for simulations\label{tab:params}}
\begin{tabular*}{\columnwidth}{@{\extracolsep\fill}llll@{\extracolsep\fill}}
Model & Penalty scale factor &  Parameter & Values  \\
 \midrule
gauss & 1 & $\mu$ ($\sigma = 1$ fixed) & $\{0,1\}$   \\
poisson & $2/3$ & $\lambda$ & $\{3,4\}$   \\
exponential & $3/4$ & $\lambda$ & $\{1,0.5\}$   \\
geometric & $2/3$ & $p$  & $\{0.5,0.7\}$   \\
bernoulli & $2/3$ & $p$  & $\{0.5,0.7\}$   \\
binomial & $1/6$ & $p$  & $\{0.5,0.7\}$\\ 
negative binomial & 1/10 & $p$ & $\{0.5,0.7\}$  \\
variance & 1 & $\sigma$ ($\mu = 0$ fixed) & $\{1,2\}$   \\
 \midrule
\end{tabular*}
\end{table}

When considering d-variate time series, we generate $d$ time series with identical change-point locations using the previous model and concatenate the $d$ copies. Further, we study three different configurations. (i) (with $k = n$) Pruning and time capacity of DUST over time and against competitors. (ii) (with $k = n$) Pruning and time robustness of DUST with respect to the penalty value at fixed data length (iii) Pruning and time robustness in the presence of changes (against competitors). For the multivariate setting, we also study these configurations with respect to the dimension.

We run each configuration $50$ to $100$ times at fixed data length depending on computational cost.\\

\subsubsection{Baseline and parameters}  We compare DUST to the PELT\footnote{\url{https://cran.r-project.org/web/packages/changepoint/index.html}} \cite{killick2012optimal} and FPOP\footnote{\url{https://cran.r-project.org/web/packages/fpopw/index.html}} \cite{maidstone2017optimal} algorithms for univariate data. In the multivariate setting, there is only one competitor and only for the Gaussian cost: GeomFPOP\footnote{\url{https://github.com/computorg/published-202406-pishchagina-change-point}} \cite{pishchagina2024geometricbased}. Each algorithm is provided with the same simulated data to ensure proper head-to-head comparisons.\\

The penalty factor for each model is calibrated to achieve similar segmentation behaviour across the range of penalty for data with no change point. With the Gaussian model as a baseline, we apply a scaling factor to the penalty for each model. This scaling factor is determined by finding the smallest penalty value where change-point detection correctly produces no change point on data of length $10^3$. The ratio between this value and the corresponding value in the Gaussian model gives us the scaling factor. For example, suppose we execute DUST on Gaussian data of length $n$ with penalty $2 \log n$. In that case, an equivalent Negative binomial simulation is done with penalty $2 a \log n$ with $a = 1/10$ (see Table~\ref{tab:params}).

\subsubsection{Metrics} 
We measure the execution time for parsing each simulated time series to compare the computational cost of each algorithm in each configuration. We further record the number of candidate indices during the execution of each DUST. The remaining number of indices at the end of the algorithm is a measure of the algorithm's pruning capacity.\\

\subsubsection{DUST variants}

DUST algorithm comes with variants depending on the dual evaluation algorithm and index selection method used.

The dual evaluation can be made:
\begin{enumerate}
\item at its maximum using the closed formula (for one-parametric cost functions only);
\item at a random point uniformly drawn in the dual domain;
\item at zero (equivalent to PELT test);
\item at its maximum using the Quasi-Newton algorithm.
\end{enumerate}
Evaluation in the one-parameter-cost case is naturally performed using a closed-form formula. Other methods are tested for multi-parameter cases to balance time complexity with pruning efficiency. Index selection for building the dual comprises two parts: indices below the index~$s$, and indices above. Indices can be chosen randomly (uniformly) or deterministically (the largest available below $s$ and the smallest available after $s$). Here, we consider only the deterministic case.

\subsection{Univariate signals}

\subsubsection{Pruning capacity}

Figure~\ref{fig:nb_plot} displays the number of non-pruned indices over time for one time series of length $10^4$ (left column) and length $10^8$ (right column) and the median and confidence intervals (over $100$ repetitions). We present the results for two cost models: Gauss, top row; Negbin, bottom row.  DUST shows a persistent pruning efficiency that is robust to the number of data points, with low variations being recorded. With $n = 10^8$, the number of remaining indices is $50$ with the exact evaluation method, and $500$ with the random method. 

\begin{figure}[ht!]
   \centering
\includegraphics[width=0.4\linewidth]{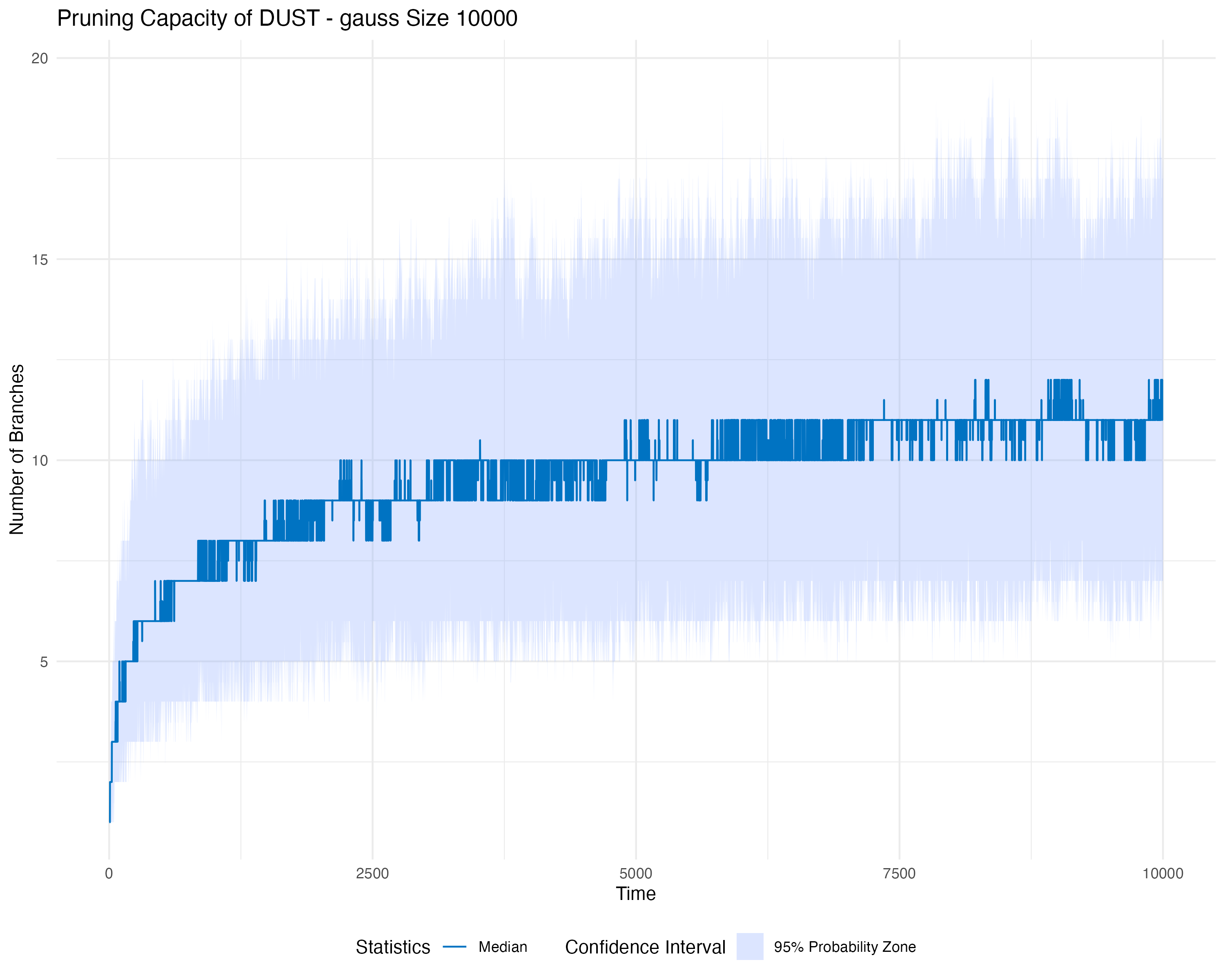}
\includegraphics[width=0.4\linewidth]{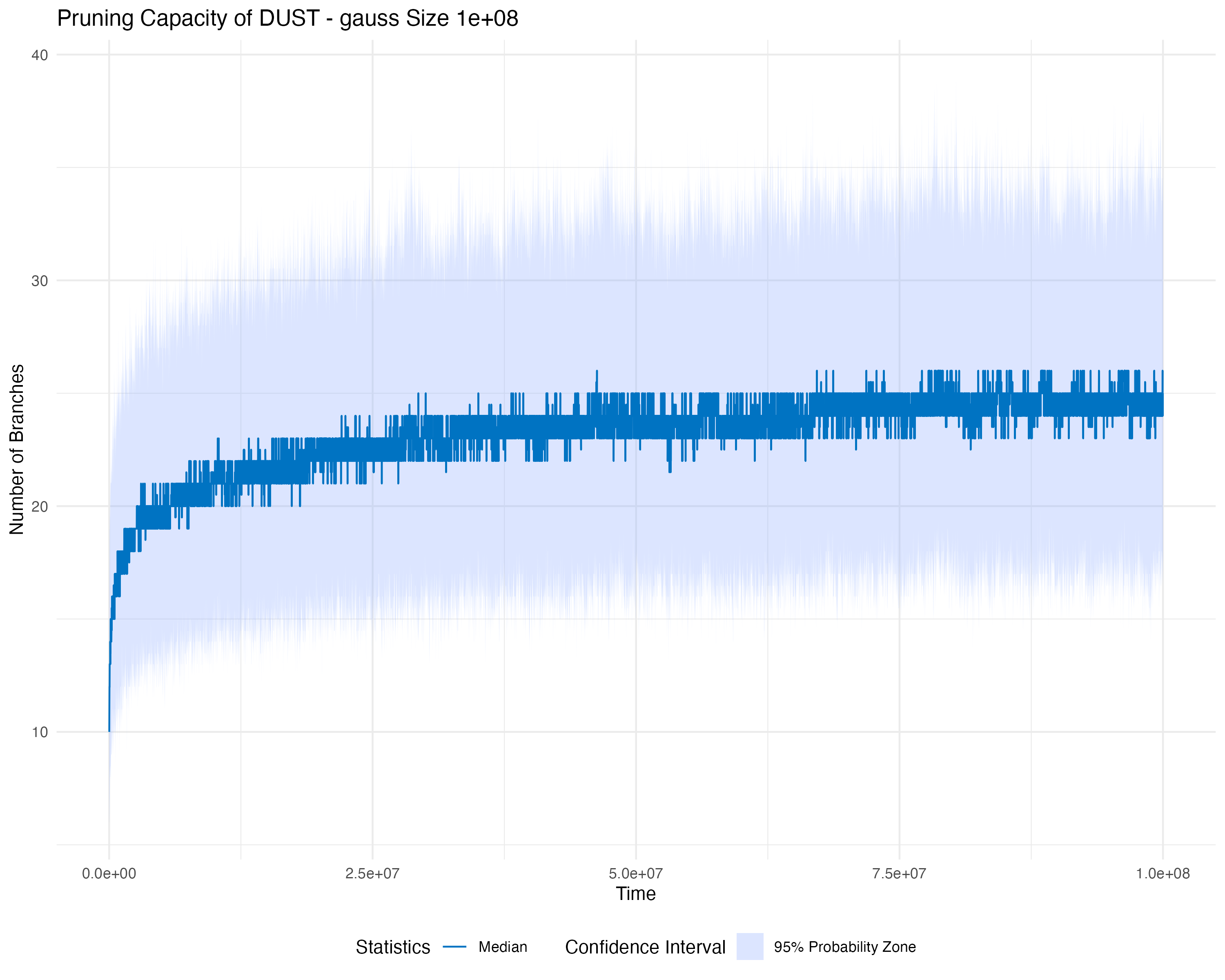}
    \newline
\includegraphics[width=0.4\linewidth]{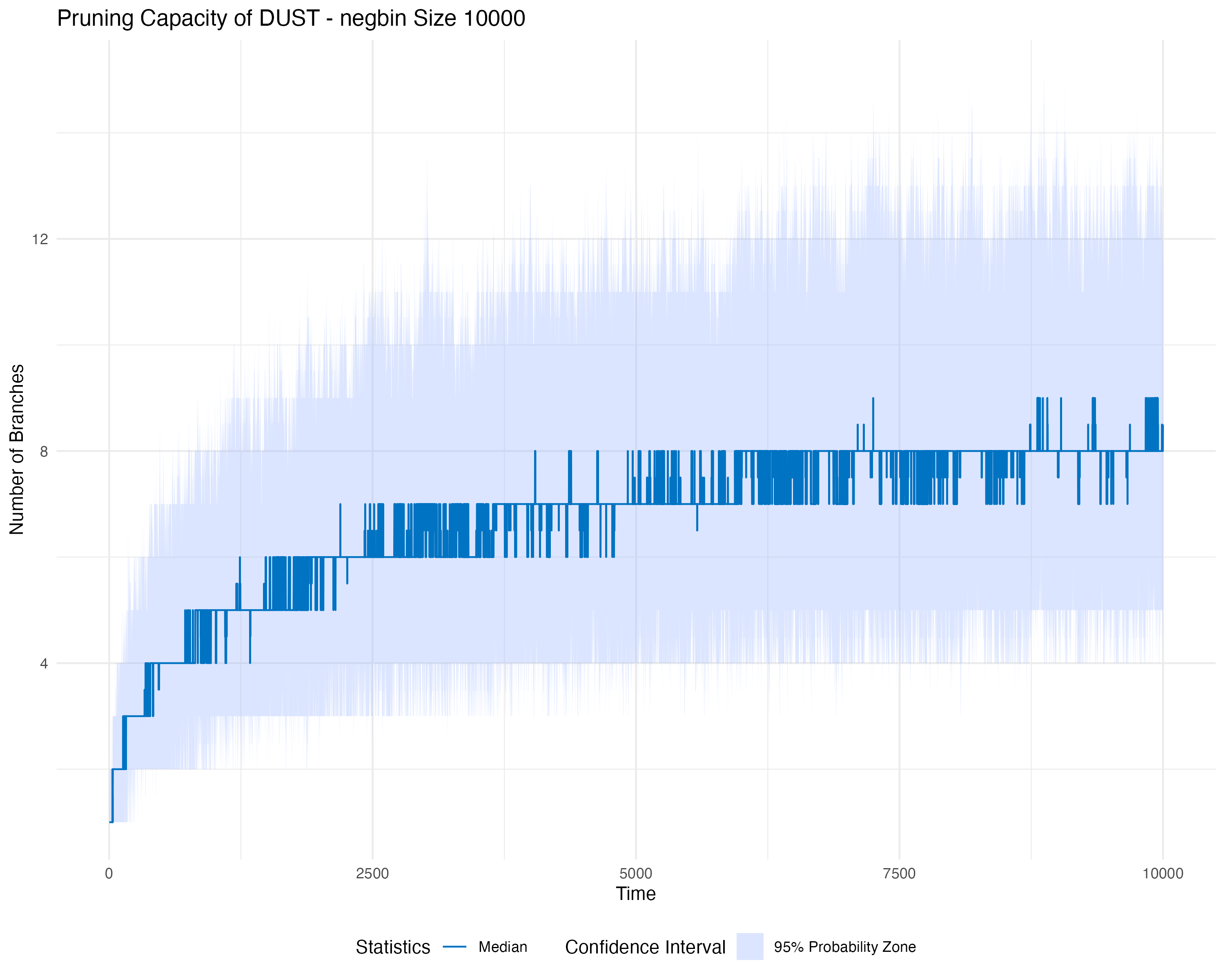}
\includegraphics[width=0.4\linewidth]{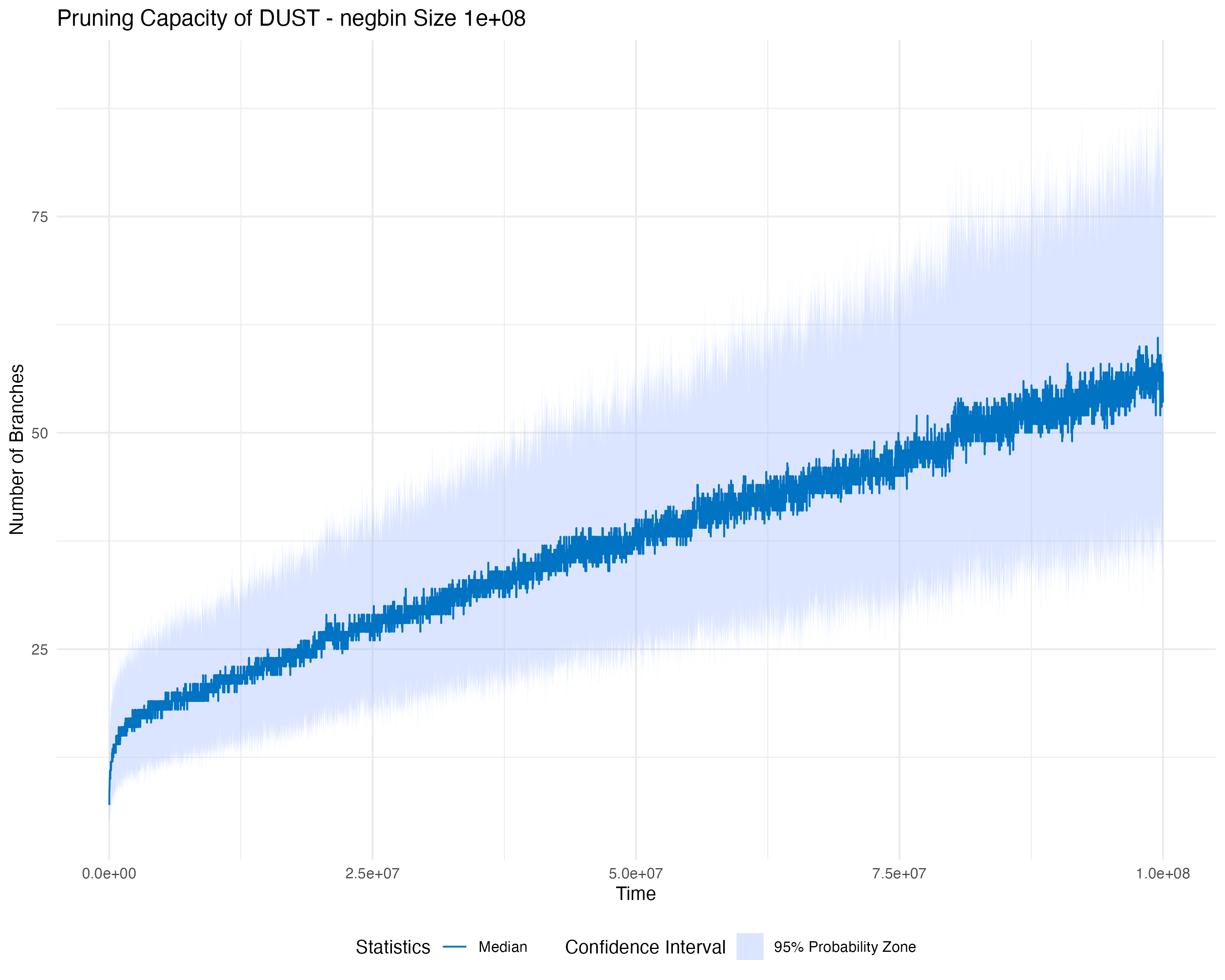}
\caption{Median number of remaining candidate indices at each time for $100$ DUST simulations for the Gaussian (top row) and negative binomial (bottom row) models, data lengths $10^4$ (left column) and $10^8$ (right column). The shaded region shows the interval between the $0.025 \, \%$ and $0.975 \, \%$ quantiles.}
    \label{fig:nb_plot}
\end{figure}

We plot the number of remaining candidate indices for different data lengths on a log-log scale in Figure~\ref{fig:nb_reg}. We run $100$ simulations for each data length among $100$ regularly spaced (in the logarithmic scale) lengths between $10^2$ and $10^6$. The shaded region shows the interval between the $0.025 \, \%$ and $0.975 \, \%$ quantiles. The solid line shows the fitted linear regression with a $95\%$ prediction confidence interval as dotted lines on either side. The slope value of either model suggests that the mean number of candidate indices remaining upon exit of the DUST algorithm is of order $n^{\alpha}$, $\alpha < 0.15$, which indicates a time complexity of order $\mathcal{O}(n^{1 + \alpha})$, with minor variations in the coefficient depending on the model being considered. 

\begin{figure}[ht!]
    \centering
\includegraphics[width=0.49\linewidth]{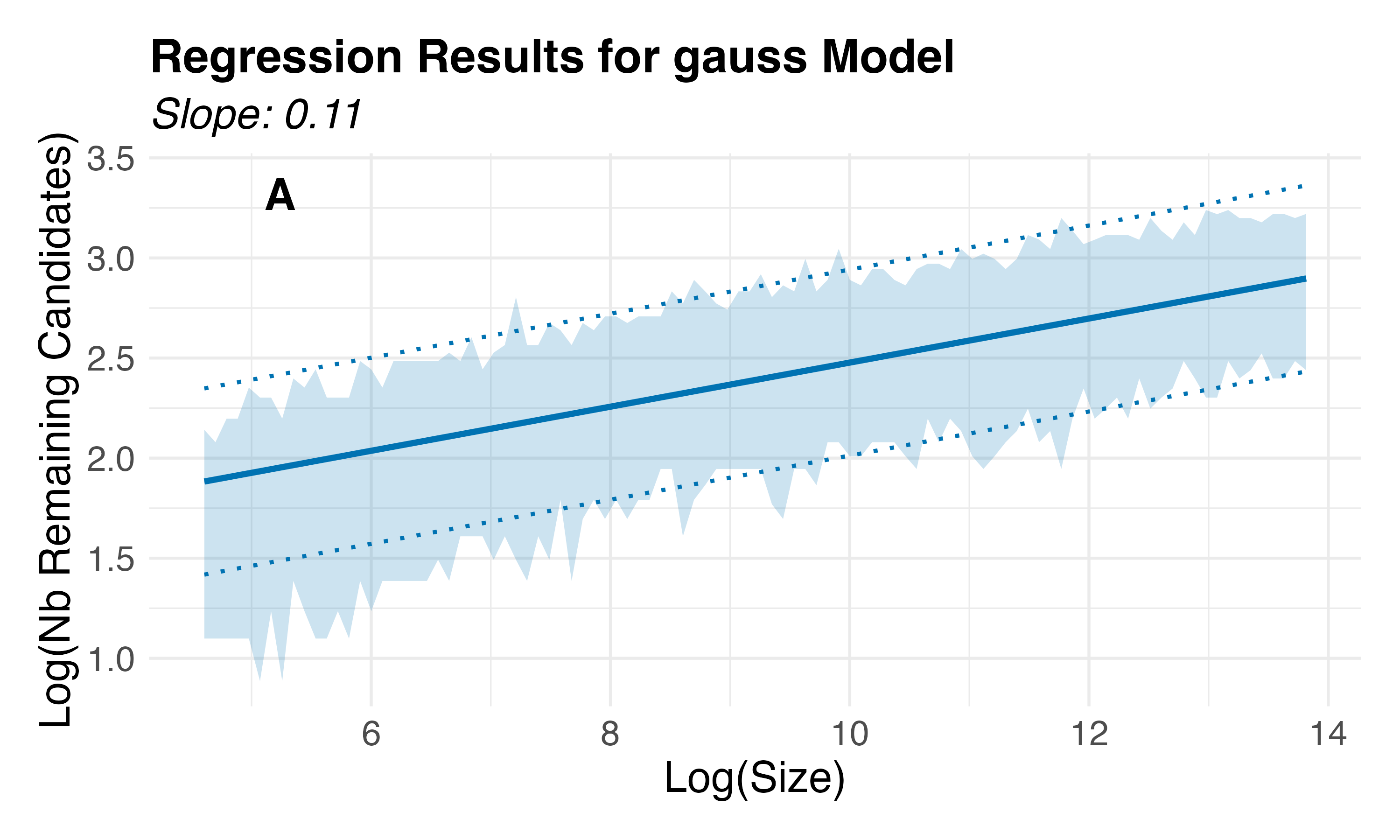}
\includegraphics[width=0.49\linewidth]{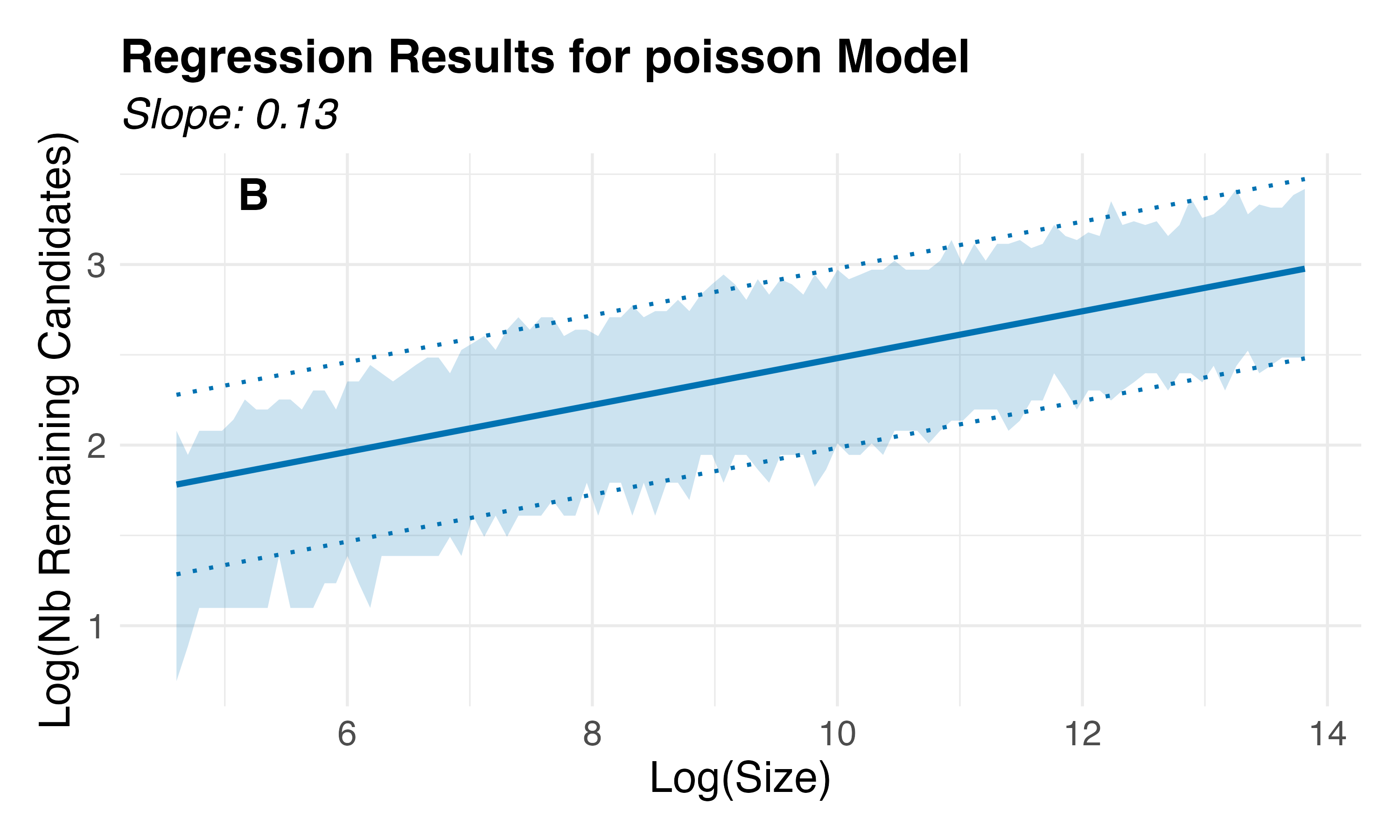}
\caption{Log-log comparison of number of remaining indices over times as a function of data size for the Gaussian (panel A) and Poisson (panel B) models. We run $100$ simulations on $100$ data lengths between $10^2$ and $10^6$, regularly spaced in the log-scale. The shaded region shows the interval between the $0.025 \, \%$ and $0.975 \, \%$ quantiles. The solid line shows the fitted linear regression with a $95\%$ prediction confidence interval as dotted lines on either side. All simulations are performed on data without a change point. }
\label{fig:nb_reg}
\end{figure}



\subsubsection{Time competition}

Figure~\ref{fig:time_reg} displays a log-log comparison of execution times between DUST and FPOP as a function of data size for the Gaussian (panel A) and Poisson (panel B) models. $100$ simulations were performed on $100$ data lengths between $10^2$ and $10^6$, regularly spaced in the log-scale. The shaded regions show the interval between each algorithm's $0.025 \, \%$ and $0.975 \, \%$ quantiles. The solid lines show the fitted linear regression for each algorithm for sizes $3125$ through $10^6$, with $95\%$ prediction confidence intervals as dotted lines on either side. The threshold at $3125$ was introduced as the relation between execution time and data length stabilizes by that point. Analysis of the interaction coefficient between data length and algorithm used shows that FPOP's slope is significantly greater than DUST's under either model, which implies that DUST's execution time scales better on larger data. Under the Gaussian model, DUST runs much slower on smaller data sizes, but DUST reliably beats FPOP beyond a break-even point at around $5,000$, and is on average $1.19$ times faster on data of length $10^6$. On non-Gaussian models such as the Poisson model, however, the DUST algorithm outperforms FPOP even at the smallest sizes, achieving speeds $5.88$ times up to $8.15$ times greater than FPOP, at sizes $10^2$ and $10^6$ respectively. Comparison of execution time is presented in 
Figure~\ref{fig:timeExecution} for relatively small data length.

\begin{figure}[ht!]
    \centering
\includegraphics[width=0.49\linewidth]{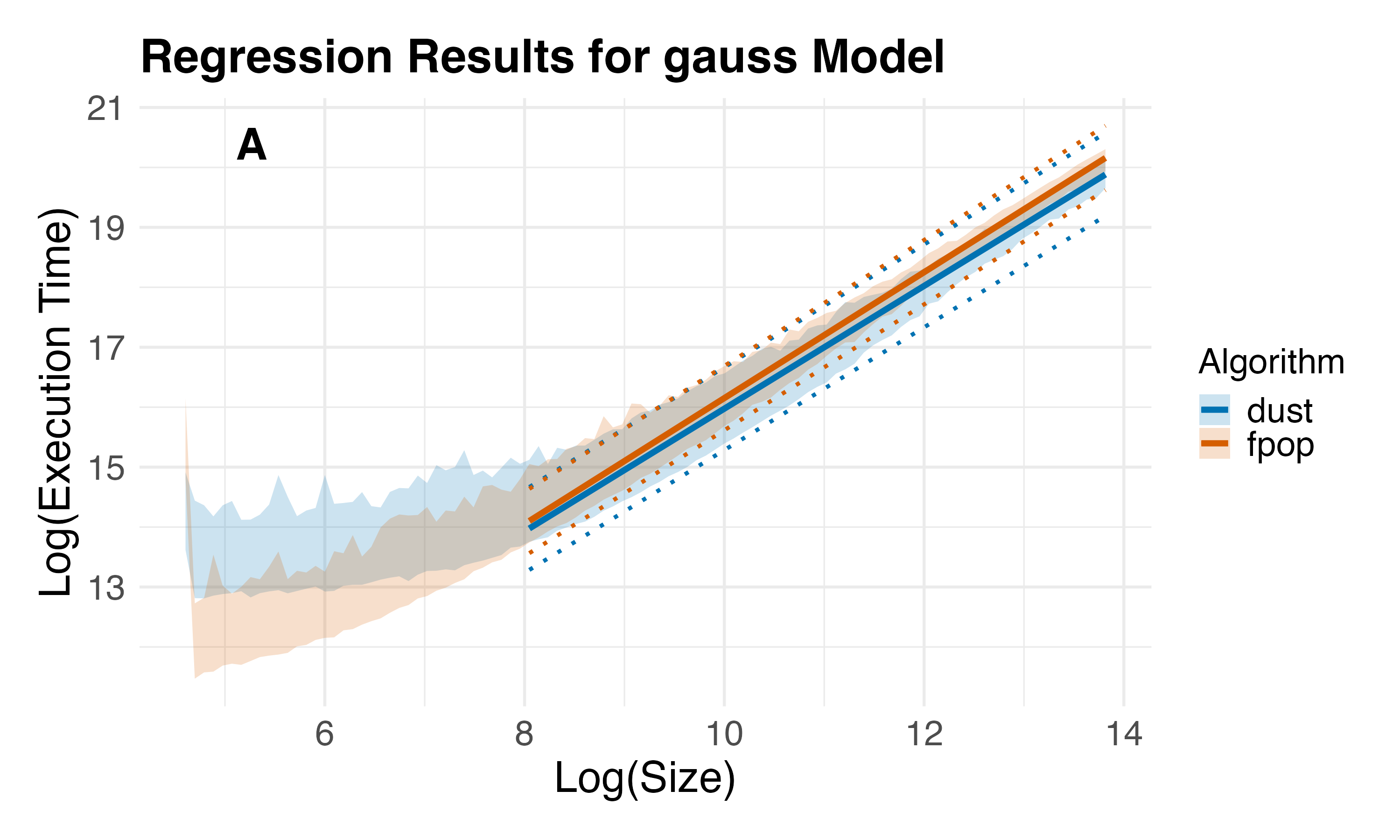}
\includegraphics[width=0.49\linewidth]{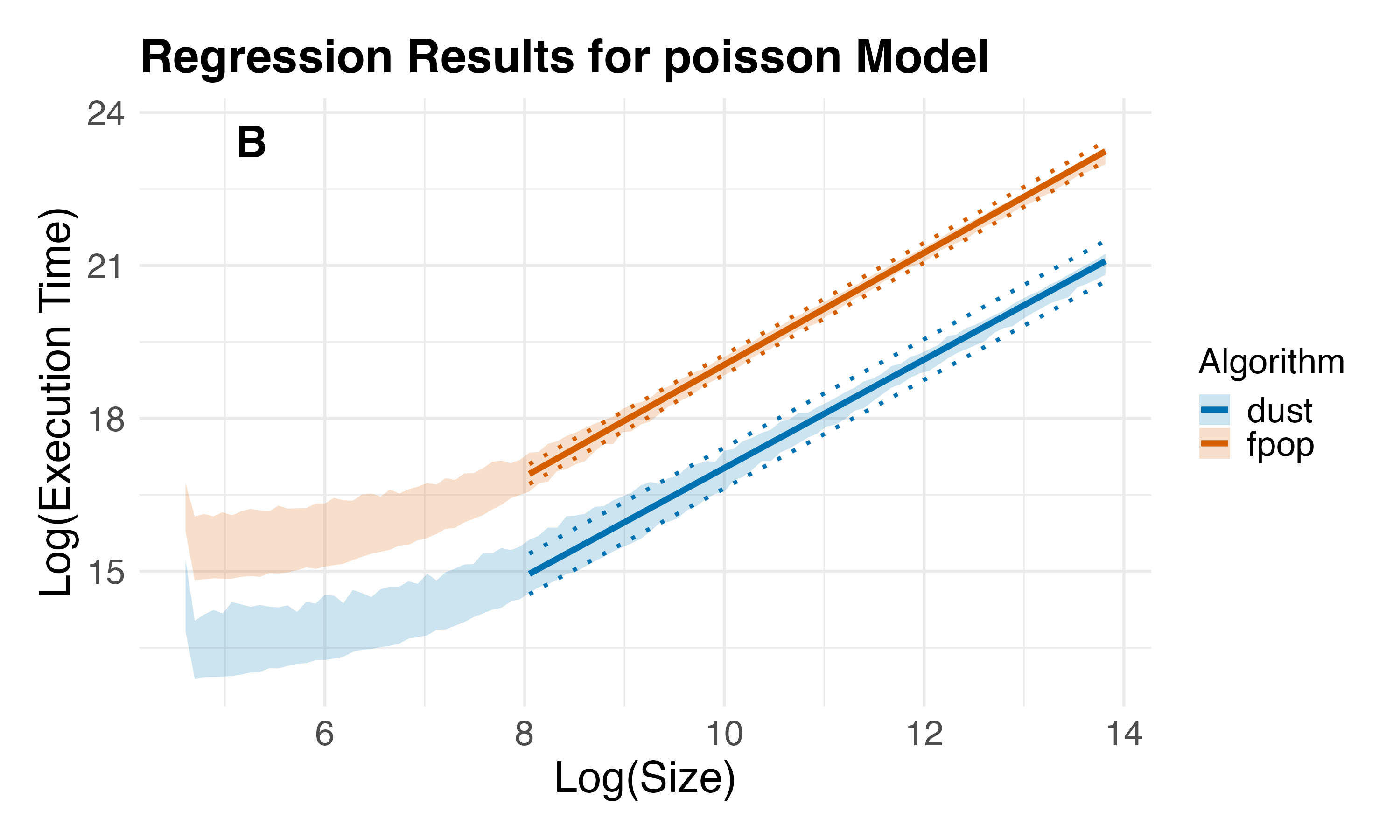}
\caption{Log-log comparison of execution times between DUST and FPOP as a function of data size for the Gaussian (panel A) and Poisson (panel B) models. $100$ simulations were performed on $10^2$ data lengths between $100$ and $10^6$, regularly spaced in the log-scale. The shaded regions show the interval between the $0.025 \, \%$ and $0.975 \, \%$ quantiles. The solid lines show the fitted linear regression for each algorithm for sizes $3125$ through $10^6$, with $95\%$ prediction confidence intervals as dotted lines on either side. All simulations are performed on data without a change point. }
    \label{fig:time_reg}
\end{figure}

\begin{figure}[ht!]
    \centering
\includegraphics[width=0.40\linewidth]{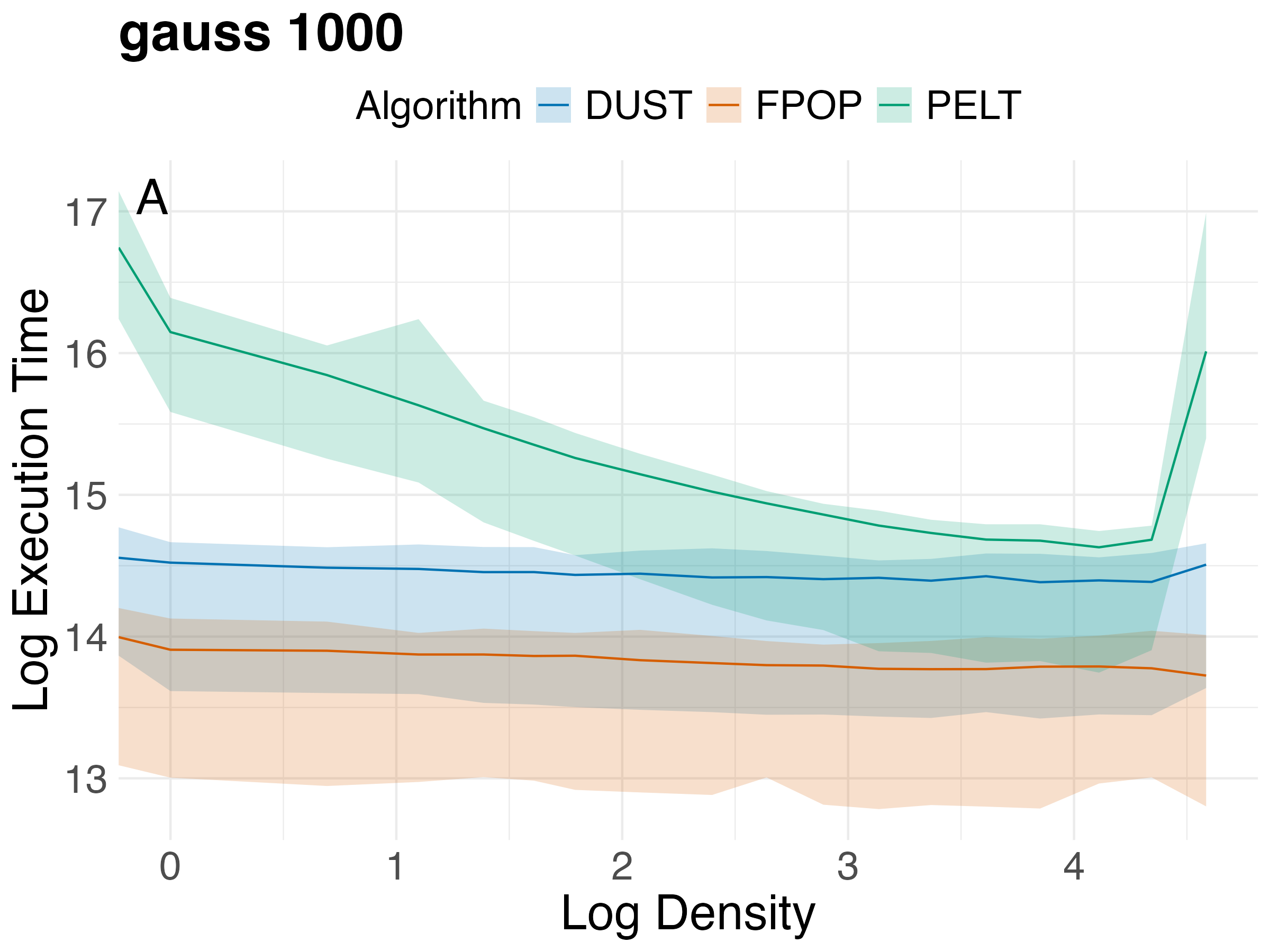}
\includegraphics[width=0.40\linewidth]{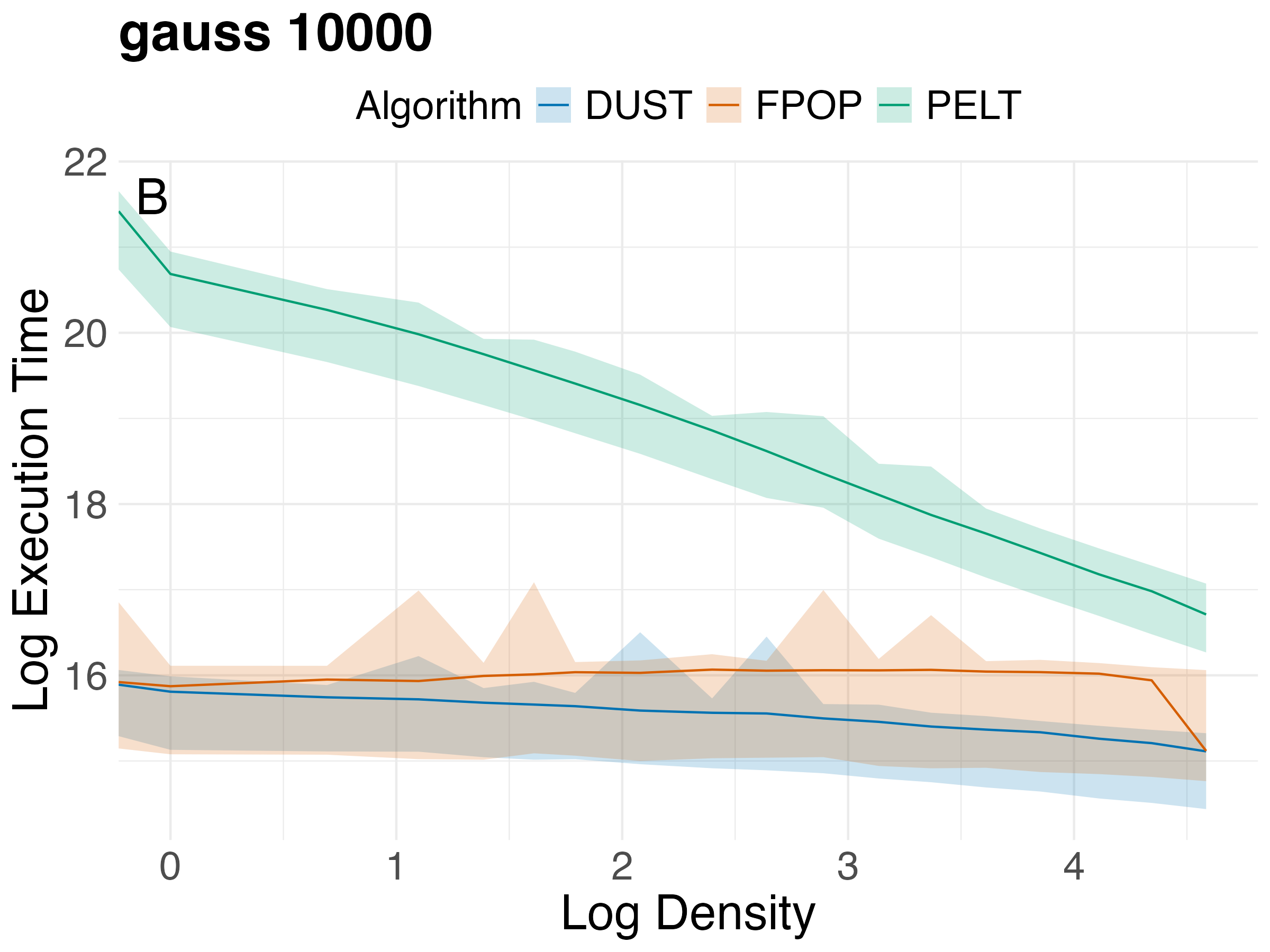} \includegraphics[width=0.40\linewidth]{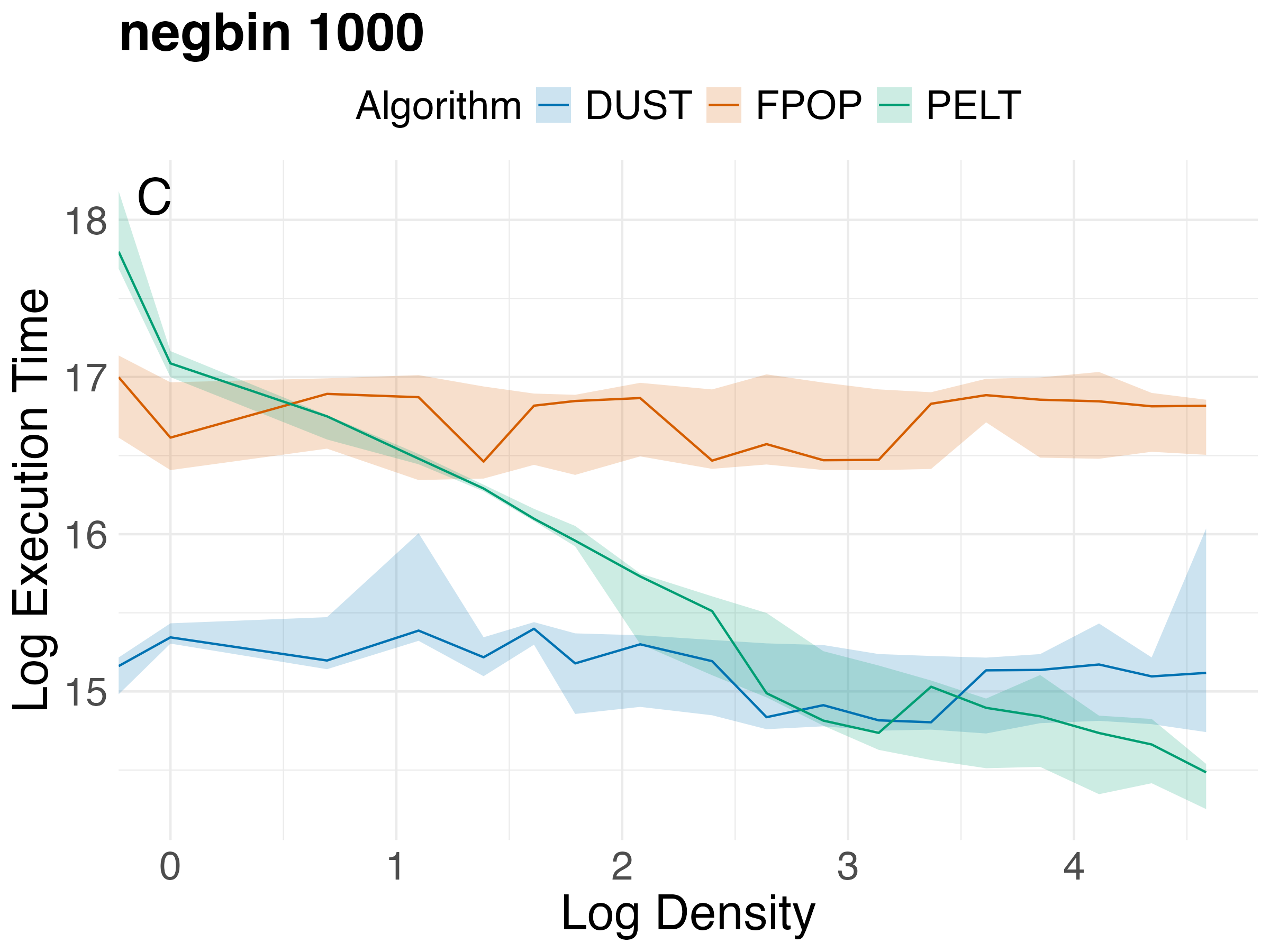}  \includegraphics[width=0.40\linewidth]{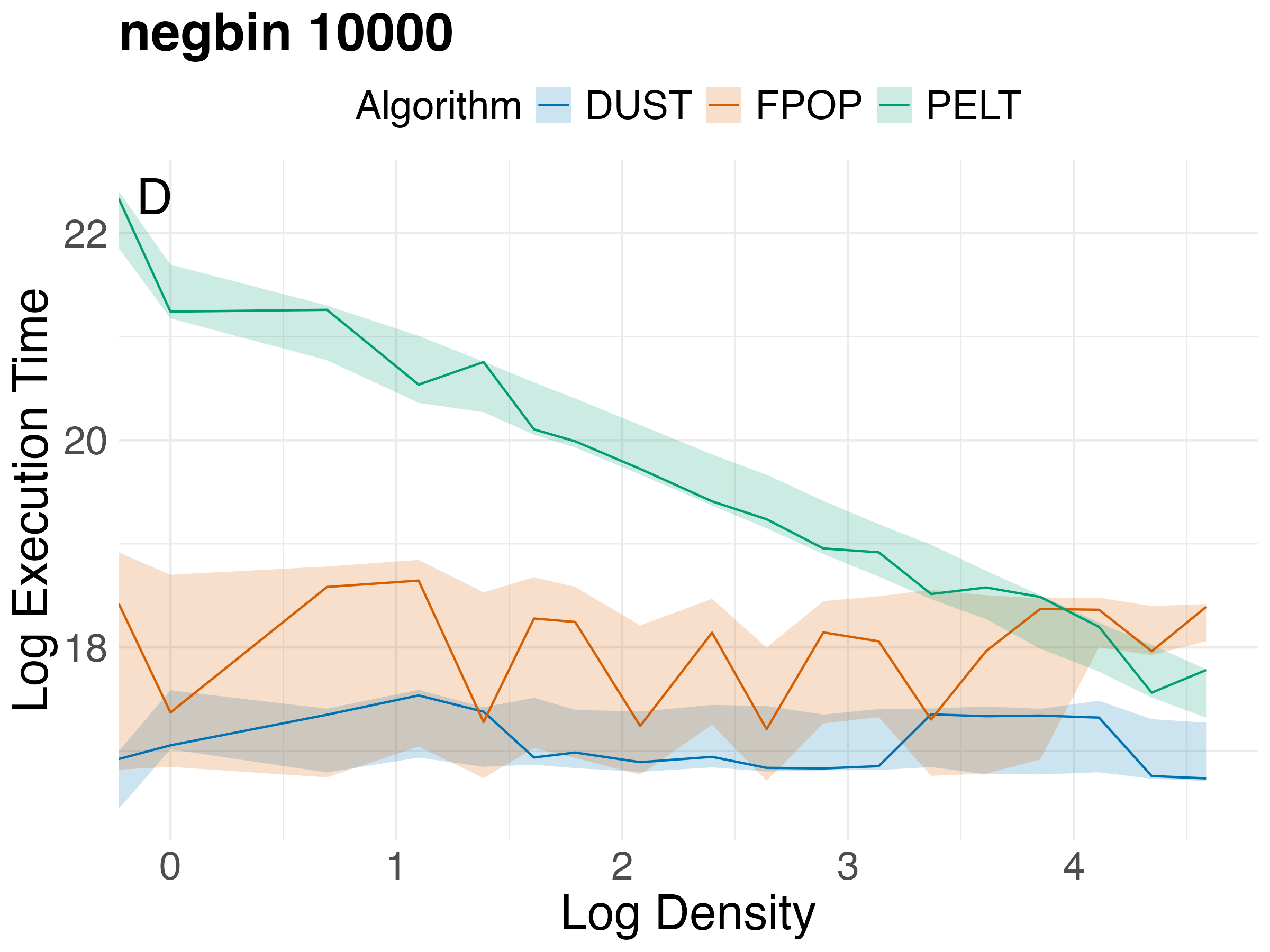}
\caption{Log-log comparison of execution times for varying numbers of change points in input data of lengths $10^3$ (panels A, C) and $10^4$ (panels B, D) under the Gaussian (panels A, B) and the Negative binomial (panels C, D) models.}
\label{fig:timeExecution}
\end{figure}

\subsubsection{Pruning exploration}

Figure~\ref{fig:nb_beta} displays the median number of candidate indices upon exit of the DUST algorithm for $100$ executions on fixed-length data with no change point for $100$ different penalty factors, under the Gaussian (panel A) and Poisson (panel B) models. Data length is $n = 10^7$ and penalty factors are of the form $\beta = a\log n$, with $100$ different $a$ values regularly spaced between $0.001$ and $20$ in the logarithmic scale. The shaded region shows the interval between the $0.025 \, \%$ and $0.975 \, \%$ quantiles. Lower penalty values produce a high number of change points. Values beyond $0.75$ predominantly produce the correct number of change points. As for the number of remaining candidates, we observe a polynomial relation with the logarithmic values of $a$ up to $a = 1$, and a negative linear relation beyond that point, with a maximum median number of candidates of $24$ under the Gaussian model, and $28$ under the Poisson model. This suggests that pruning is strong regardless of the penalty chosen, which guarantees a strongly reduced computational cost even on non-normalised data or with a poorly tuned penalty factor.

\begin{figure}[ht!]
    \centering
\includegraphics[width=0.49\linewidth]{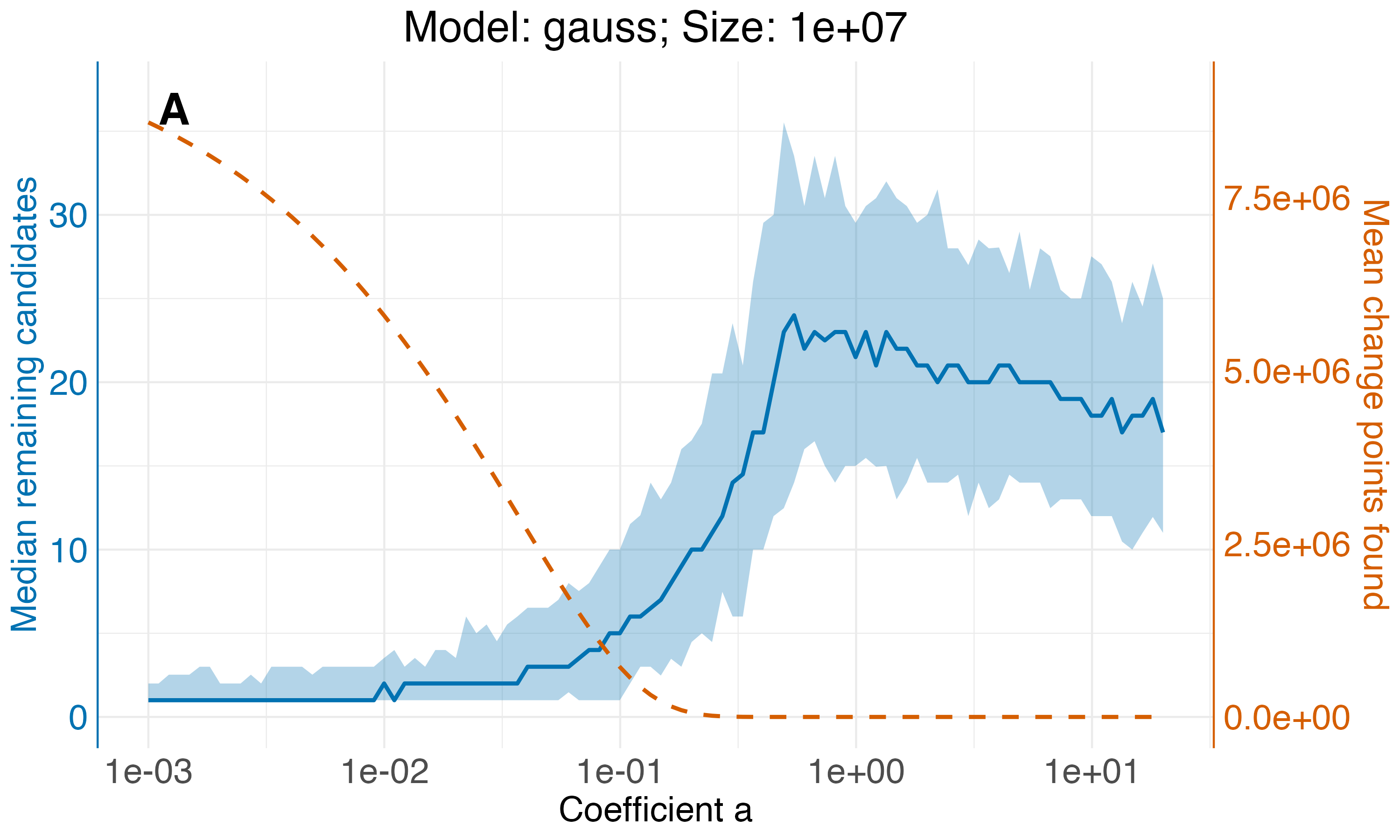}
\includegraphics[width=0.49\linewidth]{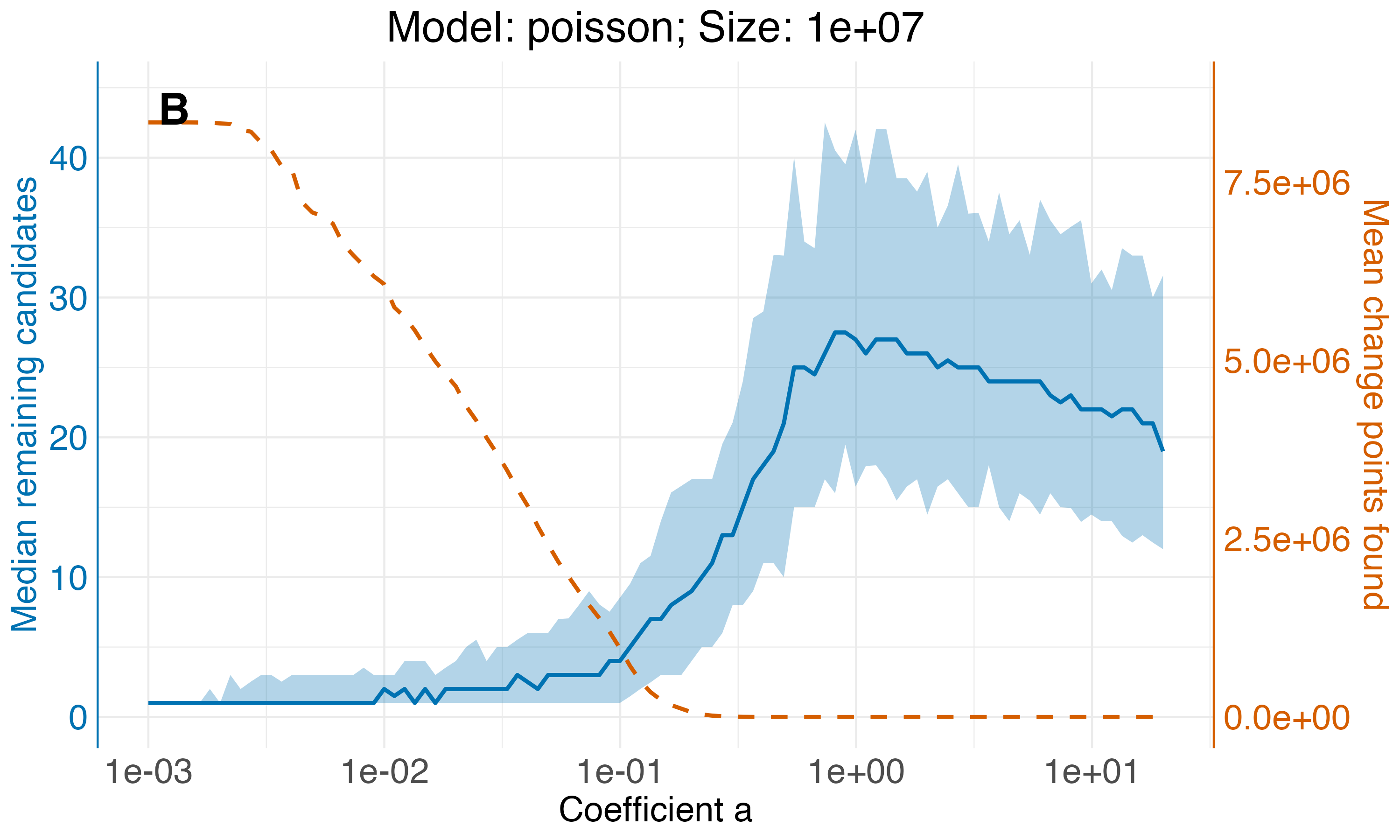}
\caption{Median number of candidate indices upon exit of the DUST algorithm for $100$ executions on fixed-length data with no change point for $100$ different penalty factors, under the Gaussian (panel A) and Poisson (panel B) models. Data length is $n = 10^7$ and Penalty factors are of the form $\beta = a_i\log n$, with $100$ different $a$ values regularly spaced between $0.001$ and $20$ in the logarithmic scale. The shaded region shows the interval between the $0.025 \, \%$ and $0.975 \, \%$ quantiles.}
    \label{fig:nb_beta}
\end{figure}

\subsection{Multivariate signals}

{\textcolor{violet}This section will be updated post-publication.}

\subsection{The DUST package}

DUST package is available on GitHub\footnote{\url{https://github.com/vrunge/dust.git}}.

\section{Application to mouse monitoring}
\label{sec:application}

This section explores the use of DUST on a real-world application taken from~\cite{Krejci2024}.

\paragraph{Context}
The neuromuscular junction (NMJ) is the synapse responsible for the chemical transmission of electrical impulses from motor neurons to muscle cells.
Studying the NMJ is crucial because it explains how nerves communicate with muscles to move. Since this synapse is a disease-prone synapse~\cite{RodrguezCruz2020}, studying the NMJ also helps understand and treat disorders like myasthenia gravis and congenital myasthenic syndrome.

In~\cite{Krejci2024}, the authors find a correlation between a partial inhibition of the NMJ and the level of muscle fatigue in mice. They measure muscle fatigue by quantifying the number and duration of active periods over several days. 

\paragraph{Data}
A typical approach to estimate resting and active states is to use a force platform, a scale measuring the ground reaction force generated by the mouse in its cage.
To compare mice, \cite{Krejci2024} uses several features, including the total duration of activity, the mean duration of activity, and the number of active periods. 
The procedure to detect whether a mouse is active or inactive consists of several steps, including filtering and simple signal transformations. Changes are detected on a simplified signal because the signals have too many samples. 

Here we consider two mice with two different genetic modifications (Mouse ColQ and Mouse A7), monitored over 2 nights and 1 day.
The cage is placed on four force sensors, and we record the evolution over time of the sum of the ground reaction forces at $10$ Hz.
If the sum of forces is constant (low variance), the mouse is not active (resting state).
If the sum varies, the mouse is active, regardless of the activity.

\paragraph{Proposed approach and results}
We propose a more straightforward approach: we use DUST on raw data. Each segment is then classified as active or inactive by thresholding the variance. Since DUST is fast, we can process 12 hours ($\sim 400$k samples) of time series in seconds without much preprocessing. The statistical model is piecewise Gaussian with a fixed mean equal to 0 and piecewise constant variance. 

From a qualitative standpoint, Figure~\ref{fig:mouse-activity} shows parts of the time series and the segmentation into active/resting stages.
Activity is defined by a high variance.

From Table~\ref{tab:mouse-activity}, which summarises with simple features the level of activity of each mouse, we draw two conclusions.
First, mice are more active during the night by a margin.
They have longer stretches of activity and spend more time in an active state. 
We recover here a well-known fact in mouse behaviour analysis.
Second, Mouse A7 is more active than Mouse ColQ during the night, but equally active during the day.
This observation is in line with the results of~\cite{Krejci2024}, albeit with a smaller number of mice.
However, our approach is simpler than that of~\cite{Krejci2024}, thanks to its ability to detect variance changes in long time series within seconds.

\begin{figure}
    \centering
    \includegraphics[width=0.7\textwidth]{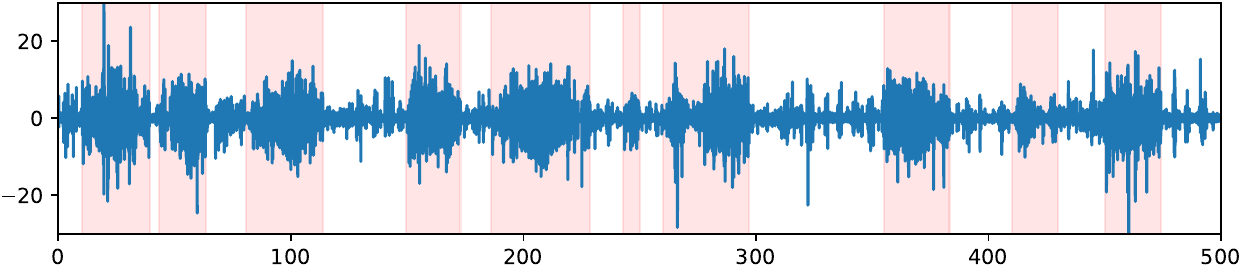}
    \centerline{(a) First night}
    \vskip1em
    \includegraphics[width=0.7\textwidth]{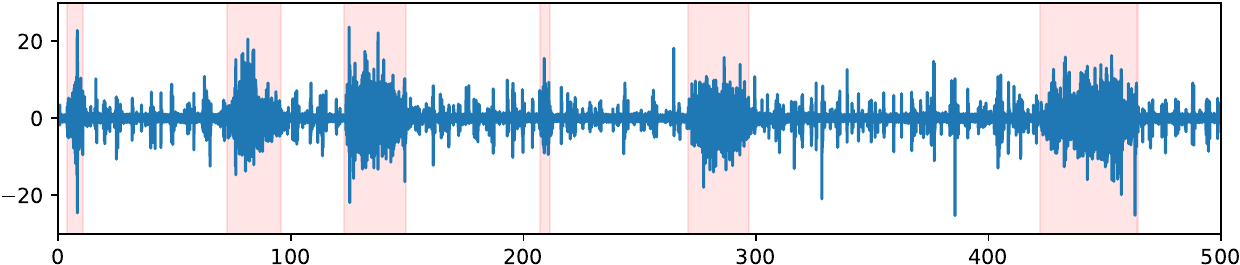}
    \centerline{(b) First day}
    \vskip1em
    \includegraphics[width=0.7\textwidth]{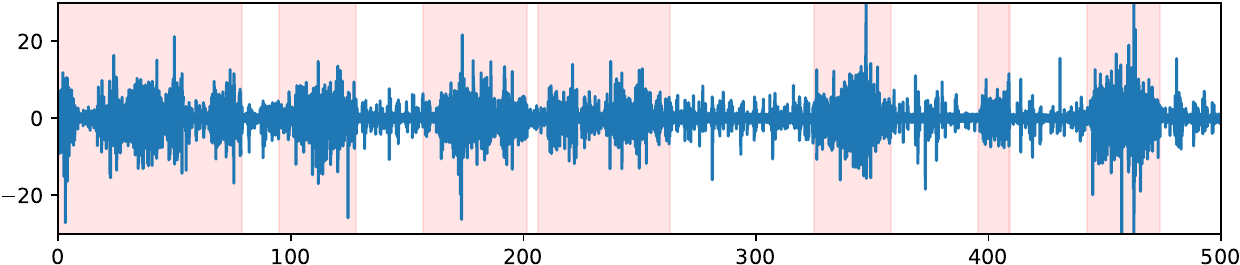}
    \centerline{(c) Second night}
    \vskip1em
    \caption{For Mouse ColQ, the activity measured by a force platform. The x-axis is the time (in minutes), and the y-axis is the ground reaction forces generated by a mouse in its cage (in arbitrary units).
    Red areas indicate activity periods found by our change-point approach; others are rest periods.}
    \label{fig:mouse-activity}
\end{figure}

\begin{table}[t!]
    \centering
    \footnotesize
    \caption{Summary of activity segments for each period.}
    \vskip1em
    \begin{tabular}{ll|ccc}
        \textbf{Mouse} & \textbf{Period} & \# of active segments & Avg. duration (min) & Total active time (h) \\
        \hline
        \multirow{3}{*}{ColQ} 
            & First night     & 18 & 21.22 & 6.37 \\
            & First day       & 21 & 11.03 & 3.86 \\
            & Second night    & 16 & 24.84 & 6.62 \\
        \hline
        \multirow{3}{*}{A7} 
            & First night     & 20 & 25.82 & 8.61 \\
            & First day       & 21 & 9.03  & 3.16 \\
            & Second night    & 15 & 35.73 & 8.93 \\
    \end{tabular}
    \label{tab:mouse-activity}
\end{table}

\newpage 

\section{Conclusion}
\label{sec:conclusion}

A large class of change-point problems can be solved exactly using dynamic programming algorithms with functional pruning rules. Here, pruning refers to identifying the functions $q_t^i(\cdot)$ that are dominated by the minimum operator $\min_j \{ q_t^j(\cdot) \}$ in problems with a functional structure that preserves pointwise ordering over time. Specifically, if there exist indices $i, j$ and a parameter $\ttheta$ such that $q_t^i(\ttheta) < q_t^j(\ttheta)$, this inequality remains valid for all $u > t$.\\
In the DUST framework, pruning is reformulated as a (generally non-convex) optimisation problem, along with its relaxations obtained by removing constraints, which we analyse via their dual and decision representations. Pruning occurs when the computed value exceeds a specified threshold. We prove that there is no duality gap when the number of constraints does not exceed the dimension of the parameter space of the cost function derived from the exponential family.\\

We focused on the one-parameter case, essentially corresponding to univariate data, for which an inequality-based pruning test as simple as PELT can be designed. Unlike PELT, however, this test remains highly efficient across all change-point regimes, whether or not there are many or few changes. This new approach can be viewed as a natural extension of PELT, where PELT corresponds to the dual evaluation at zero. In this simple case, DUST evaluates a decision function with a closed formula, equivalent to the maximum evaluation of the dual. FPOP-like algorithms can efficiently tackle one-parameter problems. DUST has been proven to be very efficient in our simulation study, saving a slightly greater number of indices than FPOP. FPOP’s reduced runtime can be attributed to the iterative root-finding algorithm used in handling the non-Gaussian cost during functional cost updates, where the piecewise structure is maintained through a sequence of intervals.\\

Extensions of DUST to higher dimensions are straightforward and we still get an efficient pruning in low-dimensional settings with a single constraint. Two-parameter cases involving two constraints are particularly attractive -- for instance, in Gaussian change-point detection for the mean and variance, the maximum of the decision function can be computed exactly. However, designing an optimal approach for scenarios with multiple constraints remains an open challenge. This would require new theoretical developments or an extensive simulation study. Current results are promising; we leave this direction for future work.\\

Using a DUST pruning rule is a natural step that could be introduced in many change-point detection algorithms to efficiently prune indices in all change-point regimes. Even if the explicit maximum of the dual or decision function cannot be explicitly derived, evaluating at one point only (e.g., randomly) is a simple rule, fast to test, with high potential for pruning. The exact evaluation in the one-parameter case provides a very efficient inequality-based test; this test should be further studied with the hope of proving time complexity bounds for very generic underlying data processes.

\begin{appendices}

\section{Useful relations between means}
\label{app:means_relations}

\begin{proposition}
 \label{lemma1}
For all $i\in \{1,\ldots,t-1\}$ we have
\begin{align}
 t V(y_{0t}) - (t-i)&V(y_{it}) - iV(y_{0i})\notag \\
\label{lemmaroots1}
   &=  \frac{(t-i)i}{t} \left(\overline{y}_{it}-\overline{y}_{0i}\right)^2\\
\label{lemmaroots2}
 &=\frac{ti}{t-i} \left(\overline{y}_{0t}-\overline{y}_{0i}\right)^2\\
\label{lemmaroots3}
  &=\frac{t(t-i)}{i} \left(\overline{y}_{0t}-\overline{y}_{it}\right)^2\\
\label{lemmaroots4}
   &=i\left(\overline{y}_{0t}-\overline{y}_{0i}\right)^2+(t-i) \left(\overline{y}_{0t}-\overline{y}_{it}\right)^2
\end{align} 
 \end{proposition}

 \begin{proof}
For any $(a,b) \in (\mathbb{N}^*)^2$ with $a<b$ we easily get the following expressions:

\begin{align}
\label{relationMu1}
   &\overline{y}_{0b}-\overline{y}_{0a} = \frac{b-a}{b}\left(\overline{y}_{ab}- \overline{y}_{0a}\right) \,,\\
\label{relationMu2}
 &\overline{y}_{0b}-\overline{y}_{ab} = \frac{a}{b}\left(\overline{y}_{0a}- \overline{y}_{ab}\right) \,.
\end{align} 
With $a=i$ and $b=t$, we derive from \eqref{lemmaroots1}  the relation \eqref{lemmaroots2} using \eqref{relationMu1} and also derive relation \eqref{lemmaroots3} using \eqref{relationMu2}. Expression \eqref{lemmaroots4} is the result of $\frac{t-i}{t}\eqref{lemmaroots2} + \frac{i}{t}\eqref{lemmaroots3}$. Therefore, we only need to prove relation \eqref{lemmaroots1}. Expanding the squares in variance expressions we get:
\begin{align*}
E &=  t V(y_{0t}) - (t-i)V(y_{it}) - iV(y_{0i})\\
   &= -t (\overline{y}_{0t})^2  + (t-i)(\overline{y}_{it})^2+ i(\overline{y}_{0i})^2\\
  & = i((\overline{y}_{0i})^2-(\overline{y}_{0t})^2) + (t-i) ((\overline{y}_{it})^2-(\overline{y}_{0t})^2)\\
  &=i(\overline{y}_{0i}-\overline{y}_{0t})(\overline{y}_{0i}+\overline{y}_{0t}) + (t-i)(\overline{y}_{it}-\overline{y}_{0t})(\overline{y}_{it}+\overline{y}_{0t})\,.
\end{align*} 
Using again relations \eqref{relationMu1} and \eqref{relationMu2}, we have
  $$E =\frac{(t-i)i}{t}\Big(\overline{y}_{0i}-\overline{y}_{it})(\overline{y}_{0i}+\overline{y}_{0t}) +(\overline{y}_{it}-\overline{y}_{0i})(\overline{y}_{it}+\overline{y}_{0t})\Big)\,,$$
  and we get the expression \eqref{lemmaroots1} by simplifying this last expression. 
 \end{proof}

\section{Proofs of Section \ref{sec:functional}}
\label{app:functional}

\subsection{Proof of Proposition \ref{prop:max_pruning}}
\label{app:max_pruning}

If we would remove index $s_0$ from $\Tau_{t}$ while still having a point $\ttheta_0 \in \Theta$ such that:
\begin{equation}
\label{eq:conditionsNonOptimal}
\left\{\begin{aligned}
&Q_t(\ttheta_0) = q_t^{s_0}(\ttheta_0)< q_t^{s}(\ttheta_0)\quad  \hbox{for} \quad  s \ne s_0\,,\\
& Q_t + \beta > q_t^{s_0}(\ttheta_0)\,,\end{aligned}
\right.
\end{equation}
we show that removing $s_0$ could lead to an under-optimal solution for \eqref{eq:fpop} and then for \eqref{eq:global_cost} at some data time $t_0 > t$. We will choose data points after time $t$ to create such a solution, i.e., with $\argmin_{\ttheta} Q_{t_0}(\ttheta)$ attained by a value on function $q_{t_0}^{s_0}$. To that end, for all further iteration $t' > t$ we choose data points $y_{t'}$ with unitary cost:
$$c(y_{t'}; \theta) = A(\ttheta) - \ttheta\cdot\mathbf{T}(y_{t'})= A(\ttheta) - \ttheta\cdot\nabla A(\ttheta_0) = c(\ttheta)\,.$$

{\it Case 1: non-optimal indices $s \ne s_0 \in \{0,\ldots, t-1\}$.} By continuity of function $Q_t(\cdot)$ there exists a ball centred on $\ttheta_0$ with small radius~$\epsilon$, $B(\ttheta_0, \epsilon)$, such that  $Q_t(\ttheta) = q_t^{s_0}(\ttheta) < q_t^{s}(\ttheta)$ for points $\ttheta$ in this ball. We consider $\ttheta$ outside $B(\ttheta_0, \epsilon)$. By strong convexity of $A$, there exists a constant $a>0$ such that: 
$$c(\ttheta) > c(\ttheta_0) + a \|\ttheta - \ttheta_0\|^2\ge c(\ttheta_0) + a \epsilon^2\,,$$
for all $\ttheta$ outside $B(\ttheta_0, \epsilon)$. We also write
$\Delta = q_t^{s_0}(\ttheta_0) - Q_t \ge 0$.
We consider $T = \lceil\frac{\Delta}{a\epsilon^2}\rceil$, then:
$$Tc(\ttheta) > Tc(\ttheta_0) + aT \epsilon^2 \ge Tc(\ttheta_0) + \Delta\,.$$
Thus,
$$Q_t + Tc(\ttheta)    >  Q_t +  Tc(\ttheta_0) + \Delta   =  Tc(\ttheta_0) + q_t^{s_0}(\ttheta_0)\,,$$
for all $\ttheta$ outside $B(\ttheta_0, \epsilon)$. Thus, the minimum of $Q_{t_0}(\cdot) = Q_{t+T}(\cdot)$ can only be in the ball. All points outside the ball corresponding to indices $s \ne s_0$ in $\{0,\ldots, t-1\}$ are not optimal, which proves case 1.\\

{\it Case 2: non-optimal indices $s \in \{t,\ldots, t+T\}$.} We first consider $s=t$. The second assumption in Equation \eqref{eq:conditionsNonOptimal} ($Q_t + \beta > q_t^{s_0}(\ttheta_0)$) ensures that values visible for function $q_t^{s_0}$ are not all too large and pruned by function $q_{t+1}^t$. We have:
\begin{align*} 
q_{t_0}^t(\ttheta) = Q_t + Tc(\ttheta) + \beta    &> q_t^{s_0}(\ttheta_0) + Tc(\ttheta)  \,,\\
&  \ge q_t^{s_0}(\ttheta_0) + Tc(\ttheta_0) \,,\\
& = q_{t+T}^{s_0}(\ttheta_0) = q_{t_0}^{s_0}(\ttheta_0)\,.
\end{align*}
All points on the function $q_{t_0}^t$ are higher than $q_{t_0}^{s_0}(\ttheta_0)$, consequently index $t$ can not be optimal. 
For $s \in \{t+1,\ldots, t+T\}$, we have for index $\overline{s} < t$:
$$\min_{\ttheta}c(y_{\overline{s}t}; \ttheta) + (s-t)c(\ttheta_0) \le \min_{\ttheta}c(y_{\overline{s}s}; \ttheta)\,,$$
or, 
$$\min_{\ttheta}c(y_{\overline{s}t}; \ttheta) +Tc(\ttheta_0) \le  \min_{\ttheta}c(y_{\overline{s}s}; \ttheta) + (t_0-s)c(\ttheta_0)\,.$$
Using for $\overline{s}$ the best index for last change point in $Q_s$ (written $\overline{s}_0$) we have 
$Q_{\overline{s}_0} + \min_{\ttheta}c(y_{\overline{s}_0s}; \ttheta)+ \beta = Q_s$ and we get:
$$ Q_{\overline{s}_0} + \min_{\ttheta}c(y_{\overline{s}_0t}; \ttheta) +\beta + Tc(\ttheta_0)  \le  Q_s + (t_0-s)c(\ttheta_0)\,.$$
By definition of $Q_t$ (see \eqref{eq:fpop}) we have 
$Q_t \le Q_{\overline{s}_0} + \min_{\ttheta}c(y_{\overline{s}_0t}; \ttheta) + \beta $ and eventually:
\begin{align*} 
q_{t_0}^t(\ttheta_0) = Q_t+T c(\ttheta_0) + \beta   & \le Q_s + (t_0-s)c(\ttheta) + \beta   \,,\\
& \le Q_s + (t_0-s)c(\ttheta) + \beta = q_{t_0}^s(\ttheta) \,.
\end{align*}

The global minimum of the function $Q_{t_0}(\cdot)$ cannot be attained in index $s$, as we proved that it cannot either in index $t$. By exclusion of all indices, except $s_0$, we have shown that the global minimum is attained at a value of the function $q_{t_0}^{s_0}$; consequently, removing $s_0$ from $\Tau_t$ would lead to an under-optimal solution.

\subsection{Proof of Proposition \ref{prop:worstcaseGauss}}
\label{app:worstcaseGauss}

The minimum of the function $q_t^s$ is defined as $m_t^s$. A sufficient condition for no pruning is to meet the condition $m_n^0 = m_n^1 = \cdots = m_n^{n-1}$ obtained at distinct values. This is the case for a strictly increasing time series. We assume that the global minimum is always attained for index $0$ and that the time series is increasing; thus, using the variance notation in Appendix \ref{app:means_relations}, we solve:
$$m_n^0 = nV(y_{0n}) = (n-t)V(y_{tn}) + tV(y_{0t}) + \beta\,,\quad \hbox{for all} \,\, t = 1,\dots, n-1\,.$$
By relations in Appendix~\ref{app:means_relations}, this leads to equations
$$ \frac{nt}{n-t} \left(\overline{y}_{0n}-\overline{y}_{0t}\right)^2 = \beta \,,\quad \hbox{for all} \,\, t = 1,\dots, n-1\,.$$
With $t=1$, we obtain $ \overline{y}_{0n} = \sqrt{\beta\frac{n-1}{n}}$ having chosen first data value $y_1 = 0$, and therefore, with other indices, we get:
$$y_1+\cdots +y_{t} = t\sqrt{\beta\frac{n-1}{n}}- \sqrt{\beta \frac{t(n-t)}{n}}\,.$$
We then find the proposed expression easily. We need to verify that the time series is strictly increasing and that the global minimum is associated with the first data point at each time step. Indeed,
$$y_{t+1} - y_t = \sqrt{\frac{\beta}{n}}\Big(-\sqrt{(t+1)(n-t-1)}+2\sqrt{t(n-t)} -\sqrt{(t-1)(n-t+1)} \Big)$$
is positive as $t \mapsto -\sqrt{t(n-t)}$ is strictly convex on interval $[0,n]$. The last change point is always associated with the first data point, as we have an increasing sequence of data with equality of the minima at the ending instant $n$ ($m_n^0 = m_n^1 = \cdots = m_n^{n-1}$).

\subsection{No pruning in the exponential family}
\label{app:worstCase2}

We say that a time series increases if it is increasing coordinate by coordinate. Under some mild assumptions, we can build a no-pruning example of increasing time series in a more general setting.

\begin{proposition} 
\label{prop:worstcase}
We consider costs of type "$a A(\ttheta) + b \cdot \ttheta + c$" derived from the natural exponential family with continuous density. We consider the following equations in variable~$x$:
\begin{equation}
\label{eq:worstEquation}
E(t): \quad g_t(x) = \frac{\beta}{n} + \mathcal{D^*}(Y)\,,\quad t = 1,\ldots,n-1\,,
\end{equation}
with
$$
\left\{
    \begin{aligned}
    &\mathcal{D^*}(x) = x\cdot (\nabla A)^{-1}(x) - A((\nabla A)^{-1}(x)),,\quad x \in \Omega\,,\\
     &g_t(x) = \frac{t}{n}\mathcal{D^*}(x)  + \frac{n-t}{n}\mathcal{D^*}\Big(\frac{nY - tx}{n-t}\Big)\,,\quad x \in \Omega_t \subset \Omega\,.\\
    \end{aligned}
    \right.
$$
If $Y \in \Omega_0 = \cap_{t=1}^{n-1}\Omega_t$ and $\beta$ is chosen such that:
$$\sup_{x \in \Omega_0\,,\,x \le Y} \{g_1(x)\} > \frac{\beta}{n} + \mathcal{D^*}(Y)\,,$$
there exists an increasing sequence of data $y_1,\ldots,y_n$ with zero pruning solving \eqref{eq:global_cost} which verifies:
$$
\left\{
    \begin{aligned}
    &Y = \overline{y}_{0n} = \frac{1}{n}\sum_{t=1}^{n}y_t\,,\\
    &x = s_t \quad \hbox{solution of}\quad E(t)\,,\\
    &y_t = t s_t - (t-1)s_{t-1}\quad t=1,\ldots,n\,,\quad \hbox{with}\quad s_0 = 0\,.
    \end{aligned}
    \right.
$$
\end{proposition}

\begin{proof}

The proof is an extension of the evidence of Proposition~\ref{prop:worstcaseGauss} and is based on a similar strategy. We look for data configurations for which the minimum of functions $q_t^s$, denoted $m_t^s$, are all the same: $m_n^0 = m_n^1 = \cdots = m_n^{n-1}$. There is no possible pruning if we find an increasing time series realising such conditions. For density from the natural exponential family, we get the relations:
$$ m_n^{t} = Q_t + (n-t)(A(\rho_t) - \overline{y}_{tn}\cdot\rho_t)  + \beta = Q_0 + n(A(\rho_0) - \overline{y}_{0n}\cdot\rho_0) + \beta  =  m_n^{0}\,,$$
with $\rho_i = (\nabla A)^{-1}(\overline{y}_{in})$. We have also $Q_0 = 0$ and $Q_t = t(A((\nabla A)^{-1}(\overline{y}_{0t})) - \overline{y}_{0t}\cdot(\nabla A)^{-1}(\overline{y}_{0t})) + \beta$ as data is not spit into segments (we have only one large segment and the minimum for $Q_t(\cdot)$ is attained on $q_t^0$). This leads to:
$$\frac{t}{n}\mathcal{D^*}(\overline{y}_{0t})  + \frac{n-t}{n}\mathcal{D^*}\Big(\overline{y}_{tn}\Big) = \frac{\beta}{n} + \mathcal{D^*}(\overline{y}_{0n})\,,$$
with $\mathcal{D^*}(x) = x\cdot (\nabla A)^{-1}(x) - A((\nabla A)^{-1}(x))$. We get Equations \eqref{eq:worstEquation} denoted $E(t)$ with $x = \overline{y}_{0t}$, $Y = \overline{y}_{0n}$ and relation $\frac{nY - tx}{n-t} = \overline{y}_{tn}$. It remains to find conditions for the existence of a solution. We define $g_t(x) = \frac{t}{n}\mathcal{D^*}(x)  + \frac{n-t}{n}\mathcal{D^*}\Big(\frac{nY - tx}{n-t}\Big)$ and study this function. We have:
$$
\left\{\begin{aligned}
&g_t(Y) = \mathcal{D^*}(Y)\,,\\
&\nabla g_t(Y) = 0\,,\\
&g_t \quad \hbox{strictly convex}\,,\\
&g_{t+1} > g_{t} \quad \hbox{except in point }\, Y\,.\\
\end{aligned}
\right.
$$
Evaluation of $g_t$ in $Y$ is obvious, for the gradient we get: 
$$ \nabla g_t(x) = \frac{t}{n}\Big( \nabla\mathcal{D^*}(x) - \nabla\mathcal{D^*}\Big(\frac{nY - tx}{n-t}\Big)\Big)\,,$$
and then $\nabla g_t(Y) = 0$. Moreover, function $g_t$ is strictly convex: $\nabla^2 g_t(x) =  \frac{t}{n} \nabla^2\mathcal{D^*}(x) + \frac{t^2}{n(n-t)}\nabla^2\mathcal{D^*}\Big(\frac{nZ - tx}{n-t}\Big)$; its Hessian is definite positive since the Hessian of $ \mathcal{D^*}$ is. We prove $g_{t+1} >g_t$ considering $t$ as a continuous variable. We define $h_x(t) = g_{t}(x)$ and we have:

$$\nabla h_x(t) = \frac{1}{n}\mathcal{D^*}(x) - \frac{1}{n}\mathcal{D^*}\Big(\frac{nY - tx}{n-t}\Big) + \frac{Y-x}{n-t}\cdot \nabla \mathcal{D^*}\Big(\frac{nY - tx}{n-t}\Big)\,.$$
However, by a characterisation of the strict convexity of $\mathcal{D^*}$ we have for $x \ne Y$:
$$ \mathcal{D^*}(x) > \mathcal{D^*}\Big(\frac{nY - tx}{n-t}\Big) + n\frac{Y-x}{n-t}\cdot \nabla \mathcal{D^*}\Big(\frac{nY - tx}{n-t}\Big)\,,$$
as $n\frac{Y-x}{n-t} = x -   \frac{nY-tx}{n-t}$. This leads obviously to $\nabla h_x(t) > 0$ and consequently to $h_x(t+1) - h_x(t) > 0$ or with the $g$ notation: $g_{t+1}(x) >g_t(x)$ for all $t$ and all $x \ne Y$. 

Now, with $\Omega_t$ the convex domain of $g_t$, if $\Omega_0 = \cap_{t=1}^{n-1}\Omega_t$ is not empty, we can choose $Y$ in $\Omega_0$. We also choose $\beta$ such that 
$\sup_{x \in \Omega_0\,,\,x \le Y} \{g_1(x)\} > \frac{\beta}{n} + \mathcal{D^*}(Y)$. By continuity of $g_t$ we can find $x = s_1$ solving $E(1)$. Using the fact that $g_{2}>g_1$ and $\min_{x\in \Omega_0} g_2(x) = \D^*(Y)$ we can now find $x = s_2 > s_1$ solving $E(2)$. Step by step, we build the sequence $(s_t)$ corresponding to an increasing sequence of means $\overline{y}_{0t}$. From this point, it is easy to extract the increasing time series $(y_t)$, which concludes the proof.
\end{proof}

Notice that the multiple independent Poisson model leads to the expression \eqref{eq:worstEquation} written as:
$$\frac{t}{n}x \log(x)  - Y \log (Y) + \Big(Y-\frac{t}{n}x\Big)\log\Big(\frac{Y - \frac{t}{n}x}{1-\frac{t}{n}}\Big) = \frac{\beta}{n}\,,$$
where $x \in \Omega_0 = \Big(0, \frac{n}{n-1}Y\Big) = \Big(0, \frac{n}{n-1}Y_1\Big)\times \cdots \times\Big(0, \frac{n}{n-1}Y_d\Big) \subset \R^d$. The set $\Omega_0$ is not empty and the limit of  $g_1$ in zero leads to the condition for beta: $\beta \le n\log\Big(\frac{n}{n-1}\Big)\Big(\sum_{i=1}^dY_i\Big)$. This bound is much lower than usual values chosen for penalty, bounded by a $\log(n)$ term \cite{cleynen2017model}. As the Poisson model needs integer data points, the solution we get has to be approximated by $\tilde y_t = \lceil y_t \rceil$ and the resulting data and functional cost $Q_n(\cdot)$ have some of their indices pruned as illustrated in Figure~\ref{fig:worstcase_poisson}. If $Y$ is chosen very large, the approximation improves and the obtained time series converges to the Gaussian data with no pruning.\\

\begin{figure}[!t]%
\centering\includegraphics[width=0.8\linewidth]{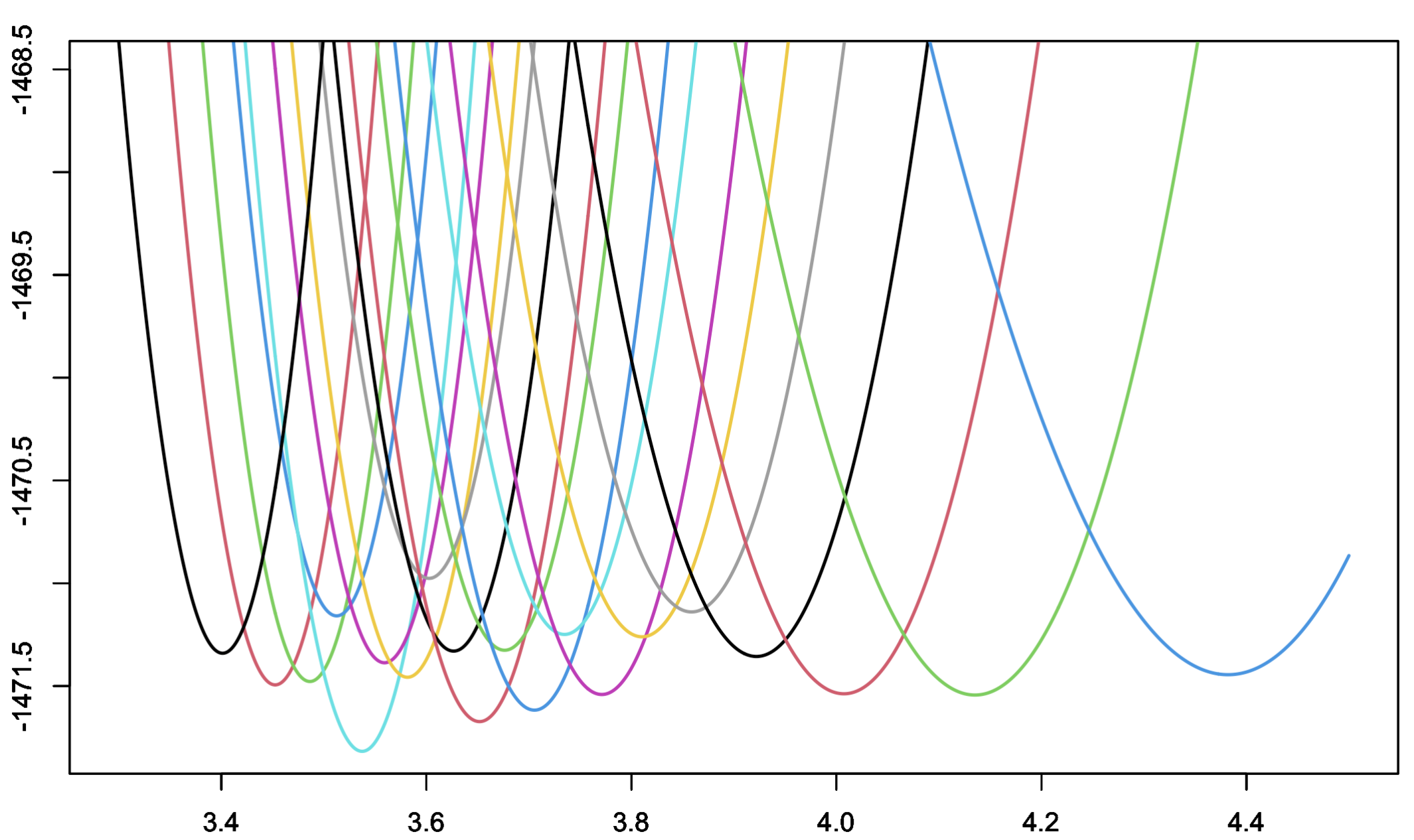}
\caption{Poisson functional cost $Q_{20}(\cdot)$ obtained from $20$ data points. We identify easily the $20$ different inner functions from which the functional cost is built. We choose $Y = 30$ and  $\beta = 0.995 Y \log(n/(n-1))$. Due to the integer approximation for $\tilde y_t$, only $14$ inner functions are visible in the functional cost.}\label{fig:worstcase_poisson}
\end{figure}

In Table \ref{dual_A} we provide the expression for $A$, $(\nabla A)^{-1}$, and its domain $\Omega$ for a few distributions of the natural exponential family. For instance, the exponential model leads to the expression \eqref{eq:worstEquation} written as:
$$\log (Y)-\frac{t}{n}\log(x)   - \Big(1 - \frac{t}{n}\Big)\log\Big(\frac{Y - \frac{t}{n}x}{1-\frac{t}{n}}\Big) = \frac{\beta}{n}\,,$$
where $x \in \Omega_0 = \Big(0, \frac{n}{n-1}Y\Big)$. For a multiple independent exponential model, the dual $\mathcal{D^*}$ can be decomposed as the sum of dual functions, one for each dimension: $\mathcal{D^*}(x) = \sum_{i=1}^{d}\D^*_i(x_i)$, with $x = (x_1,\ldots, x_d)^T$, justifying the following rectangular form for $\Omega_0 = \Big(0, \frac{n}{n-1}Y_1\Big)\times \cdots \times\Big(0, \frac{n}{n-1}Y_d\Big) \subset \R^d$. The set $\Omega_0$ is not empty, and the limit of  $g_1$ is infinite; thus, there is no restriction on the $\beta$ value. A more general result for the exponential family seems possible. Still, it requires solving Equation \eqref{eq:worstEquation} for $x = \frac{1}{t}\sum_{i=1}^t \mathbf{T}(y_i)$ and inverting $\mathbf{T}$ to explicit the time series $(y_t)$. This step depends on the chosen distribution and requires additional effort, which is left for further work.

\begin{figure}[!t]%
\centering\includegraphics[width=0.8\linewidth]{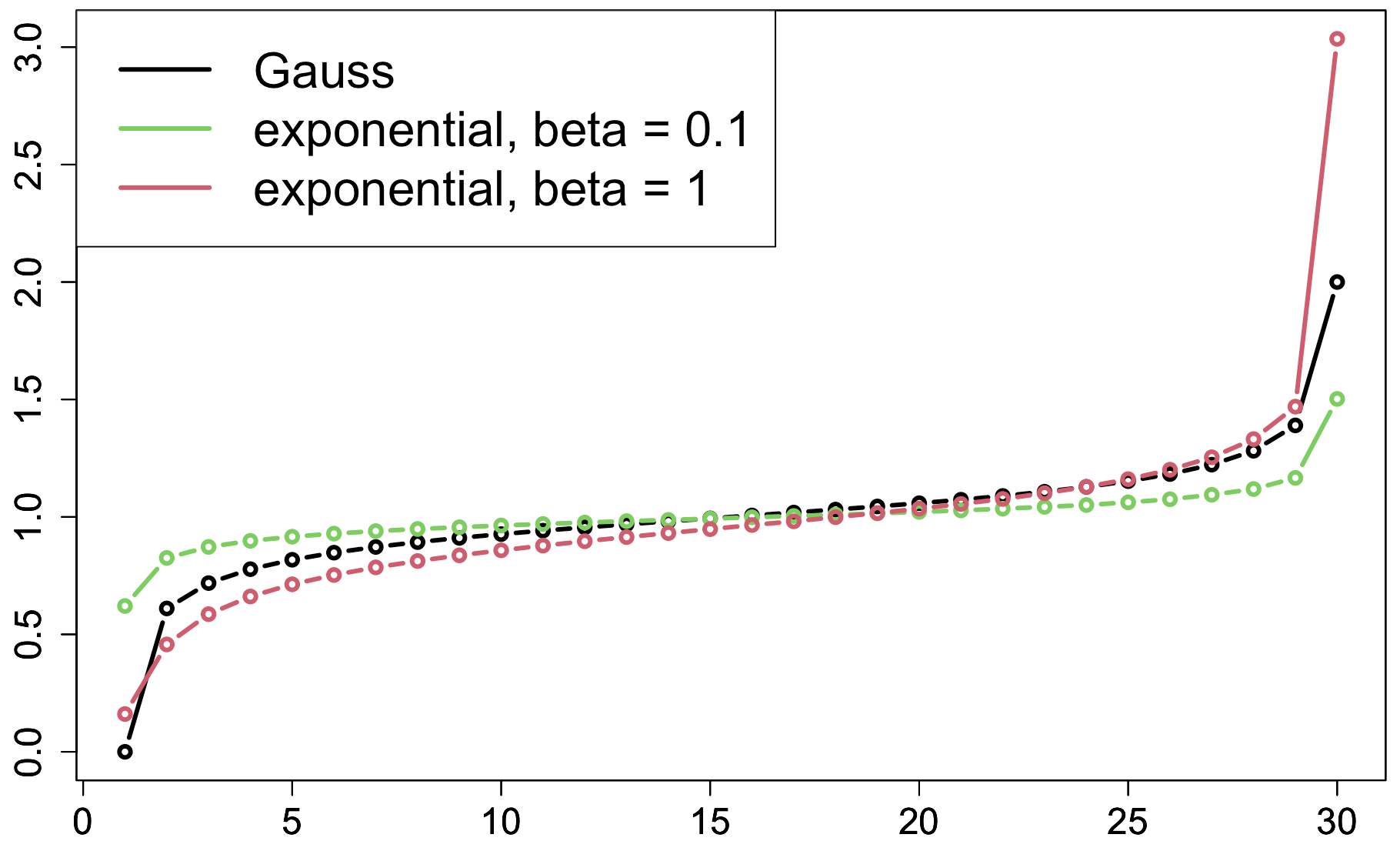}
\caption{Examples with $30$ data points in worst case complexity (no pruning) with Gaussian model (middle curve in black) and two examples with exponential model with two different penalty values (beta). We chose $Y = 1$.}\label{fig:worstcase}
\end{figure}

\newpage

\section{Proofs of Section \ref{sec:simple}}
\label{app:simple}

\subsection{Proof of Proposition \ref{prop:dual1D}}
\label{app:dual1D}

The primal Lagrangian function is given for a non-negative multiplier $\mu$ by the relation:
\begin{align*}
\mathcal{L}(\ttheta, \mu) &= q_t^s(\ttheta)+ \mu( q_t^s(\ttheta)- q_t^{r}(\ttheta) )\,,\\
&=  (t-s)A(\ttheta) - \ttheta\cdot \mathbf{S}_{st} + Q_s + \beta  - \mu \Big((s-r)A(\ttheta) - \ttheta\cdot \mathbf{S}_{rs}  + Q_{r} - Q_s\Big)\,, \\
&= \Big((t-s)-\mu(s-r)\Big)A(\ttheta) - \ttheta\cdot \Big(\mathbf{S}_{st}-\mu \mathbf{S}_{rs}\Big) + \beta + Q_s  - \mu  (Q_{r} - Q_s)\,.
\end{align*}
We get the critical point solving $\nabla_{\ttheta} \mathcal{L} (\ttheta, \mu) = 0$:
$$\ttheta^*(\mu) = (\nabla A)^{-1}\Big(\frac{ \mathbf{S}_{st} - \mu \mathbf{S}_{rs}}{(t-s) - \mu (s-r)}\Big) = (\nabla A)^{-1}\Big(\mathbf{m}(\mu)\Big)\,,$$
which leads to the proposed expression:
$$
\mathcal{L}(\ttheta^*(\mu), \mu) = \D(\mu) =  -\Big((t-s)-\mu(s-r)\Big) \mathcal{D^*}(\mathbf{m}(\mu))  + \beta + (1+\mu)Q_s  - \mu  Q_{r}\,,
$$
where $\mathcal{D^*}(x) = x \cdot (\nabla A)^{-1}(x) - A((\nabla A)^{-1}(x))$. We have the constraint:
$$\nabla A(\ttheta) = \frac{\mathbf{S}_{st} - \mu \mathbf{S}_{rs}}{(t-s) - \mu (s-r)} \in \mathcal{M}^o\,,$$
where $\mathcal{M}^o$ is an open convex set (\cite{brown1986fundamentals,wainwright2008graphical}) including the value $\mathbf{S}_{st}$ and thus $\mu = 0$ is in the definition domain and $\mu = \frac{t-s}{s-r}$ is the maximum feasible value (but it can be a smaller value, depending on the shape of $\mathcal{M}^o$). The result of proposition \ref{prop:dual1D} is obtained after the simple change of variable $\mu \to \mu \frac{t-s}{s-r}$.

\subsection{Proof of Theorem \ref{prop:dual1D_max_exact}}
\label{app:dual1D_max_exact}

We consider the decision function~\eqref{eq:decisionDual1D} defined on segment $[0, x_{max})$ with univariate data:
$$\mathbb{D}(x) = -\D^*\Big(\overline{S}_{st} + x\Delta \overline{S}_{rst}\Big) - \Big(\overline{Q}_{st} + x\Delta \overline{Q}_{rst}\Big) \,,$$
where we assume $\Delta \overline{S}_{rst} \ne 0$ in order to rule out the trivial linear case. Differentiating $\mathbb{D}$ with respect to $x$, we obtain:
$$ - \Delta \overline{S}_{rst}(\nabla A)^{-1}(\overline{S}_{st} + x\Delta \overline{S}_{rst}) - \Delta \overline{Q}_{rst} = 0\,.$$ 
The critical point corresponds to a maximum, since the decision function is concave. We thus obtain:
$$ x^\star = \frac{1}{\Delta \overline{S}_{rst}}\Big(\nabla A \Big(-\frac{\Delta \overline{Q}_{rst}}{\Delta \overline{S}_{rst}}\Big) - \overline{S}_{st}\Big)\,,$$ 
and 
\begin{align*}
\mathbb{D}(x^\star) &=  -\D^*\Big(\nabla A \Big(-\frac{\Delta \overline{Q}_{rst}}{\Delta \overline{S}_{rst}}\Big)\Big) - \Big(\overline{Q}_{st} + x^*\Delta \overline{Q}_{rst}\Big)\,,\\
&= -\nabla A \Big(-\frac{\Delta \overline{Q}_{rst}}{\Delta \overline{S}_{rst}}\Big) \Big(-\frac{\Delta \overline{Q}_{rst}}{\Delta \overline{S}_{rst}}\Big) + A\Big(-\frac{\Delta \overline{Q}_{rst}}{\Delta \overline{S}_{rst}}\Big) - \Big(\overline{Q}_{st} +  \frac{\Delta \overline{Q}_{rst}}{\Delta \overline{S}_{rst}}\Big(\nabla A \Big(-\frac{\Delta \overline{Q}_{rst}}{\Delta \overline{S}_{rst}}\Big) - \overline{S}_{st}\Big)\Big)\,,\\
&= A\Big(-\frac{\Delta \overline{Q}_{rst}}{\Delta \overline{S}_{rst}}\Big) - \Big(-\frac{\Delta \overline{Q}_{rst}}{\Delta \overline{S}_{rst}}\Big) \overline{S}_{st} - \overline{Q}_{st}\,.
\end{align*}
We notice that:

$$(t-s)\mathbb{D}(x^\star)= (t-s)A\Big(-\frac{\Delta \overline{Q}_{rst}}{\Delta \overline{S}_{rst}}\Big) - \Big(-\frac{\Delta \overline{Q}_{rst}}{\Delta \overline{S}_{rst}}\Big) S_{st} - (Q_t - Q_s)$$
and the pruning occurs when $\mathbb{D}(x^\star)> 0$ that is when:
$$q_t^s\Big(-\frac{\Delta \overline{Q}_{rst}}{\Delta \overline{S}_{rst}}\Big) = (t-s)A\Big(-\frac{\Delta \overline{Q}_{rst}}{\Delta \overline{S}_{rst}}\Big) - \Big(-\frac{\Delta \overline{Q}_{rst}}{\Delta \overline{S}_{rst}}\Big) S_{st} + Q_s + \beta  > Q_t + \beta\,.$$
If $x^{\star} <0$ or  $x^{\star} > x_{max}$, the decision function is evaluated in $0$ or $x_{max}$, respectively. In case $x^{\star} <0$, the test is exactly the PELT inequality pruning test: $q_t^s\Big((\nabla A)^{-1}(\overline{S}_{st})\Big)  > Q_t + \beta$.

\subsection{Dual for the change-in-mean-and-variance problem}
\label{app:meanVar}

We use the canonical form of the exponential family, for $(\theta_1, \theta_2) \in \R \times \R^-$:
\begin{equation}
c(y_{st}; [\theta_1; \theta_2]) =  (t-s)A(\theta_1,\theta_2) -  \theta_1 \Big(\sum_{i=s+1}^{t}y_i\Big) - \theta_2 \Big(\sum_{i=s+1}^{t}y_i^2\Big)\,,
\end{equation}
where 
$$A(\theta_1,\theta_2) =-\frac{\theta_1^2}{4\theta_2} + \frac{1}{2}\log\Big(-\frac{1}{2\theta_2} \Big)\,.$$
Its gradient is:
$$\nabla A (\theta) = \Big(-\frac{\theta_1}{2\theta_2}, \frac{\theta_1^2}{4\theta_2^2} - \frac{1}{2\theta_2}\Big)\,,$$
and$$\nabla A (\theta) = (u,v) \iff (\theta_1, \theta_2) = \Big(-\frac{u}{u^2-v},\frac{1}{2(u^2-v)}\Big)\,.$$
The minimum cost is then
\begin{equation}
c(y_{st}) =  \frac{t-s}{2}\Big(1 +  \log(v-u^2)\Big)\quad \hbox{with}\quad u = \overline{y}_{st}\quad v = \overline{y_{st}^2}\,.
\label{mincostMeanVar}
\end{equation}
We obtain the dual, considering the case $r<s$ (the case producing an efficient pruning) with constraint $q_t^{s} -q_t^{r}\le 0$:
$$
\D(\mu) =  \frac{1}{2}\Big((t-s)-\mu(s-r)\Big)\Bigg[1+ \log\Big( \frac{{S^2}_{st}-\mu{S^2}_{rs} }{(t-s)-\mu(s-r)}-\Big(\frac{{S}_{st}-\mu{S}_{rs}}{(t-s)-\mu(s-r)}\Big)^2\Big)\Bigg]$$
$$+ \beta + (1+\mu)Q_s  - \mu  Q_{r}\,.$$

The dual is transformed into its decision function:
$$
\mathbb{D}(x) =  \frac{1}{2}\Bigg[1+ \log\Big( \overline{S^2_{st}}+x\Delta\overline{S_{rst}^2}-(\overline{S}_{st}+ x\Delta\overline{S}_{rst})^2\Big)\Bigg]- \Big(\overline{Q}_{st} + x\Delta \overline{Q}_{rst}\Big)\,,$$
its domain is given by the non-negative $x$ satisfying the condition $\overline{S^2_{st}}+x\Delta\overline{S_{rst}^2}-(\overline{S}_{st}+ x\Delta\overline{S}_{rst})^2 > 0$, which we write after simplification as:

$$-(\Delta\overline{S}_{rst})^2x^2 + (V(y_{st}) - V(y_{rs}) -(\Delta\overline{S}_{rst})^2))x + V(y_{st}) > 0\,.$$

We have the following roots, between which the polynomial takes positive values:

\begin{align*}
x_{\pm} &= \frac{V(y_{st}) - V(y_{rs}) -(\Delta\overline{S}_{rst})^2}{2(\Delta\overline{S}_{rst})^2} \pm \sqrt{\frac{\Big(V(y_{st}) - V(y_{rs}) -(\Delta\overline{S}_{rst})^2\Big)^2 + 4(\Delta\overline{S}_{rst})^2V(y_{st})}{4(\Delta\overline{S}_{rst})^4}}\\
&= x_0 + \sqrt{x_0^2 + \frac{V(y_{st})}{(\Delta\overline{S}_{rst})^2}}\,,
\end{align*}
where 
$$x_0 = \frac{1}{2} \Big( \frac{V(y_{st}) - V(y_{rs})}{(\Delta\overline{S}_{rst})^2} - 1 \Big)\,.$$
To get the argument of the maximum, we differentiate $\mathbb{D}$ to search for the critical point, which yields:
$$ \frac{1}{2}\frac{\Delta\overline{S_{rst}^2}-2(\overline{S}_{st}+ x\Delta\overline{S}_{rst})}{\overline{S^2_{st}}+x\Delta\overline{S_{rst}^2}-(\overline{S}_{st}+ x\Delta\overline{S}_{rst})^2}- \Delta \overline{Q}_{rst} = 0$$
or 
$$ \Delta\overline{S_{rst}^2}-2(\overline{S}_{st}+ x\Delta\overline{S}_{rst})- 2\Delta \overline{Q}_{rst}(\overline{S^2_{st}}+x\Delta\overline{S_{rst}^2}-(\overline{S}_{st}+ x\Delta\overline{S}_{rst})^2) = 0\,.$$
After simple computations we get the desired result when $\Delta \overline{Q}_{rst} \ne 0$.

\section{Proofs of Section \ref{sec:duality}}
\label{app:duality}

\subsection{Proof of Theorem \ref{th:dualMultiple}}
\label{app:dualMultiple}

The proof is similar to the proof with a single constraint. We differentiate in $\ttheta$ the Lagrangian function $\mathcal{L}(\ttheta, \overline{\mu}) = q_t^s(\ttheta)+ \sum_{r \ne s}\overline{\mu}_{r}( q_t^s(\ttheta)- q_t^{r}(\ttheta) )$ with $q_t^r(\ttheta) = Q_r + (t-r)A(\ttheta) - \ttheta\cdot \mathbf{S}_{rt} + \beta$ and inject its solution $\ttheta^*(\overline{\mu}) = (\nabla A)^{-1}(\mathbf{m}(\overline{\mu}))$ to get $\D_{st}(\overline{\mu}) = \mathcal{L}(\ttheta(\overline{\mu}), \overline{\mu})$. The final result comes after the change of variable $\overline{\mu}_{r} = \mu_{r} \frac{t-s}{|s-r|}$. The decision function is obtained by the simple change of variable $x_r = \mu_r / l(\mu)$.

\subsection{Proof of Theorem \ref{th:noDualityGap}}
\label{app:noDualityGap}

For simplicity, we rename the indices as follows: $s \to 0$, $r_i \to i$ for $i=1,\ldots,d$, skip the $t$ variable, and move the upper indices to a lower position. Therefore we solve: $\min_{\ttheta\in \Theta} q_0(\ttheta)$ under constraints $q_0(\ttheta)- q_i(\ttheta) \le 0$ for all indices $i \in \{1,\ldots,d\}$. We also use the vector function notation $q = (q_1,\ldots, q_d)$.

The proof for strong duality relies on the following three lemmas.

\begin{lemma}
\label{lemma:convexityDerivative}
We consider a function $f:\R^d \to \R$ for which the evaluation at point $c \in \R^p$ can only be done through an auxiliary bijection $g: \R^p \to \R^p$ such that there exists $\ttheta_c \in \R^p$ with  $g(\ttheta_c) =c$. Function $f$ is convex if for any $a,b$ in the convex domain $\Omega$ of $f$ there exists a parametric curve $g_{ab}$ joining in a straight line the points $a$ to $b$ such that $g_{ab}(\ttheta_{(1-\alpha)a + \alpha b}) = (1-\alpha)a + \alpha b$ and with notation $h(\alpha) = \ttheta_{(1-\alpha)a + \alpha b}$ ($h: (0,1) \to \R^p$) we have:
$$\nabla (fg_{ab}h(0)) \cdot (h'(0)) \le fg_{ab}h(1)- fg_{ab}h(0)\,.$$
\end{lemma}

It's a generalisation of a well-known result. Function $f: \R^p \to \R$ is convex if for all $a,b$ in its domain:
$\nabla f(a) \cdot (b - a) \le f(b)- f(a)$. When $g_{ab}$ is the identity, we return to this result.

\begin{proof}(of Lemma \ref{lemma:convexityDerivative}) The convexity of $f$ is the same as the convexity of $fgh$ by definition of $fgh$. Using the definition of convexity for $f g h$, we have for all $\alpha \in (0,1)$:
$$fgh(\alpha) \le (1-\alpha) fgh(0) + \alpha fgh(1)\quad \hbox{or} \quad fgh(\alpha) - fgh(0) \le \alpha (fgh(1) - fgh(0))\,.$$
Dividing by $\alpha$ and taking the limit $\alpha \to 0$: $\nabla (fgh(0)) \cdot (h'(0)) \le fgh(1)- fgh(0)$.
\end{proof}

\begin{lemma}
\label{lemma:intersection}
Consider the $d$ constraints.
Let $\ttheta_0$ be in $\R^d$ such that $(q_0-q)(\ttheta_0) = c$  with $c \in \R^p$, then $\ttheta_0 \in \{x^+ u + w, x^- u + w\}$ with $x^+, x^- \in \R$ and $u,w \in \R^p$. They are the two intersection points between a straight line and a level curve $(q_0-q_1) = c_1$.
\end{lemma}

\begin{proof}(of Lemma \ref{lemma:intersection})
We have for $i=1,\ldots,p$:
$$(q_0-q_i)(\ttheta_0) =(r_i-s)A(\ttheta_0) - \ttheta_0\cdot \mathbf{S}_{r_is} + Q_s- Q_{r_i} = c_i\,,$$
or
$$A(\ttheta_0) - \ttheta_0\cdot \overline{\mathbf{S}}_{r_is} = \frac{c_i}{r_i-s} + \frac{Q_{r_i}- Q_s}{r_i - s}\,.$$
Combining equations with indices $i$ and $i+1$ for $i=1,\ldots,p-1$ to remove $A(\ttheta_0)$ we get:
$$\ttheta_0 \cdot( \overline{\mathbf{S}}_{r_{i+1}s} -  \overline{\mathbf{S}}_{r_is}) =  \frac{c_i}{r_i - s}-\frac{c_{i+1}}{r_{i+1}-s}+ \frac{Q_{r_i}-Q_s}{r_i - s}   -  \frac{Q_{r_{i+1}} - Q_s}{r_{i+1} - s}\,.$$
With the assumption that the points $\overline{\mathbf{S}}_{r_{i}s}$ are in general position (no redundant equation or unsolvable system), we get $\ttheta_0(x) = x u + w$ for potential solutions where $x \in \R$ is a parameter and $u,w \in \R^p$ are fixed values. As we considered that a solution exists and as $A$ is convex, this straight line intersects in two points the level curve $(q_0 - q_1)=c_1$ (counted with their multiplicity).
\end{proof}

\begin{remark}
\label{rem:sign_intersection_u}
Using relation $( \overline{\mathbf{S}}_{r_{i+1}s} -  \overline{\mathbf{S}}_{r_is})\cdot u = 0$, we easily get that for all indices $i$ and a fixed intersection value $\ttheta_0$ we have expressions $(\nabla A(\ttheta_0) -  \overline{\mathbf{S}}_{r_is} )\cdot u$ is a constant.
\end{remark}

\begin{lemma}
\label{lemma:OprojectionConvex}
The orthogonal projection of $\mathcal{O}$ on its last $d$ variables ($\mathcal{O} \cdot (\emptyset, \R^d)$) is a convex object.
\end{lemma}
\begin{proof}(of Lemma \ref{lemma:OprojectionConvex})
Let $\ttheta_a$ and $\ttheta_b$ be in $\R^d$ such that $(q_0-q)(\ttheta_a) = a$ and $(q_0-q)(\ttheta_b) = b$ with $a,b \in \R^p$. We need to prove that for all $\alpha \in (0,1)$, there exists $\ttheta_{\alpha}$ such that $(q_0-q)(\ttheta_{\alpha}) = (1-\alpha)a + \alpha b$. Using Lemma \ref{lemma:intersection} there are only two points for each $\alpha$ (counting multiplicity) in $\mathcal{O} \cdot (\emptyset, \R^d)$ solving the equations: $x^-(\alpha) u + \alpha v +  w$ and $x^+(\alpha) u + \alpha v + w$. They intersect in $\alpha = 0$ and $\alpha = 1$ the level curve  $(q_0-q_1) = a$ and  $(q_0-q_1) = b$, respectively. We must show that the straight line also intersects $(q_0-q_1) = (1-\alpha)a + \alpha b$. We increase by one the size of the problem, considering $\alpha$ to be one of the variables. the plane $xu+\alpha v + w$ intersect $f(\ttheta,\alpha)=(q_0-q_1)(\ttheta) - \alpha(b-a) -a$ in $\alpha = 0$ and $\alpha = 1$. $f(\ttheta,\alpha) \le 0$ or $f(\ttheta,\alpha)  \ge 0$ is convex and bounded. Thus, the intersection also occurs for all $\alpha$ in $(0,1)$.
\end{proof}

{\bf Proof of the Theorem}: We need to prove the convexity of the epigraph, that is, for the lowest values of $q_0$ for points on the projection $\mathcal{O} \cdot (\emptyset, \R^d)$. This is not equivalent to the convexity of $q_0$ as the evaluation is done here on the initial parameter's functions (the constraint values) instead of $\ttheta$. Let $\ttheta_a$ and $\ttheta_b$ be in $\R^d$ and $a,b$ be in $\mathcal{O} \cdot (\emptyset, \R^d)$ such that $(q_0-q)(\ttheta_a) = a$ and $(q_0-q)(\ttheta_b) = b$.

We have:
$$h^-: \R \to \R^p, \alpha \mapsto x^-(\alpha)u + \alpha v + w\,,\quad h^+: \R \to \R^p, \alpha \mapsto x^+(\alpha)u + \alpha v + w\,.$$ 
such that $\ttheta_a \in \{h^-(0), h^+(0)\}$,  $\ttheta_b \in \{h^-(1), h^+(1)\}$ and $g(h^-(\alpha)) = g(h^+(\alpha)) = (1-\alpha)a+ \alpha b$. Function $h^-$ corresponds to the solution for $q_0$ returning the smallest value:
$$q_0(g_{ab}(h^-(\alpha)) \le q_0(g_{ab}(h^+(\alpha))\,.$$
$h^-$ and $h^+$ can be discontinuous, however $h^-$ is differentiable in $0$ (no jump in minimum value at $0$). Applying Lemma \ref{lemma:convexityDerivative} with $h = h^-$ (and $x = x^-$) we need to prove:
$$\nabla (fg_{ab}h(0)) \cdot (h'(0)) \le fg_{ab}h(1)- fg_{ab}h(0)\,.$$
That is:
\begin{equation}
\label{equationToSolve}
\Big(\nabla A(x(0)u + w) -  \overline{\mathbf{S}}_{st} \Big)\cdot \Big(\frac{dx}{d \alpha}(0) u + v \Big) \le \frac{q_0(\ttheta_b) - q_0(\ttheta_a)}{t-s}
\end{equation}

Notice that we decide to choose $u$ such that $\Big(\overline{\mathbf{S}}_{s_1s}  -  \overline{\mathbf{S}}_{st} \Big)\cdot u \ge 0$ (as $u$ can be replaced by $-u$ if we have the wrong sign).
We can also derive two interesting equalities. We consider the intersection between the line and the level curve (see Lemma \ref{lemma:intersection}), which also intersects the curves $q_0 = q_0(\ttheta_a)$ and  $q_0 = q_0(\ttheta_b)$. For $\alpha = 0$ and $\alpha = 1$ we get four relations:

$$\left\{
\begin{aligned}
&A(\ttheta_a) - \ttheta_a\cdot \overline{\mathbf{S}}_{st} + \frac{Q_{s}+\beta}{t- s}= \frac{q_0(\ttheta_a)}{t-s} \,,\\
&A(\ttheta_b) - \ttheta_b\cdot \overline{\mathbf{S}}_{st} + \frac{Q_{s}+\beta}{t- s}= \frac{q_0(\ttheta_b)}{t-s}  \,,\\
& A(\ttheta_a) - \ttheta_a\cdot \overline{\mathbf{S}}_{s_1s} - \frac{Q_{s_1} - Q_{s}}{s_1- s}= \frac{(q_0-q_1)(\ttheta_a)}{s_1-s}\,,\\
&A(\ttheta_b) - \ttheta_b\cdot \overline{\mathbf{S}}_{s_1s} - \frac{Q_{s_1} - Q_{s}}{s_1- s}= \frac{(q_0-q_1)(\ttheta_b)}{s_1-s}\,,
\end{aligned}
\right.$$
leading to relation:
\begin{equation}
\label{equation1ab}
(\overline{\mathbf{S}}_{s_1s}  - \overline{\mathbf{S}}_{st} )\cdot (\ttheta_b - \ttheta_a)= \frac{q_0(\ttheta_b) - q_0(\ttheta_a)}{s_1-s} - \frac{(q_0-q_1)(\ttheta_b) - (q_0-q_1)(\ttheta_a)}{s_1-s}\,.
\end{equation}
The second relation is the dynamic equation for points $x^{\pm}(\alpha)u + \alpha v + w$. We use relations:
$$(q_0-q_1)(x(0)u + w) = (q_0-q_1)(\ttheta_a)$$
and 
$$(q_0-q_1)(x(\epsilon)u + \epsilon v + w) = (1-\epsilon)(q_0-q_1)(\ttheta_a) + \epsilon (q_0-q_1)(\ttheta_b)$$
We derive, taking the difference with $\epsilon \to 0$:
\begin{equation}
\label{equation2ab}
\Big(\nabla A(x(0)u + w) -  \overline{\mathbf{S}}_{s_1s} \Big)\cdot(\frac{dx}{d\alpha}(0) u + v)  = \frac{(q_0-q_1)(\ttheta_b) - (q_0-q_1)(\ttheta_a)}{s_1-s}\,.
\end{equation}
Using \eqref{equation1ab} and \eqref{equation2ab} we can now reformulate \eqref{equationToSolve} as:
$$\Big(\overline{\mathbf{S}}_{s_1s}  -  \overline{\mathbf{S}}_{st} \Big)\cdot(\frac{dx}{d\alpha}(0) u + v)  \le (\overline{\mathbf{S}}_{s_1s}  - \overline{\mathbf{S}}_{st} )\cdot (\ttheta_b - \ttheta_a)\,,$$
or 
$$\frac{dx}{d\alpha}(0) \Big(\overline{\mathbf{S}}_{s_1s}  -  \overline{\mathbf{S}}_{st} \Big)\cdot u  \le (x(1)-x(0))(\overline{\mathbf{S}}_{s_1s}  - \overline{\mathbf{S}}_{st} )\cdot u\,.$$
As we set $\Big(\overline{\mathbf{S}}_{s_1s}  -  \overline{\mathbf{S}}_{st} \Big)\cdot u \ge 0$, this gives us the relation:
$$\nabla x(0) \cdot (1 - 0) \le x(1)- x(0)\,,$$
to be proven. This is true if $\alpha \mapsto x(\alpha)$ is convex. Generalising Equation \eqref{equation2ab} for all $\alpha \in (0,1)$ we get:
$$\Big(\nabla A(x(\alpha)u + \alpha v + w) -  \overline{\mathbf{S}}_{s_1s} \Big)\cdot(\frac{dx}{d\alpha}(\alpha) u + v)  = \frac{(q_0-q_1)(\ttheta_b) - (q_0-q_1)(\ttheta_a)}{s_1-s}\,,$$
and differentiating in $\alpha$:
$$h'(\alpha) \cdot \Big(\nabla A^2(h(\alpha))\Big)\cdot h'(\alpha)  = - \frac{d^2x}{d\alpha^2}(\alpha) \Big(\nabla A(h(\alpha)) -  \overline{\mathbf{S}}_{s_1s} \Big)\cdot u \,.$$
However, $\Big(\nabla A(h(\alpha)) -  \overline{\mathbf{S}}_{s_1s} \Big)\cdot u < 0$, knowing the choice of $u$ we made and Remark~\ref{rem:sign_intersection_u}. We exclude the nullity case as it is possible only for a point on the boundary of $\mathcal{O} \cdot (\emptyset, \R^d)$. We can restrict points $a$ and $b$ to be in the interior and still prove the same result. Using the convexity of $A$, we obtain $\frac{d^2x}{d\alpha^2}(\alpha) > 0$ and we find that $\alpha \mapsto x(\alpha)$ is convex.

\subsection{Proof of Proposition~\ref{prop:decisionMD}}

The gradient of $\D^*$ is equal to $(\nabla A)^{-1}$ which leads to relation:
$$\nabla \D^*\Big(\overline{\mathbf{S}}_{st} + \sum_{r \ne s} x_{r} \Delta \overline{\mathbf{S}}_{rst}\Big) = (\Delta \overline{\mathbf{S}}_{\bullet st})^{T}\Big( (\nabla A)^{-1}\Big(\overline{\mathbf{S}}_{st} + \sum_{r \ne s} x_{r} \Delta \overline{\mathbf{S}}_{rst}\Big)\Big)\,.$$
We introduce notation $\mathbf{y} =  (\nabla A)^{-1}\Big(\overline{\mathbf{S}}_{st} + \sum_{r \ne s} x_{r} \Delta \overline{\mathbf{S}}_{rst}\Big)$ which leads to the first system of equations: $(\Delta \overline{\mathbf{S}}_{\bullet st}) \, \mathbf{x}  = \nabla A(\mathbf{y}) - \overline{\mathbf{S}}_{st}$. With such notation, the critical point (the maximum, function $\mathbb{D}$ being concave), is the solution of $\nabla \mathbb{D}(\mathbf{x}) = 0$ and gives:
$$-(\Delta \overline{\mathbf{S}}_{\bullet st})^{T}\mathbf{y} - \Delta \overline{Q}_{\bullet st} = 0\,.$$

\section{Quadratic cost function}
\label{app:quadratic}

\subsection{The General case}

We consider the following functions of the type:
$$q_t^{s}(\theta_1,\theta_2) = A_{st}\theta_1^2 + 2B_{st}\theta_1  \theta_2 + C_{st} \theta_2^2 + 2D_{st} \theta_1 + 2E_{st} \theta_2 +  F_{st}\,.$$
A natural restriction is to study the case when all the cost functions have a unique finite minimum value and a unique argument for the minimum. Thus, we can define the cost of a segment and associate with it a non-ambiguous parameter value. For any function $q_t^{k}$ it is equivalent to condition $A_{kt} C_{kt} -  B_{kt}^2> 0$ with  $A_{kt} > 0$ (such functions are then strictly convex).
With the Lagrangian with one constraint given by inequality $q_t^{s}-q_t^{r} \le 0$ (always considering that $r<s$, unless otherwise specified) we obtain the coefficients:
$$A(\mu) = A_{st} + \mu(A_{st} - A_{rt})\,,\,B( \mu) = B_{st} +\mu(B_{st} - B_{rt})\,\ldots$$
After computation, we have the following dual function given by the expression
\begin{equation}
\label{dual_quadratic}
\D(\mu) = \frac{2B( \mu)D(\mu)E(\mu) - A(\mu)E^2(\mu)-C(\mu)D^2(\mu)}{A( \mu)C(\mu)-B^2(\mu)} + F(\mu)\,.
\end{equation}

Suppose there is no possible reduction of the rational function defining the dual. When the underlying process generating the time-series and therefore the functions $q_t^k$ is continuous, the possibility of a reduction is certainly an event of measure zero. Thus, we are looking for the first positive value $\mu_{max}$ such that $A(\mu_{max})C(\mu_{max})-B^2(\mu_{max})=0$. This leads to the following result.

\begin{proposition}
The maximal value $\mu_{max}$ of the dual function given by Equation \eqref{dual_quadratic}, if no possible reduction of the fraction, is the smallest positive root of:
$$A(\mu)C(\mu)-B^2(\mu) = (\omega_1^2 - 2\Delta + \omega_2^2)\mu^2 - 2(\Delta-\omega_1^2)\mu + \omega_1^2= 0\,,$$
with
$$\omega_1^2 = A_{st}C_{st}-B_{st}^2\,,\,\, 2\Delta = A_{st}C_{rt}+A_{rt}C_{st} - 2B_{st}B_{rt}\,\, \hbox{and}\,\, \omega_2^2 = A_{rt}C_{rt}-B_{rt}^2\,.$$
We have $\omega_1^2 >0$ and $\omega_2^2 > 0$ and we consider that $\omega_1^2 \ne \omega_2^2$. The discriminant is also positive: $4(\Delta^2 - \omega_1^2 \omega_2^2) > 0$.
The maximal value is then given by:
\begin{equation}
\label{eq:quadratic_mu_max}
\mu_{max} = \left\{\begin{aligned}
& \frac{\Delta - \omega_1^2 - \sqrt{\Delta^2 - \omega_1^2 \omega_2^2}}{\omega_1^2 - 2\Delta + \omega_2^2}\,,\quad &\hbox{if}\quad   2\Delta > \omega_1^2 + \omega_2^2 & \,,\\
& &\hbox{or}\quad  2\Delta < \omega_1^2 + \omega_2^2\,, \,&\omega_1^2 < \Delta  \,,\\
&\frac{\omega_1^2}{\omega_2^2 -\omega_1^2}\,,\quad & \hbox{if}\quad  2\Delta = \omega_1^2 + \omega_2^2\,,\,& \omega_1^2 < \Delta  \,,\\
&+\infty\,, \quad &\hbox{if}\quad  2\Delta \le \omega_1^2 + \omega_2^2\,,& \, \omega_1^2 > \Delta\,.\\
\end{aligned}
\right.
\end{equation}

\end{proposition}

\begin{proof}
First, we propose that the discriminant is positive. We can write:
$$\Delta^2 - \omega_1^2 \omega_2^2 = A_{st}C_{st}\Big( A_{st}C_{st}(d_C-d_A)^2 + 4B^2_{st}d_Ad_C\Big)\,,$$
with $$d_A = \frac{ B_{st}}{B_{rt}} - \frac{A_{st}}{A_{rt}} \quad \hbox{and} \quad d_C = \frac{B_{st}}{B_{rt}} - \frac{C_{st}}{C_{rt}}\,.$$
As we have $A_{st}C_{st}>B^2_{st}$ and $(d_C-d_A)^2 > -4d_Ad_C$, multiplying these two inequalities give us a positive determinant. Second, we determine the smallest positive root by studying the sign of the product of the roots, which is equal to $\omega_1^2 (\omega_1^2 - 2\Delta + \omega_2^2)^{-1}$. With $2\Delta > \omega_1^2 + \omega_2^2$, the product is negative and the biggest root is the only one positive. If on the contrary $2\Delta < \omega_1^2 + \omega_2^2$, the roots have the same sign and the sign of $\Delta - \omega_1^2$ determines whether both roots are positive (in that case the smallest root is the solution) or both negative (in that case $\mu_{max}$ is infinite). The degenerated case with no quadratic term is obvious.
\end{proof}

It is also possible to consider the  quadratic form in dimension $p$ more than $2$. However, determining the boundary for $\mu$ is challenging to determine explicitly: this is the first mu value such that $\det(A-\mu B)=0$ with $A$ and $B$ two $p \times p$ matrices. 

\subsection{Changes in simple regression}

A direct application of the previous result is the problem of detecting a change point in simple regression. In this setting, we gather two-dimensional $(x_t,y_t)$ points over time. At a change location, the linear bound linking $x_t$ and $y_t$ changes. To be specific, this corresponds to a model with data points $(x_t,y_t)$ generated by a succession of simple linear models (choosing $x_t$ and then getting the response $y_t$ through the regression)
$$y_t = a_i x_t + b_i + \epsilon_t\,,\quad t=\tau_{i}+1,\dots,\tau_{i+1}\,,\quad i = 0,\dots,K\,,$$
with $t \mapsto a_t$ and  $t \mapsto b_t$ piecewise constant time series and $\epsilon_t \sim \mathcal{N}(0,\sigma^2)$ identically and independently distributed. The vector $(\tau_1,\ldots,\tau_K)$ of strictly increasing integers subdivides the time-series into $K+1$ consecutive segments with natural notations $\tau_{0}=0$ and $\tau_{K+1}=n$. Using the maximum likelihood approach, we see that this leads to considering the cost function $q_t^s(\theta_1, \theta_2) = Q_s + c(y_{st}; [\theta_1; \theta_2]) + \beta$ with
\begin{equation*}
\begin{split}
c(y_{st}; [\theta_1; \theta_2]) &= \sum_{j = s+1}^{t} (y_j - (\theta_1 x_{j} +\theta_2))^2 \,,\\
& = (t - s) \Big( \overline{x_{st}^2} \theta_1^2+ 2 \overline{x}_{st}\theta_1 \theta_2  + \theta_2^2 - 2\overline{(xy)_{st}}\theta_1 - 2 \overline{y}_{st} \theta_2  + \overline{y_{st}^2} \Big)\,.
\end{split}
\end{equation*}
Identifying the coefficients term by term, we introduce the notations:
$$q_t^s(\theta_1,\theta_2) = A_{st}\theta_1^2 +2B_{st}\theta_1\theta_2 + C_{st} \theta_2^2 + 2 D_{st} \theta_1 + 2E_{st}\theta_2 + F_{st}\,.$$
Writing as previously $A(\mu) = A_{st} + \mu(A_{st} - A_{rt}),\ldots$ we obtain the dual function in expression \eqref{dual_quadratic}.

\section{Examples of decision functions}
\label{app:exampleDual}

We consider the function $q_t^s$ constrained by the function $q_t^{r}$. We recall the shape of the decision function:
$$\mathbb{D}_{st}(\mathbf{x}) 
= -\D^*\Big(\overline{\mathbf{S}}_{st} + \sum_{r \ne s} x_{r} \Delta \overline{\mathbf{S}}_{rst}\Big) - \Big(\overline{Q}_{st} + \sum_{r \ne s} x_r \Delta \overline{Q}_{rst}\Big)\,.$$
Convex functions $\mathcal{D^*}$ as well as their associated domain $\Omega_{\mathbf{x}}$ are distribution-dependent. The domain is also data-dependent. We give their form on many examples from the exponential family in Table~\ref{dual_table}. Function $\mathcal{D^*}$ can be expressed through the log-partition function $A$ as $\mathcal{D^*}(x) = x (\nabla A)^{-1}(x) - A((\nabla A)^{-1}(x))$. We consider the single-constraint case with support $[0,x_{max})$ and notation $\D^*(\overline{\mathbf{S}}_{st} + x \Delta \overline{\mathbf{S}}_{rst}) = \D^*(\sigma_1 + x (\sigma_1-\sigma_2))$ in Table~\ref{dual_table}. Notice that we address here the case $r<s$. When $r>s$, the bound is always $x_{max} = 1$. Details for computing $\mathcal{D^*}$ are given in Table~\ref{dual_A}. More details on the properties of $\mathcal{D^*}$ can be found in Chapter 3 of \cite{wainwright2008graphical}.

\begin{table}[!ht]
\caption{Univariate case. Distribution, their corresponding $\mathcal{D}^{\star}$ function and maximum $x$ value in case $r < s$.\label{dual_table}}
\begin{tabular*}{\columnwidth}{@{\extracolsep\fill}llll@{\extracolsep\fill}}
Distribution &  Function $x \mapsto \mathcal{D^*}(x)$ & $x_{max}$ ($\sigma_1 < \sigma_2$) & $x_{max}$ ($\sigma_1 > \sigma_2$)  \\
 \midrule
Gauss & $\frac{1}{2}x^2$ & $+ \infty$  & $+ \infty$   \\
Exponential & $-\log x - 1$ & $-\frac{\sigma_1}{\sigma_1 - \sigma_2}$ & $+\infty$ \\
Poisson & $x(\log x - 1)$& $-\frac{\sigma_1}{\sigma_1 - \sigma_2}$  & $+\infty$ \\
Geometric & $(x-1)\log (x-1) - x\log x$ &  $-\frac{\sigma_1-1}{\sigma_1 - \sigma_2}$ &  $+\infty$ \\
Bernoulli/Binomial & $ x\log x + (1-x)\log(1-x)$ & $-\frac{\sigma_1}{\sigma_1 - \sigma_2}$ & $\frac{1-\sigma_1}{\sigma_1 - \sigma_2}$  \\
Negative Binomial & $ x\log x - (1+x)\log(1+x)$  &  $-\frac{\sigma_1}{\sigma_1 - \sigma_2}$ & $+\infty$ \\
Variance & $-\frac{1}{2}(\log x + 1)$& $-\frac{\sigma_1}{\sigma_1 - \sigma_2}$ & $+\infty$ \\
 \midrule
\end{tabular*}
In the Gaussian case, data is standardised: divided by the estimated standard deviation. In binomial and negative binomial cases, data is divided by the estimated number of trials and successes (respectively) so that the logarithmic values can be computed.
\end{table}

\begin{remark}
From Table \ref{dual_table} we can easily write down the decision functions in an independent multivariate setting. In that case, we sum over all dimensions the  function $\mathcal{D}^{\star}$ obtained for each univariate time series.
\end{remark}

\begin{table}[H]
\caption{Functions $A$ and $(\nabla A)^{-1}$ used to compute $\mathcal{D^*}(x) = x (\nabla A)^{-1}(x) - A((\nabla A)^{-1}(x))$ for some standard distributions in exponential family.\label{dual_A}}
\begin{tabular*}{\columnwidth}{@{\extracolsep\fill}lllll@{\extracolsep\fill}}
Distribution &  $y \mapsto A(y)$ & $x \mapsto (\nabla A)^{-1}(x)$ & $x \in \Omega$  \\
 \midrule
Gauss & $\frac{1}{2}y^2$ & $x$  & $\Omega=\R$ \\
Exponential & $-\log(-y)$ & $-\frac{1}{x}$   & $\Omega= (0, +\infty)$ \\
Poisson & $\exp(y)$ & $\log(x)$   & $\Omega= (0, +\infty)$ \\
Geometric & $-\log(e^{-y}-1)$ & $\log\Big( \frac{x-1}{x}\Big)$   & $\Omega=(1, +\infty)$  \\
Bernoulli/Binomial & $\log(1+e^y)$ &  $\log\Big( \frac{x}{1-x}\Big)$   & $\Omega=(0, 1)$\\
Negative Binomial & $-\log(1-e^y)$  & $\log\Big( \frac{x}{1+x}\Big)$    & $\Omega=(0, +\infty)$ \\
Variance & $-\frac{1}{2}\log(-2y)$  & $-\frac{1}{2x}$    & $\Omega=(0, +\infty)$ \\
 \midrule
\end{tabular*}
\end{table}

For the sake of completeness, we provide details about the limiting cases, defined by the limit values for $\overline{\mathbf{S}}_{st}$ or $\overline{\mathbf{S}}_{rs}$, with $\overline{\mathbf{S}}_{st} \ne \overline{\mathbf{S}}_{rs}$, on the boundary of $\Omega$. These limit values are denoted $\partial \Omega$ (equal to $0$ or $1$ with the given distributions in Table~\ref{dual_A}). If $\overline{\mathbf{S}}_{st}=\partial \Omega$, then $x_{max} = 0$ for all distributions (except for Gauss). If $\overline{\mathbf{S}}_{rs}=\partial \Omega$, then $x_{max} = +\infty$ except in the case Bernoulli/Binomial.

\begin{remark}
If $\overline{\mathbf{S}}_{st} = \overline{\mathbf{S}}_{rs}= \mathbf{S} \in \Omega \cup \partial \Omega$ for all $r$ indices used, the decision function is linear with maximum in $x = 0$ or $x = +\infty$, depending on the sign of the slope.
\end{remark}

\section{Additional simulations}
\label{app:add_simus}

\textcolor{violet}{To be determined. This section will be updated post-publication.}

\end{appendices}

\section*{Competing interests}
No competing interest is declared.



\bibliographystyle{plain}
\bibliography{reference}

\end{document}